%% file: qrb-lmcs-5.tex
\documentclass{LMCS}

\def\doi{8 (1:14) 2012}
\lmcsheading%
{\doi}
{1--33}
{}
{}
{Jan.~27, 2011}
{Feb.~29, 2012}
{}

\usepackage{amsmath,stmaryrd,amsfonts,amssymb,latexsym,url}
\usepackage{graphicx}
\usepackage[usenames]{color}
\usepackage[all]{xy}
\usepackage{ifpdf}
\usepackage{enumerate,hyperref}
\DeclareFontFamily{U}{matha}{\hyphenchar\font45} 
\DeclareFontShape{U}{matha}{m}{n}{
      <5> <6> <7> <8> <9> <10> gen * matha
      <10.95> matha10 <12> <14.4> <17.28> <20.74> <24.88> matha12
      }{}
\DeclareSymbolFont{matha}{U}{matha}{m}{n}
\DeclareFontSubstitution{U}{matha}{m}{n}
\DeclareMathSymbol{\cll}{3}{matha}{"CE} 

\newcommand\Min{\mathop{\text{Min}}}

\newcommand\limp{\Rightarrow}
\newcommand\nat{\mathbb{N}}

\newcommand\Smyth{\mathcal Q}
\newcommand\V{{\mathcal V}}
\newcommand\SV{\Smyth_\V}

\newcommand\img{\mathop{\mathrm{im}}}
\newcommand\Open{\mathcal O}

\newcommand\upc{\mathop{\uparrow}\nolimits}
\newcommand\dc{\mathop{\downarrow}\nolimits}
\newcommand\uuarrow{\rlap{$\uparrow$}\raise.5ex\hbox{$\uparrow$}}
\newcommand\ddarrow{\rlap{$\downarrow$}\raise.5ex\hbox{$\downarrow$}}
\newcommand\Fin{\mathop{\text{Fin}}}

\newcommand{\interior}[1]{int ({#1})} 

\newcommand\QRB{\mathbf{QRB}}
\newcommand\B{\mathbf{B}}
\newcommand\RB{\mathbf{RB}}
\newcommand\FS{\mathbf{FS}}
\newcommand{\identity}[1]{\mathrm{id}_{#1}}

\newcommand\qs{\varsigma}

\newcommand\Val{\mathbf V}
\newcommand{\real}{\mathbb{R}}
\newcommand{\creal}{\overline{\real^+_\sigma}}

\newcommand\dG{{\mathsf{d}}}
\newcommand\patch{{\mathsf{patch}}}

\begin{document}

\title[QRB-Domains]{$\QRB$-Domains and the Probabilistic Powerdomain\rsuper*}

\author[J.~Goubault-Larrecq]{Jean Goubault-Larrecq}	
\address{LSV, ENS Cachan, CNRS, INRIA, France}	
\email{goubault@lsv.ens-cachan.fr}  




\keywords{Domain theory, quasi-continuous domains, probabilistic powerdomain.}
\subjclass{D.3.1, F.1.2, F.3.2}

\titlecomment{{\lsuper*}An extended abstract already appeared in Proc. 25th {A}nnual {IEEE} {S}ymposium on {L}ogic in {C}omputer {S}cience ({LICS}'10).}



\begin{abstract}
  Is there any Cartesian-closed category of continuous domains that
  would be closed under Jones and Plotkin's probabilistic powerdomain
  construction?  This is a major open problem in the area of
  denotational semantics of probabilistic higher-order languages.  We
  relax the question, and look for quasi-continuous dcpos instead.
  We introduce a natural class of such quasi-continuous dcpos, the
  omega-QRB-domains.  We show that they form a category omega-QRB with
  pleasing properties: omega-QRB is closed under the probabilistic
  powerdomain functor, under finite products, under taking bilimits of
  expanding sequences, under retracts, and even under so-called
  quasi-retracts.  But\ldots{} omega-QRB is not Cartesian closed.  We
  conclude by showing that the QRB domains are just one half of an
  FS-domain, merely lacking control.
\end{abstract}

\maketitle


\section{Introduction}
\label{sec:intro}

\subsection{The Jung-Tix Problem}
A famous open problem in denotational semantics is whether the
probabilistic powerdomain $\Val_1 (X)$ of an $\FS$-domain $X$ is again
an $\FS$-domain \cite{JT:troublesome}, and similarly with
$\RB$-domains in lieu of $\FS$-domains.  $\Val_1 (X)$ (resp.\
$\Val_{\leq 1} (X)$) is the dcpo of all continuous probability (resp.,
subprobability) valuations over $X$: this construction was introduced
by Jones and Plotkin to give a denotational semantics to higher-order
probabilistic languages \cite{JP:proba}.

More generally, is there a category of nice enough dcpos that would be
Cartesian-closed and closed under $\Val_1$?  We call this the {\em
  Jung-Tix problem\/}.  By ``nice enough'', we mean nice enough to do
any serious mathematics with, e.g., to establish definability or full
abstraction results in extensional models of higher-order,
probabilistic languages.  It is traditional to equate ``nice enough''
with ``continuous'', and this is justified by the rich theory of
continuous domains \cite{GHKLMS:contlatt}.

However, {\em quasi-continuous\/} dcpos (see \cite{GLS:quasicont}, or
\cite[III-3]{GHKLMS:contlatt}) generalize continuous dcpos and are
almost as well-behaved.  We propose to widen the scope of the problem,
and ask for a category of quasi-continuous dcpos that would be closed
under $\Val_1$.  We show that, by mimicking the construction of
$\RB$-domains \cite{AJ:domains},
with some flavor of ``quasi'', we obtain a category $\omega\QRB$ of
so-called $\omega\QRB$-domains that not only has many desired, nice
mathematical properties (e.g., it is closed under taking bilimits of
expanding sequences, and every $\omega\QRB$-domain is stably compact),
but is also closed under $\Val_1$.

We failed to solve the Jung-Tix problem: $\omega\QRB$ is indeed not
Cartesian-closed.  In spite of this, we believe our contribution to
bring some progress towards settling the question, and at least to
understand the structure of $\Val_1 (X)$ better.  To appreciate this,
recall what is currently known about $\Val_1$.  There are two landmark
results: $\Val_1 (X)$ is a continuous dcpo as soon as $X$ is
(\cite{Edalat:int}, building on Jones
\cite{JP:proba}), and $\Val_1 (X)$ is stably compact (with its weak
topology) whenever $X$ is \cite{JT:troublesome,AMJK:scs:prob}.  Since
then, no significant progress has been made.  When it comes to solving
the Jung-Tix problem, we must realize that there is {\em little
  choice\/}: the only known Cartesian-closed categories of (pointed)
continuous dcpos that may suit our needs are $\RB$ and $\FS$
\cite{JT:troublesome}.  I.e., all other known Cartesian-closed
categories of continuous dcpos, e.g., bc-domains or L-domains, are
{\em not\/} closed under $\Val_1$.
Next, we must recognize that {\em little is known\/} about the
(sub)probabilistic powerdomain of an $\RB$ or $\FS$-domain.  In trying
to show that either $\RB$ or $\FS$ was closed under $\Val_1$, Jung and
Tix \cite{JT:troublesome} only managed to show that the
subprobabilistic powerdomain $\Val_{\leq 1} (X)$ of a {\em finite
  tree\/} $X$ was an $\RB$-domain, and that the subprobabilistic
powerdomain of a {\em reversed finite tree\/} was an $\FS$-domain.
This is still far from the goal.

\begin{figure}
  \centering
  \ifpdf
  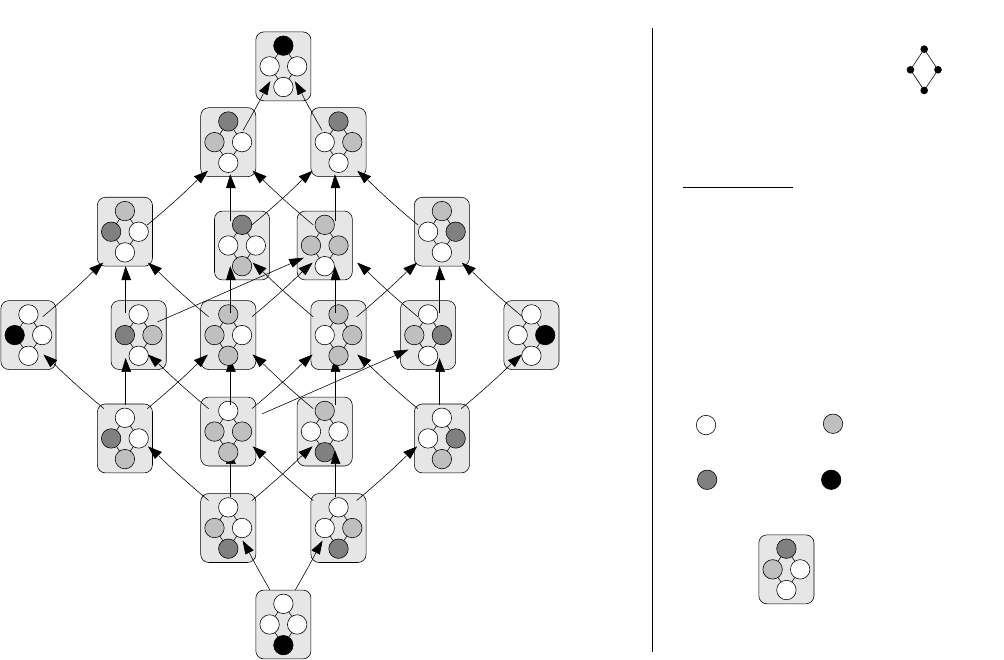
  \else
  \input{v3.pstex_t}
  \fi
  \caption{Part of the Hasse Diagram of $\Val_1 (X)$}
  \label{fig:v3}
\end{figure}

In fact, we do not know whether $\Val_1 (X)$ is an $\RB$-domain when
$X$ is even the simple poset $\{\bot, a, b, \top\}$ ($a$ and $b$
incomparable, $\bot \leq a, b \leq \top$, see Figure~\ref{fig:v3},
right)---but it is an $\FS$-domain.  For a more complex (arbitrarily
chosen) example, take $X$ to be the finite pointed poset of
Figure~\ref{fig:ex1}~$(i)$: then $\Val_1 (X)$ and $\Val_{\leq 1} (X)$
are continuous and stably compact, but not known to be $\RB$-domains
or $\FS$-domains (and they are much harder to visualize, too).

\begin{figure}
  \centering
  \ifpdf
  \input{ex1.pdftex_t}
  \else
  \input{ex1.pstex_t}
  \fi
  \caption{Poset Examples}
  \label{fig:ex1}
\end{figure}

No progress seems to have been made on the question since Jung and
Tix' 1998 attempt.  As part of our results, we show that for every
finite pointed poset $X$, e.g.\ Figure~\ref{fig:ex1}~$(i)$, $\Val_1
(X)$ is a continuous $\omega\QRB$-domain.
This is also one of the basic results that we then leverage to show
that $\Val_1 (X)$ is an $\omega\QRB$-domain for {\em any\/}
$\omega\QRB$-domain, in particular every $\RB$-domain, not just every
finite pointed poset, $X$.


One may obtain some intuition as to why this should be so, and at the
same time give an idea of what ($\omega$)$\QRB$-domains are.  Let $X$
be a finite pointed poset.  In attempting to show that $\Val_1 (X)$ is
an $\RB$-domain, we are led to study the so-called {\em deflations\/}
$f : \Val_1 (X) \to \Val_1 (X)$, i.e., the continuous maps $f$ with
{\em finite\/} range such that $f (\nu) \leq \nu$ for every continuous
probability valuation $\nu$ on $X$, and we must try to find deflations
$f$ such that $f (\nu)$ is as close as one desires to $\nu$.  All
natural definitions of $f$ fail to be continuous, and in fact to be
monotonic.  (E.g., Graham's construction \cite{Graham:rb:V} is not
monotonic, see Jung and Tix.)  Looking for maps $f$ such that $f
(\nu)$ is instead a finite, non-empty {\em set\/} of valuations below
$\nu$ shows more promise---the monotonicity requirements are slightly
more relaxed.  Such a set-valued function is what we call a {\em
  quasi-deflation\/} below.  For example, one may think of fixing $N
\geq 1$ ($N=3$ in Figure~\ref{fig:v3}), and mapping $\nu$ to the
collection of all valuations $\nu'$ below $\nu$ such that the measure
of any subset is a multiple of $1/N$, keeping only those $\nu'$ that
are maximal.  (Pick them from the left of Figure~\ref{fig:v3}, in our
example.)  This still does not provide anything monotonic, but we
managed to show that one can indeed approximate every element $\nu$ of
$\Val_1 (X)$, continuously in $\nu$, using quasi-deflations.  The
proof is non-trivial, and rests on deep properties relating
$\QRB$-domains and \emph{quasi-retractions}, all notions that we
define and study.


\subsection{Outline}

We introduce most of the required notions in Section~2.  Since we
shall only start studying the probabilistic powerdomain in
Section~\ref{sec:qretr:V}, we shall refrain from defining valuations,
probabilities, and related concepts until then.

We introduce $\QRB$-domains in Section~\ref{sec:qrb}.  They are
defined just as $\RB$-domains are, only with a flavor of ``quasi'',
i.e., replacing approximating elements by approximating {\em sets\/}
of elements.  We establish their main properties there, in particular
that they are quasi-continuous, stably compact, and Lawson-compact.
Much as $\RB$-domains are also characterized as the retracts of
bifinite domains, we show that, up to a few details, the
$\QRB$-domains are the quasi-retracts of bifinite domains in
Section~\ref{sec:qretr}.  This allows us to parenthesize $\QRB$ as
quasi-(retract of bifinite domain) or as (quasi-retract) of bifinite
domain.  Quasi-retractions are an essential concept in the study of
$\QRB$-domains, and we introduce them here, as well as the related
notion of quasi-projections---images by proper maps.

We also show that the category of countably-based $\QRB$-domains is
closed under finite products (easy) and taking bilimits of expanding
sequences (hard, but similar to the case of $\RB$-domains) in
Section~\ref{sec:bilim}.

The core of the paper is Section~\ref{sec:qretr:V}, where we show that
the category $\omega\QRB$ of countably-based $\QRB$-domains is closed
under the probabilistic powerdomain construction.  This capitalizes on
all previous sections, and will follow from a variant of Jung and Tix'
result that $\Val_1 (X)$ is an $\RB$-domain whenever $X$ is a finite
tree, and applying suitable quasi-projections and bilimits.
The key result will then be Theorem~\ref{thm:qretr:V}, which shows
that for any quasi-projection $Y$ of a stably compact space $X$,
$\Val_1 (Y)$ is again a quasi-projection of $\Val_1 (X)$, again up to
a few details.

We conclude in Section~\ref{sec:conc}.

\subsection{Other Related Work}

Instead of solving the Jung-Tix problem, one may try to circumvent it.
One of the most successful such attempts led to the discovery of {\em
  qcb-spaces\/} \cite{BSS:qcb} and to compactly generated
countably-based monotone convergence spaces \cite{BSS:cgdom}, as
Cartesian-closed categories of topological spaces where a reasonable
amount of semantics can be done.  This provides exciting new
perspectives.  The category of qcb-spaces accommodates two
probabilistic powerdomains \cite{BS:twoval}.  The observationally
induced one is essentially $\Val_1 (X)$ (with the weak topology), but
differs from the one obtained as a free algebra.

\section{Preliminaries}
\label{sec:prelim}

We refer to \cite{AJ:domains,GHKLMS:contlatt,Mislove:topo:CS} for
background material.
A {\em poset\/} $X$ is a set with a partial ordering $\leq$.  Let $\dc
A$ be the downward closure $\{x \in X \mid \exists y \in A \cdot x
\leq y\}$; we write $\dc x$ for $\dc \{x\}$, when $x \in X$.  The
upward closures $\upc A$, $\upc x$ are defined similarly.  When $x
\leq y$, $x$ is {\em below\/} $y$ and $y$ is {\em above\/} $x$.  $X$
is {\em pointed\/} iff it has a least element $\bot$.  A {\em dcpo\/}
is a poset $X$ where every directed family ${(x_i)}_{i \in I}$ has a
least upper bound $\sup_{i \in I} x_i$; directedness means that $I\neq
\emptyset$ and for every $i, i' \in I$, there is an $i'' \in I$ such
that $x_i, x_{i'} \leq x_{i''}$.

Every poset, and more generally each preordered set $X$ comes with a
topology, whose opens $U$ are the upward closed subsets such that, for
every directed family ${(x_i)}_{i \in I}$ that has a least upper bound
in $U$, $x_i \in U$ for some $i \in I$.  This is the {\em Scott
  topology\/}.  When we see a poset or dcpo $X$ as a topological
space, we will implicitly assume the latter, unless marked otherwise.

There is a deep connection between order and topology.  Given %
any topological space $X$, its {\em specialization preorder\/}
$\leq$ is defined by $x \leq y$ iff every open containing $x$ also
contains $y$.  $X$ is $T_0$ iff $\leq$ is an ordering, i.e., $x \leq
y$ and $y \leq x$ imply $x=y$.  The specialization preorder of a
dcpo $X$ (with ordering $\leq$, and equipped with its Scott topology),
is the original ordering $\leq$.

A subset $A$ of a topological space $X$ is {\em saturated\/} iff it is
the intersection of all opens $U$ containing $A$.  Equivalently, $A$
is upward closed in the specialization preorder \cite[Remark
after Definition~4.34]{Mislove:topo:CS}.  So we can, and shall often
prove inclusions $A \subseteq B$ where $B$ is upward closed by showing
that every open $U$ containing $B$ also contains $A$.

A map $f : X \to Y$ between topological spaces is {\em continuous\/}
iff $f^{-1} (V)$ is open for every open subset $V$ of $Y$.  Every
continuous map is monotonic with respect to the underlying
specialization preorders.  When $X$ and $Y$ are preordered sets, it is
equivalent to require $f$ to be {\em Scott-continuous\/}, i.e., to be
monotonic and to preserve existing directed least upper bounds.  A
{\em homeomorphism\/} is a bijective continuous map whose inverse is
also continuous.

Given a set $X$, and a family $\mathcal B$ of subsets of $X$, there is
a smallest topology containing $\mathcal B$: then $\mathcal B$ is a
{\em subbase\/} of the topology, and its elements are the {\em
  subbasic opens\/}.  To show that $f : X \to Y$ is continuous, it is
enough to show that the inverse image of every subbasic open of $Y$ is
open in $X$.  A subbase $\mathcal B$ is a \emph{base} if and only if
every open is a union of elements of $\mathcal B$.  This is the case,
for example, if $\mathcal B$ is closed under finite intersections.

The {\em interior\/} $\interior A$ of a subset $A$ of a topological
space $X$ is the largest open contained in $A$.  $A$ is a
\emph{neighborhood} of $x$ if and only if $x \in \interior A$, and a
neighborhood of a subset $B$ if and only if $B \subseteq \interior A$.
A subset $Q$ of a topological space $X$ is {\em compact\/} iff one can
extract a finite subcover from every open cover of $Q$.  The important
ones are the {\em saturated\/} compacts.
$X$ is {\em locally compact\/} iff for each open $U$ and each $x \in
U$, there is a compact saturated subset $Q$ such that $x \in \interior
Q$ and $Q \subseteq U$.  In any locally compact space, we have the
following interpolation property: whenever $Q$ is a compact subset of
some open $U$, then there is a compact saturated subset $Q_1$ such
that $Q \subseteq \interior {Q_1} \subseteq Q_1 \subseteq U$.

$X$ is {\em sober\/} iff every irreducible closed subset is the
closure of a unique point; in the presence of local compactness (and
when $X$ is $T_0$),
it is equivalent to require that $X$ be {\em well-filtered\/}
\cite[Theorem~II-1.21]{GHKLMS:contlatt}, i.e., to require that, for
every open $U$, for every filtered family ${(Q_i)}_{i \in I}$ of
saturated compacts such that $\bigcap_{i \in I}^\downarrow Q_i
\subseteq U$, $Q_i \subseteq U$ for some $i \in I$ already.  We say
that the family is {\em filtered\/} iff it is directed in the
$\supseteq$ ordering, and make it explicit by using
$\downarrow$ as superscript.  (Symmetrically, we write $\bigcup^\uparrow$ for
directed unions.)

Given a topological space $X$, let $\Smyth (X)$ be the collection of
all non-empty compact saturated subsets $Q$ of $X$.  There are two
prominent topologies one can put on $\Smyth (X)$.  The {\em upper
  Vietoris\/} topology has a subbase of opens of the form $\Box U$,
$U$ open in $X$, where we write $\Box U$ for the collection of compact
saturated subsets $Q'$ included in $U$.  We shall write $\SV (X)$ for
the space $\Smyth (X)$ with the upper Vietoris topology, and call it
the \emph{Smyth powerspace}.  The specialization ordering of $\SV (X)$
is reverse inclusion $\supseteq$.  On the other hand, we shall reserve
the notation $\Smyth_\sigma (X)$ for the {\em Smyth powerdomain\/} of
$X$, which is equipped with the Scott topology of $\supseteq$ instead.
When $X$ is well-filtered, $\Smyth (X)$ is a dcpo, with least upper
bounds of directed families computed as filtered intersections, and
$\Box U$ is Scott-open for every open subset $U$ of $X$, i.e., the
Scott topology is finer than the upper Vietoris topology.  When $X$ is
locally compact and sober (in particular, well-filtered), the two
topologies coincide, and $\Smyth_\sigma (X)$ is then a continuous dcpo
(see below), where $Q \ll Q'$ iff $Q' \subseteq \interior Q$
\cite[Proposition~I-1.24.2]{GHKLMS:contlatt}.  Schalk
\cite[Chapter~7]{Schalk:PhD} provides a deep study of these spaces.

For every finite subset $E$ of a topological space $X$, $E$ is compact
and $\upc E$ is saturated compact in $X$.  We call {\em finitary
  compact\/} those subsets of the form $\upc E$ with $E$ finite, and
let $\Fin (X)$ be the subset of $\Smyth (X)$ consisting of the
non-empty finitary compacts.  $\Fin (X)$ can be topologized with the
subspace topology from $\SV (X)$, in which case we obtain a space we
write $\Fin_\V (X)$, or with the Scott topology of reverse inclusion
$\supseteq$, yielding a space that we write $\Fin_\sigma (X)$.

Given any poset $X$, any finite subset $E$ of $X$, and any element $x$
of $X$, we write $E \leq x$ iff $x \in \upc E$, i.e., iff there is a
$y \in E$ such that $y\leq x$.  Given any upward closed subset $U$ of
$X$, we shall write $U \cll x$ iff for every directed family
${(x_i)}_{i \in I}$ that has a least upper bound above $x$, then $x_i$
is in $U$ for some $i \in I$.  Then a finite set $E$ {\em
  approximates\/} $x$ iff $\upc E \cll x$.  This is usually written $E
\ll x$ in the literature.
We shall also write $y \ll x$, when $y \in X$, as shorthand for $\upc
y \cll x$.  This is the more familiar way-below relation, and a poset
is {\em continuous\/} if and only if the set $\ddarrow x$ of all
elements $y$ such that $y \ll x$ is directed and has $x$ as least
upper bound.  One should be aware that $\upc E \cll x$ means that the
elements of $E$ approximate $x$ {\em collectively\/}, while none in
particular may approximate $x$ individually.  E.g., in the poset
$\mathcal N_2$ (Figure~\ref{fig:ex1}~$(ii)$), the sets $\{(0, m), (1,
n)\}$ approximate $\omega$, for all $m, n \in \nat$; but $(0, m) \not
\ll \omega$, $(1, n) \not\ll \omega$.

It may be helpful to realize that $\Fin (X)$ can also be presented in
the following equivalent way.  Given two finitary compacts $\upc E$
and $\upc E'$, $\upc E \supseteq \upc E'$ if and only if for every $x'
\in E'$, there is an $x \in E$ such that $x \leq x'$, and then we
write $E \leq^\sharp E'$: this is the so-called {\em Smyth
  preorder\/}.  Then we can equate the finitary compacts $\upc
E$ with the equivalence classes of finite subsets $E$, up to the
equivalence $\equiv$ defined by $E \equiv E'$ iff $\upc E = \upc E'$
iff $E \leq^\sharp E'$ and $E' \leq^\sharp E$, declare that $\Fin (X)$
is the set of equivalence classes of non-empty finite sets, ordered by
$\leq^\sharp$.  But the approach based on finitary compacts is
mathematically smoother.

Among the Cartesian-closed categories of continuous dcpos, one finds
in particular the $\B$-domains (a.k.a., the bifinite domains), the
$\RB$-domains, i.e., the retracts of bifinite domains
\cite[Section~4.2.1]{AJ:domains}, and the $\FS$-domains
\cite[Section~4.2.2]{AJ:domains}\cite[Section~II.2]{GHKLMS:contlatt}.
There are several equivalent definitions of the first two.

For our purposes, an {\em $\RB$-domain\/} is a pointed dcpo $X$ with a
directed family ${(f_i)}_{i \in I}$ of deflations such that $\sup_{i
  \in I} f_i = \identity X$ \cite[Exercise~4.3.11(9)]{AJ:domains}.  A
{\em deflation\/} $f$ on $X$ is a continuous map from $X$ to $X$ such
that $f (x) \leq x$ for every $x \in X$, and that has {\em finite
  image\/}.  We order deflations, as well as all maps with codomain a
poset, pointwise: i.e., $f \leq g$ iff $f (x) \leq g (x)$ for every $x
\in X$; knowing this, directed families and least upper bounds of
deflations make sense.  Every $\RB$-domain is a continuous dcpo, and
$f_i (x) \ll x$ for every $i \in I$ and every $x \in X$.

A {\em $\B$-domain\/} is defined similarly, except the deflations
$f_i$ are now required to be {\em idempotent\/}, i.e., $f_i \circ f_i
= f_i$ \cite[Theorem~4.2.6]{AJ:domains}.  This implies that $f_i (x)
\ll f_i (x)$, i.e., that all the elements $f_i (x)$ are finite; hence
all bifinite domains are also algebraic.  Every bifinite domain is an
$\RB$-domain.  Conversely, the $\RB$-domains are exactly the retracts
of bifinite domains: we shall define what this means and extend this
in Section~\ref{sec:qretr}.

An {\em $\FS$-domain\/} is defined similarly again, except the
functions $f_i$ are no longer deflations, but continuous functions
that are {\em finitely separated from $\identity X$\/}.  That is, we
now require that there is a finite set $M_i$ such that for every $x
\in X$, there is an $m \in M_i$ such that $f_i (x) \leq m \leq x$.  We
say that $M_i$ is {\em finitely separating\/} for $f_i$ on $X$.

Every deflation is finitely separated from $\identity X$: take $M_i$
to be the image of $f_i$.  The converse fails.  E.g., for every
$\epsilon > 0$, the function $x \mapsto \max (x - \epsilon, 0)$ is
finitely separated from the identity on $[0,1]$, but is not a
deflation \cite[Section~3.2]{JT:troublesome}.  Every $\RB$-domain is
an $\FS$-domain.  The converse is not known.

A {\em quasi-continuous dcpo\/} $X$ (see \cite{GLS:quasicont} or
\cite[Definition~III-3.2]{GHKLMS:contlatt}) is a dcpo such that, for
every $x \in X$, the collection of all $\upc E \in \Fin (X)$ that
approximate $x$ ($\upc E \cll x$) is directed (w.r.t.\ $\supseteq$)
and their least upper bound in $\Smyth (X)$ is $\upc x$, i.e.,
$\bigcap_{\substack{\upc E \in \Fin (X)\\ \upc E \cll x}} \upc E =
\upc x$.  %
The theory of quasi-continuous dcpos is less well explored than that
of {\em continuous dcpos\/}, but quasi-continuous dcpos retain many of
the properties of the latter.  (Every continuous dcpo is
quasi-continuous, but not conversely.  A counterexample is given by
$\mathcal N_2$, see Figure~\ref{fig:ex1}~$(ii)$.)
Every quasi-continuous dcpo $X$ is locally compact and sober in its
Scott topology \cite[III-3.7]{GHKLMS:contlatt}.  In a quasi-continuous
dcpo $X$, for every $\upc E \in \Fin (X)$, the set $\uuarrow E$
defined as $\{x \in X \mid \upc E \cll x\}$, is open, and equals the
interior $\interior {\upc E}$ \cite[III-3.6(ii)]{GHKLMS:contlatt};
every open $U$ is the union of all the subsets $\uuarrow E$ with $\upc
E \in \Fin (X)$ contained in $U$ \cite[III-5.6]{GHKLMS:contlatt}; and
for every compact saturated subset $Q$ and every open subset $U$
containing $Q$, there is a finitary compact subset $\upc E$ of $X$
such that $Q \subseteq \uuarrow E$ and $\upc E \subseteq U$
\cite[III-5.7]{GHKLMS:contlatt}.
In particular, $Q = \bigcap_{\upc E \in \Fin (X),\: Q \subseteq \uuarrow
  E}^\downarrow \upc E$.  Another consequence is {\em
  interpolation\/}: writing $\upc E \cll \upc E'$ for $\upc E \cll y$
for every $y$ in $E'$ (equivalently, $\upc E' \subseteq \uuarrow E$),
if $\upc E \cll x$ in a quasi-continuous dcpo $X$, for some $\upc E
\in \Fin (X)$, and $x \in X$, then $\upc E \cll \upc E' \cll x$ for
some $\upc E' \in \Fin (X)$.

If $X$ is a quasi-continuous dcpo, the formula $Q = \bigcap_{\upc E
  \in \Fin (X),\: Q \subseteq \uuarrow E}^\downarrow \upc E$, valid
for every $Q \in \Smyth (X)$, shows that $Q$ is the filtered
intersection of its finitary compact neighborhoods, equivalently the
directed least upper bound of those non-empty finitary compacts $\upc
E$ ($E \in \Fin (X)$) that are way-below $Q$.  In other words, the
finitary compacts form a {\em basis\/} of $\Smyth (X)$.


\section{$\QRB$-Domains}
\label{sec:qrb}

We model $\QRB$-domains after $\RB$-domains, replacing single
approximating elements $f_i (x)$, where $f_i$ is a deflation, by
finite subsets, as in quasi-continuous dcpos.

\begin{defi}[$\QRB$-Domain]
  \label{defn:qrb}
  A {\em quasi-deflation\/} on a poset $X$ is a continuous map
  $\varphi : X \to \Fin_\sigma (X)$ such that $x \in \varphi (x)$ for
  every $x \in X$, and $\img \varphi = \{\varphi (x) \mid x \in X\}$
  is finite.

  A {\em $\QRB$-domain\/} is a pointed dcpo $X$ with a {\em generating
    family of quasi-deflations\/}, i.e., a directed family of
  quasi-deflations ${(\varphi_i)}_{i \in I}$ with $\upc x = \bigcap_{i
    \in I}^\downarrow {\varphi_i (x)}$ for each $x\in X$.
\end{defi}
We order quasi-deflations pointwise, i.e., $\varphi \leq \psi$ iff
$\varphi (x) \supseteq \psi (x)$ for every $x \in X$.  Above, we write
$\bigcap^\downarrow$ instead of $\bigcap$ to stress the fact that the
family ${({\varphi_i (x)})}_{i \in I}$ of which we are taking the
intersection is {\em filtered\/}, i.e., for any two $i, i' \in I$,
there is an $i'' \in I$ such that ${\varphi_{i''} (x)}$ is contained
in both ${\varphi_{i} (x)}$ and ${\varphi_{i'} (x)}$.  It is
equivalent to say that ${(\varphi_i (x))}_{i \in I}$ is directed in
the $\supseteq$ ordering of $\Fin (X)$.

One can see the finitary compacts $\varphi_i (x)$ as being smaller and
smaller upward closed sets containing $x$.  The intersection
$\bigcap_{i \in I}^\downarrow {\varphi_i (x)}$ is then just the least
upper bound of ${({\varphi_i (x)})}_{i \in I}$ in the Smyth
powerdomain $\Smyth (X)$.  On the other hand, $X$ embeds into $\SV
(X)$ by equating $x \in X$ with $\upc x \in \Smyth (X)$.  Modulo this
identification, the condition $\upc x = \bigcap_{i \in I}^\downarrow
{\varphi_i (x)}$ requires that $x$ is the least upper bound of ${(
  {\varphi_i (x)})}_{i \in I}$ in $\Smyth (X)$.

That $\varphi$ is continuous means that $\varphi$ is monotonic ($x
\leq y$ implies $\varphi (x) \supseteq \varphi (y)$), and that for
every directed family ${(x_j)}_{j \in J}$ of elements of $X$, $\varphi
(\sup_{j \in J} x_j)$ is equal to $\bigcap_{i \in I}^\downarrow
\varphi (x_j)$---this implies that the latter is finitary compact, in
particular.

\begin{prop}
  \label{prop:FS:QRB}
  Every $\RB$-domain 
  is a $\QRB$-domain.
\end{prop}
\begin{proof}
  Given a directed family of deflations ${(f_i (x))}_{i \in I}$,
  define $\varphi_i (x)$ as $\upc f_i (x)$.  If $f_i \leq f_j$, then
  $\varphi_i (x) \supseteq \varphi_j (x)$ for every $x \in X$, so
  ${(\varphi_i)}_{i \in I}$ is directed.  Also, $\bigcap_{i \in
    I}^\downarrow {\varphi_i (x)}$ is the set of upper bounds of
  ${(f_i (x))}_{i \in I}$, of which the least is $x$.  So this set is
  exactly $\upc x$.
\end{proof}
We shall improve on this in Theorem~\ref{thm:ctrlQRB}, which implies
that not only the $\RB$-domains, but all $\FS$-domains, are
$\QRB$-domains.

For any deflation $f$, and more generally whenever $f$ is finitely
separated from the identity, $f (x)$ is way-below $x$
\cite[Lemma~II-2.16]{GHKLMS:contlatt}.  Similarly:
\begin{lem}
  \label{lemma:qdefl:ll}
  Let $X$ be a poset, and $\varphi$ be a quasi-deflation on $X$.  For
  every $x \in X$, $\varphi (x) \cll x$.
\end{lem}
\begin{proof}
  Let ${(x_j)}_{j \in J}$ be a directed family having a least upper
  bound above $x$.  Since $\varphi$ is continuous, $\bigcap_{j \in
    J}^\downarrow {\varphi (x_j)} \subseteq \varphi (x)$.  But since
  $\img \varphi$ is finite, there are only finitely many sets $\varphi
  (x_j)$, $j \in J$.  So ${\varphi (x_j)} \subseteq \varphi (x)$ for
  some $j \in J$.  Since $x_j \in {\varphi (x_j)}$, $x_j \in \varphi
  (x)$.
\end{proof}

\begin{cor}
  \label{corl:qrb:qcont}
  Every $\QRB$-domain is quasi-continuous.
\end{cor}
In general, $\QRB$-domains are not continuous.  E.g., $\mathcal N_2$
(Figure~\ref{fig:ex1}~$(ii)$) is not continuous.  However, $\mathcal
N_2$ is a $\QRB$-domain: for all $i,j \in \nat$, take $\varphi_{ij}
(\omega) = \upc \{(0, i), (1, j)\}$, $\varphi_{ij} (0,m) = \upc \{(0,
\min (m,i)), (1, j)\}$, $\varphi_{ij} (1,m) = \upc \{(0,i), (1, \min
(m,j))\}$.  Then
${(\varphi_{ij})}_{i, j \in \nat}$ is the desired directed family of
quasi-deflations.

$\RB$-domains, and more generally $\FS$-domains, are not just
continuous domains, they are {\em stably compact\/}, i.e., locally
compact, sober, compact and coherent (see, e.g.,
\cite[Theorem~4.2.18]{AJ:domains}).  We say that a topological space
is {\em coherent\/} iff the intersection of any two compact saturated
subsets is compact (and saturated).  In a stably compact space, the
intersection of any family of compact saturated subsets is compact.
We show that $\QRB$-domains are stably compact as well.

Since every quasi-continuous dcpo is locally compact and sober
\cite[Proposition~III-3.7]{GHKLMS:contlatt}, and also compact since
pointed, only coherence remains to be shown.  For this, we need the
following consequence of Rudin's Lemma, a finitary form of
well-filteredness:
\begin{prop}[\protect{\cite[Corollary~III-3.4]{GHKLMS:contlatt}}]
  \label{prop:heckmann}
  Let $X$ be a dcpo, ${(\upc E_i)}_{i \in I}$ be a directed family in
  $\Fin (X)$.  For every open subset $U$ of $X$, if $\bigcap_{i \in
    I}^\downarrow \upc E_i \subseteq U$, then $\upc E_i \subseteq U$
  for some $i \in I$.
\end{prop}

It follows that, if $X$ is a dcpo, then the Scott topology on $\Fin
(X)$ is finer than the upper Vietoris topology.  Indeed, this reduces
to showing that $\Fin (X) \cap \Box U$ is Scott-open in $\Fin (X)$,
for every open subset $U$ of $X$.  And this is
Proposition~\ref{prop:heckmann}, plus the easily checked fact that
$\Box U$ is upward closed in $\supseteq$.
\begin{cor} 
  \label{corl:BoxU}
  Let $X$ be a dcpo.  The Scott topology is finer than the upper
  Vietoris topology on $\Fin (X)$, and coincides with it whenever $X$
  is quasi-continuous.
\end{cor}
\begin{proof}
  It remains to show that, if $X$ is a quasi-continuous dcpo, then
  every Scott-open $\mathcal U$ of $\Fin (X)$ is open in the upper
  Vietoris topology.  Let $\upc E \in \Fin (X)$ be in $\mathcal U$.
  It suffices to show that there is an open subset $U$ of $X$ such
  that $\upc E \in \Box U \subseteq \mathcal U$.  Write $E = \{x_1,
  \ldots, x_n\}$.  For each $i$, $1\leq i\leq n$, $\upc x_i$ is the
  filtered intersection of all finitary compacts $\upc E_i \cll x_i$.
  The unions $\upc E_1 \cup \ldots \cup \upc E_n = \upc (E_1 \cup
  \ldots \cup E_n)$, with $\upc E_1 \cll x_1$, \ldots, $\upc E_n \cll
  x_n$, also form a directed family in $\Fin (X)$, and their
  intersection is $\upc E$.  So there are finitary compacts $\upc E_1
  \cll x_1$, \ldots, $\upc E_n \cll x_n$ whose union is in $\mathcal
  U$.  Since $\upc E_i \cll x_i$ for each $i$, each $x_i$ is in the
  Scott-open $\uuarrow E_i$, so $\upc E \in \Box U$ with $U = \uuarrow
  E_1 \cup \ldots \cup \uuarrow E_n$.  Moreover, $\Box U \subseteq
  \mathcal U$: for each $\upc E' \in \Box U$, $\upc E'$ is included in
  $U \subseteq \upc E_1 \cup \ldots \upc E_n$; since $\upc E_1 \cup
  \ldots \upc E_n$ is in $\mathcal U$ and $\mathcal U$ is
  upward-closed in $\supseteq$, $\upc E'$ is in $\mathcal U$.  
\end{proof}

Schalk \cite[Chapter~7]{Schalk:PhD} proved that $\SV$ defines a monad
on the category of topology spaces (see \cite{Mog91} for an
introduction to monads and their importance in programming language
semantics).  This means first that there is a {\em unit map\/}
$\eta_X$---here, $\eta_X$ maps $x \in X$ to $\upc x \in \SV (X)$, and
this is continuous because $\eta_X^{-1} (\Box U) = U$.  That $\SV$ is
a monad also means that every continuous map $h : X \to \SV (Y)$ has
an {\em extension\/} $h^\dagger : \SV (X) \to \SV (Y)$, i.e.,
$h^\dagger$ is continuous and $h^\dagger \circ \eta_X = h$.  This is
defined by $h^\dagger (Q) = \bigcup_{x \in Q} h (x)$ in our case.
Again, $h^\dagger$ is continuous, because ${h^\dagger}^{-1} (\Box U) =
\Box h^{-1} (\Box U)$.  And the {\em monad laws\/} are satisfied:
$\eta_X^\dagger = \identity {\SV (X)}$, $h^\dagger \circ \eta_X = h$,
and $(g^\dagger \circ h)^\dagger = g^\dagger \circ h^\dagger$.  One
should be careful here: $\SV$ is a monad, but $\Smyth_\sigma$ is not a
monad, except on specific subcategories, e.g., sober locally compact
spaces $X$, where $\Smyth_\sigma (X) = \SV (X)$ anyway.

The continuity claims in the following lemma are then obvious.
\begin{lem}
  \label{lemma:dagger:new}
  Let $X$, $Y$ be topological spaces.  Given any continuous map $\psi
  : X \to \Fin_\V (Y)$, its extension $\psi^\dagger$ restricts to a
  continuous map $\psi^\dagger : \Fin_\V (X) \to \Fin_\V (Y)$.  If
  $\img \psi$ is finite, then $\psi^\dagger$ maps $\SV (X)$
  continuously into $\Fin_\V (Y)$.
\end{lem}
\begin{proof}
  In each case, one only needs to show that $\psi^\dagger$ maps
  relevant compacts to finitary compacts.  In the first case, for
  every finitary compact $\upc E \in \Fin (X)$, $\psi^\dagger (\upc E)
  = \bigcup_{x \in \upc E} \psi (x) = \bigcup_{x \in E} \psi (x)$
  (because $\psi$ is monotonic), and this is finitary compact.  In the
  second case, $\psi^\dagger (Q) = \bigcup_{x \in Q} \psi (x)$ is a
  finite union of finitary compacts since $\img \psi$ is finite.
\end{proof}
One would also like $\psi^\dagger$ to be continuous from
$\Smyth_\sigma (X)$ to $\Fin_\sigma (Y)$, in the face of the
importance of the Scott topology.  This is a consequence of the above
when $X$ is sober and locally compact, and $Y$ is a quasi-continuous
dcpo, since $\Smyth_\sigma (X) = \SV (X)$ and $\Fin_\sigma (Y) =
\Fin_\V (Y)$ in this case.  However, one can also prove this in a more
general setting, using the following observation.  For each
topological space $Z$, write $Z_\sigma$ for $Z$ with the Scott
topology of its specialization preorder.  For short, we shall
call {\em quasi monotone convergence space\/} any space $Z$ such that
the (Scott) topology on $Z_\sigma$ is finer than that of $Z$, i.e.,
such that every open subset of $Z$ is open is Scott-open.  This is a
slight relaxation of the notion of {\em monotone convergence space\/},
i.e., of a quasi monotone convergence space that is a dcpo in its
specialization preorder
\cite[Definition~II-3.12]{GHKLMS:contlatt}.  E.g., every sober space
is a monotone convergence space, and in particular a quasi monotone
convergence space.
\begin{lem}
  \label{lemma:fin:cont}
  Let $Z$ be a quasi monotone convergence space and $Z'$ be a
  topological space.  Every continuous map $f : Z \to Z'$ is
  Scott-continuous, i.e., continuous from $Z_\sigma$ to $Z'_\sigma$.
\end{lem}
\begin{proof}
  Since $f$ is continuous, it is monotonic with respect to the
  underlying specialization preorders.  Let ${(z_i)}_{i \in I}$ be any
  directed family of elements of $Z$, with least upper bound $z$.
  Certainly $f (z)$ is an upper bound of ${(f (z_i))}_{i \in I}$.  Let
  us show that, for any other upper bound $z'$, $f (z) \leq z'$.  It is
  enough to show that every open neighborhood $V$ of $f (z)$ contains
  $z'$.  Since $f (z) \in V$, $z$ is in the open subset $f^{-1} (V)$,
  which is Scott-open by assumption, so $z_i \in f^{-1} (V)$ for some
  $i \in I$.  It follows that $f (z_i)$ is in $V$, hence also $z'$
  since $V$ is upward closed.
\end{proof}

When $X$ is sober and locally compact, the topology of $\Smyth_\sigma
(X)$ coincides with that of $\SV (X)$.  In particular, $Z = \SV (X)$
is a quasi-monotone convergence space.  Taking $Z' = \SV (Y)$ in
Lemma~\ref{lemma:fin:cont}, one obtains the following corollary.
\begin{cor}
  \label{corl:fin:cont:Smyth}
  Let $X$ be a sober, locally compact space, and $Y$ be a topological
  space.  Every continuous map from $\SV (X)$ to $\SV (Y)$ is also
  Scott-continuous from $\Smyth (X)$ to $\Smyth (Y)$.
\end{cor}
Similarly, with $Z' = \Fin_\V (Y)$:
\begin{cor}
  \label{corl:fin:cont:Fin}
  Let $Y$ be a topological space, $Z$ be a quasi monotone convergence
  space.  Every continuous map from $Z$ to $\Fin_\V (Y)$ is
  Scott-continuous, i.e., continuous from $Z_\sigma$ to $\Fin_\sigma
  (Y)$.
\end{cor}

\begin{lem}
  \label{lemma:psi*}
  Let $X$ be a $\QRB$-domain, and ${(\varphi_i)}_{i \in I}$ a
  generating family of quasi-deflations.  For every open subset $U$ of
  $X$, $\bigcup_{i \in I}^\uparrow \varphi_i^{-1} (\Box U) = U$.
\end{lem}
\begin{proof}
  The union is directed, since $\varphi_i^{-1} (\Box U) \subseteq
  \varphi_{i'}^{-1} (\Box U)$ whenever $\varphi_i$ is pointwise below
  $\varphi_{i'}$, i.e., when $\varphi_i (x) \supseteq \varphi_{i'}
  (x)$ for all $x \in X$.  For every $i \in I$, $\varphi_i^{-1} (\Box
  U) \subseteq U$: every element $x$ of $\varphi_i^{-1} (\Box U)$ is
  indeed such that $x \in \varphi_i (x) \subseteq U$.  Conversely, we
  claim that every element $x$ of $U$ is in $\varphi_i^{-1} (\Box U)$
  for some $i \in I$.  Indeed, $\upc x \subseteq U$, so $\bigcap_{i
    \in I} \upc \varphi_i (x) \subseteq U$.  By
  Proposition~\ref{prop:heckmann}, $\varphi_i (x) \subseteq U$ for
  some $i \in I$, i.e., $\varphi_i (x) \in \Box U$.
\end{proof}

\begin{lem}
  \label{lemma:psi*:comp}
  Let $X$ be a $\QRB$-domain, and ${(\varphi_i)}_{i \in I}$ a
  generating family of quasi-deflations.  For every compact saturated
  subset $Q$ of $X$, $Q = \bigcap_{i \in I}^\downarrow
  \varphi_i^\dagger (Q)$.
\end{lem}
\begin{proof}
  Since $x \in \varphi_i (x)$ for every
  $x$, 
  $\varphi_i^\dagger (Q)$ contains $Q$ for every $i \in I$.  So $Q
  \subseteq \bigcap_{i \in I}^\downarrow \varphi_i^\dagger (Q)$.
  Conversely, since $Q$ is saturated, it is enough to show that every
  open $U$ containing $Q$ also contains $\bigcap_{i \in I}^\downarrow
  \varphi_i^\dagger (Q)$.  Since $Q \subseteq U$, by
  Lemma~\ref{lemma:psi*}, $Q \subseteq \bigcup_{i \in I}^\uparrow
  \varphi_i^{-1} (\Box U)$.  By compactness, $Q \subseteq
  \varphi_i^{-1} (\Box U)$ for some $i \in I$, i.e., for every $x \in
  Q$, $\varphi_i (x) \subseteq U$.  So $\varphi_i^\dagger (Q)
  \subseteq U$.
\end{proof}

\begin{prop}
  \label{prop:qrb:Smyth}
  For every $\QRB$-domain $X$, $\Smyth (X)$ is an $\RB$-domain.
\end{prop}
\begin{proof}
  Assume $X$ is a $\QRB$-domain, with generating family of
  quasi-deflations ${(\varphi_i)}_{i \in I}$.  The family
  ${(\varphi_i^\dagger)}_{i \in I}$ is directed, since if $\varphi_i$
  is below $\varphi_j$, i.e., if $\varphi_i (x) \supseteq \varphi_j
  (x)$ for every $x \in X$, then $\varphi_i^\dagger (Q) = \bigcup_{x
    \in Q} \varphi_i (x) \supseteq \bigcup_{x \in Q} \varphi_j (x) =
  \varphi_j^\dagger (Q)$.  Since $X$ is quasi-continuous
  (Corollary~\ref{corl:qrb:qcont}), it is sober and locally compact.
  So Corollary~\ref{corl:fin:cont:Smyth} applies, showing that
  $\varphi_i^\dagger$ is Scott-continuous from $\Smyth (X)$ to $\Smyth
  (X)$.  Lemma~\ref{lemma:psi*:comp} states that the least upper bound
  of ${(\varphi_i^\dagger)}_{i \in I}$ is the identity on $\Smyth
  (X)$.  Clearly, $\varphi_i^\dagger$ has finite image.  So $\Smyth
  (X)$ is an $\RB$-domain.
\end{proof}

\begin{thm}
  \label{thm:qrb:scomp}
  Every $\QRB$-domain is stably compact.
\end{thm}
\begin{proof}
  Let $X$ be a $\QRB$-domain, with generating family of
  quasi-deflations ${(\varphi_i)}_{i \in I}$.  We claim that, given
  any two compact saturated subsets $Q$ and $Q'$ of $X$, $Q \cap Q'$
  is again compact saturated.  This is obvious if $Q \cap Q'$ is
  empty.  Otherwise, writing $\upc_Y y$ for the upward closure of an
  element $y$ of a poset $Y$, $\upc_{\Smyth (X)} Q \cap \upc_{\Smyth
    (X)} Q'$ is an intersection of two finitary compacts in $\SV (X)$.
  Since $X$ is a quasi-continuous dcpo by
  Corollary~\ref{corl:qrb:qcont}, $X$ is sober and locally compact, so
  $\SV (X) = \Smyth_\sigma (X)$.  Moreover, $\Smyth (X)$ is an
  $\RB$-domain (Proposition~\ref{prop:qrb:Smyth}), so $\SV (X)$ is
  coherent.  Therefore $\upc_{\Smyth (X)} Q \cap \upc_{\Smyth (X)} Q'$
  is compact saturated in $\SV (X)$.  It is also non-empty: pick $x
  \in Q \cap Q'$, then $\upc_X x$ is in $\upc_{\Smyth (X)} Q \cap
  \upc_{\Smyth (X)} Q'$.  So $\upc_{\Smyth (X)} Q \cap \upc_{\Smyth
    (X)} Q'$ is in $\Smyth (\SV (X))$.  Now there is a (continuous)
  map $\mu_X : \SV (\SV (X)) \to \SV (X)$ defined as $\identity {\SV
    (X)}^\dagger$---this is the so-called {\em multiplication\/} of
  the monad---and $\mu_X (\upc_{\Smyth (X)} Q \cap \upc_{\Smyth (X)}
  Q')$ is then an element of $\Smyth (X)$, i.e., a compact subset of
  $X$.  We now observe that $\mu_X (\upc_{\Smyth (X)} Q \cap
  \upc_{\Smyth (X)} Q') = \bigcup_{\substack{Q'' \in \Smyth (X)\\Q''
      \subseteq Q, Q'}} Q''$ is equal to $Q \cap Q'$: the left to
  right inclusion is obvious, and conversely every $x \in Q \cap Q'$
  defines an element $Q'' = \upc_X x$ of $\Smyth (X)$ that is included
  in $Q$ and $Q'$.  So $Q \cap Q'$ is compact saturated.  We conclude
  that $X$ is coherent.

  $X$ is compact since pointed, and also locally compact and sober, as
  a quasi-continuous dcpo, hence stably compact.  
\end{proof}

The {\em Lawson topology\/} is the smallest topology containing both
the Scott-opens and the complements of all finitary compacts $\upc E
\in \Fin (X)$.  When $X$ is a quasi-continuous dcpo, since $\upc E$ is
compact saturated and every non-empty compact saturated subset is a
filtered intersection of such sets $\upc E$, the Lawson topology
coincides with the {\em patch topology\/}, i.e., the smallest topology
containing the original Scott topology and all complements of compact
saturated subsets.  Every stably compact space is patch-compact, i.e.,
compact in its patch topology
\cite[Section~VI-6]{GHKLMS:contlatt}.  So:
\begin{cor}
  \label{corl:qrb:lcomp}
  Every $\QRB$-domain is Lawson-compact.
\end{cor}

In the sequel, we shall need some form of countability:
\begin{defi} 
  \label{defn:omega:qrb}
  An {\em $\omega\QRB$-domain\/} is a $\QRB$-domain with a {\em
    countable\/} generating family of quasi-deflations.
\end{defi}

\begin{prop}
  \label{prop:omega:qrb}
  A pointed dcpo $X$ is an $\omega\QRB$-domain iff there is a
  generating {\em sequence\/} of quasi-deflations ${(\varphi_i)}_{i
    \in \nat}$, i.e., for every $i,i' \in \nat$, $i \leq i'$,
  $\varphi_i (x) \supseteq \varphi_{i'} (x)$ for every $x \in X$, and
  $\upc x = \bigcap_{i \in \nat}^\downarrow {\varphi_i (x)}$ for every
  $x \in X$.
\end{prop}
\begin{proof}
  Let $X$ be an $\omega\QRB$-domain, and ${(\psi_j)}_{j \in \nat}$ be
  a countable generating family of quasi-deflations.  Build a sequence
  ${(j_i)}_{i \in \nat}$ by letting $j_0 = 0$, and $j_{i+1}$ be any $j
  \in \nat$ such that $\psi_j$ is above $\psi_i$ and $\psi_{j_i}$, by
  directedness.  Then let $\varphi_i = \psi_{j_i}$ for every $i \in
  \nat$.  By construction, whenever $i \leq i'$, $\varphi_i$ is below
  $\varphi_{i+1}$.  And for every $i \in \nat$, $\psi_i$ is below
  $\varphi_i = \psi_{j_i}$, so $\upc x = \bigcap_{i \in
    \nat}^\downarrow {\varphi_i (x)}$ for every $x \in X$.  So
  ${(\varphi_i)}_{i \in \nat}$ is the desired sequence.
\end{proof}

Recall that a topological space is \emph{countably-based} if and only
if it has a countable subbase, or equivalently, a countable base.
\begin{prop}
  \label{prop:omega:qrb:2}
  A $\QRB$-domain $X$ is an $\omega\QRB$-domain iff it is
  countably-based.
\end{prop}
\begin{proof}
  Only if: let ${(\varphi_i)}_{i \in \nat}$ be a generating sequence
  of quasi-deflations on $X$.  For each $i \in \nat$, enumerate $\img
  \varphi_i$ as $\{\upc E_{i1}, \ldots, \upc E_{in_i}\} \subseteq \Fin
  (X)$, and let $E_i$ be the finite set $\bigcup_{j=1}^{n_i} E_{ij}$.
  We claim that the countably many subsets $\interior {\varphi_i
    (y)}$, $y \in E_j$, $i, j \in \nat$, form a base of the topology.

  It is enough to show that, for every open $U$ and every element $x
  \in U$, $x \in \interior {\varphi_i (y)}$ for some $y \in E_j$, $i,
  j\in \nat$, such that $\varphi_i (y) \subseteq U$: since $\upc x =
  \bigcap_{j \in \nat}^\downarrow \varphi_j (x) \subseteq U$, use
  Proposition~\ref{prop:heckmann} to find $j \in \nat$ such that
  $\varphi_j (x) \subseteq U$.  Since $x \in \varphi_j (x)$ and
  $\varphi_j (x) = \upc E_{jk}$ for some $k$, there is a $y \in E_{jk}
  \subseteq E_j$ such that $y \leq x$, and $y \in U$.  Repeating the
  argument on $y$, we find $i \in \nat$ such that $\varphi_i (y)
  \subseteq U$.  By Lemma~\ref{lemma:qdefl:ll}, $\varphi_i (y) \cll
  y$, i.e., $y$ is in $\interior {\varphi_i (y)}$ since $X$ is
  quasi-continuous.  Since $y \leq x$, $x$ is in $\interior {\varphi_i
    (y)}$.

  If: let ${(\varphi_i)}_{i \in I}$ be a generating family of
  quasi-deflations on $X$, and assume that the topology of $X$ has a
  countable base $\{U_k \mid k\in \nat\}$.  Assume without loss of
  generality that $U_k \neq \emptyset$ for every $k \in \nat$.  For
  every pair $\ell, k \in \nat$ such that $U_\ell \subseteq \upc E
  \subseteq U_k$ for some finite set $E$, pick one such finite set and
  call it $E_{\ell k}$.  One can enumerate all such pairs as $\ell_m,
  k_m$, $m \in \nat$.  By Lemma~\ref{lemma:psi*:comp}, $\bigcap_{i \in
    I}^\downarrow \varphi_i^\dagger (\upc E_{\ell_m k_m}) = \upc
  E_{\ell_m k_m}$.  By Proposition~\ref{prop:heckmann},
  $\varphi_i^\dagger (\upc E_{\ell_m k_m}) \subseteq U_{k_m}$ for some
  $i \in I$: pick such an $i$ and call it $i_m$.  By directedness, we
  may also assume that $\varphi_{i_m}$ is also above $\varphi_{i_n}$,
  $0 \leq n < m$.  Define $\psi_m$ as $\varphi_{i_m}$.  This yields a
  non-decreasing sequence of quasi-deflations ${(\psi_m)}_{m \in
    \nat}$.

  We claim that it is generating.  On one hand, $\upc x \subseteq
  \bigcap_{k \in \nat}^\downarrow \psi_k (x)$ since each $\psi_k$ is a
  quasi-deflation.  Conversely, every open neighborhood $U$ of $x$
  contains some $U_k$, $k \in \nat$, with $x \in U_k$.  Then $\upc x =
  \bigcap_{i \in I}^\downarrow \varphi_i (x) \subseteq U_k$, so
  $\varphi_i (x) \subseteq U_k$ for some $i \in I$.  Write $\varphi_i
  (x)$ as $\upc E$, where $E$ is finite.  By
  Lemma~\ref{lemma:qdefl:ll}, $\varphi_i (x) \cll x$, so $x \in
  \uuarrow E \subseteq \upc E \subseteq U_k$.  As $\uuarrow E$ is
  open, $x \in U_\ell \subseteq \uuarrow E$ for some $\ell \in \nat$.
  In particular, $U_\ell \subseteq \upc E \subseteq U_k$.  So $\ell,
  k$ is a pair of the form $\ell_m, k_m$.  By definition
  $\psi_m^\dagger (\upc E_{\ell k}) \subseteq U_k$.  Since $x \in
  U_\ell \subseteq \upc E_{\ell k}$, $\psi_m (x) = \psi_m^\dagger
  (\upc x) \subseteq \psi_m^\dagger (\upc E_{\ell k}) \subseteq U_k
  \subseteq U$.  So every open neighborhood $U$ of $x$ contains
  $\psi_m (x)$ for some $m \in \nat$, hence $\bigcap_{m \in
    \nat}^\downarrow \psi_m (x)$.  So $\bigcap_{m \in \nat}^\downarrow
  \psi_m (x) \subseteq \upc x$, whence the equality.
\end{proof}

\section{Quasi-Retracts of Bifinite Domains}
\label{sec:qretr}

The $\RB$-domains can be characterized as the retracts of bifinite
domains.  Recall that a {\em retraction\/} of $X$ onto $Y$ is a
continuous map $r : X \to Y$ such that there is continuous map $s : Y
\to X$ (the {\em section\/}) with $r (s (y)) = y$ for every $y \in Y$.

We shall show that ($\omega$)$\QRB$-domains are not just closed under
retractions, but under a more relaxed notion that we shall call {\em
  quasi-retractions\/}.  More precisely, our aim in this section is to
show that the $\omega\QRB$-domains are exactly the quasi-retracts of
bifinite domains, up to some details.

\begin{figure}
  \centering
  \ifpdf
  \input{qretr.pdftex_t}
  \else
  \input{qretr.pstex_t}
  \fi
  \caption{A quasi-retraction}
  \label{fig:qretr}
\end{figure}
For each continuous $r : X \to Y$, define $\Smyth r : \SV (X) \to \SV
(Y)$ by $\Smyth r (Q) = \upc \{r (x) \mid x \in Q\}$.  $\Smyth r$ is
continuous, since $\Smyth r^{-1} (\Box V) = \Box r^{-1} (V)$ for every
open $V$.  This is the action of the $\SV$ functor of the Smyth
powerspace monad \cite[Chapter~7]{Schalk:PhD}, equivalently $\Smyth r
= {(\eta_Y \circ r)}^\dagger$.

\begin{defi}[Quasi-retract]
  \label{defn:qretr}
  A {\em quasi-retraction\/} $r : X \to Y$ of $X$ onto $Y$ is a
  continuous map such that there is a continuous map $\qs : Y \to\SV
  (X)$ (the {\em quasi-section\/}) such that $\Smyth r (\qs (y)) =
  \upc y$ for every $y \in Y$.

  A topological space $Y$ is a {\em quasi-retract\/} of $X$ iff there
  is a quasi-retraction of $X$ onto $Y$.
\end{defi}
In diagram notation, we require the bottom right triangle to commute,
but \emph{not} the top left triangle---what the puncture {\LARGE+}
indicates; the outer square always commutes:
\begin{equation}
  \label{eq:qretr}
  \xymatrix{
    X \ar[r]^r \ar[d]_{\eta_X} & Y \ar[d]^{\eta_Y} \ar[ld]|{\qs}_(0.6){\text{\LARGE+}} \\
    \SV (X) \ar[r]_{\Smyth r} & \SV (Y)
  }
\end{equation}
While a section $s : Y \to X$ picks an element $s (y)$ in the inverse
image $r^{-1} (y)$, continuously, a quasi-section is only required to
pick a non-empty compact saturated collection of elements from $r^{-1}
(\upc y)$ meeting $r^{-1} (y)$ (see Figure~\ref{fig:qretr}),
continuously again.

Every retraction $r$ (with section $s$) defines a canonical
quasi-retraction: let $\qs (y) = \upc s (y)$, then $\Smyth r (\qs (y))
= \upc \{r (z) \mid s (y) \leq z\} = \upc r (s (y)) = \upc y$.

The converse fails.  For example, $\mathcal N_2$ is a quasi-retract of
$\nat_\omega + \nat_\omega$ (see Figure~\ref{fig:ex1}~$(iii)$): $r$
maps both $(0,\omega)$ and $(1,\omega)$ to $\omega \in \mathcal N_2$,
and $\qs (y) = r^{-1} (\upc y)$ for every $y$.  But $Y$ is not a
retract of $X$: $X$ is a continuous dcpo, and every retract of a
continuous dcpo is again one; recall that $\mathcal N_2$ is not
continuous.

Every quasi-retraction $r : X \to Y$ induces a continuous map $\eta_Y
\circ r : X \to \SV (Y)$, which is then a retraction in the Kleisli
category $\pmb{C}_{\Smyth}$.  A retraction in a category is a morphism
$r : X \to Y$ such that there is a section morphism $s : Y \to X$,
i.e., one with $r \circ s = \identity Y$.  It is easy to see that the
quasi-retractions are exactly those continuous maps $r : X \to Y$ such
that $\eta_Y \circ r$ is a retraction in $\pmb{C}_{\Smyth}$.

\begin{lem}
  \label{lemma:qretr:surj}
  Every quasi-retraction $r : X \to Y$ onto a $T_0$ space $Y$ is
  surjective.  More precisely, if $\qs$ is a matching quasi-section,
  then every element $y \in Y$ is of the form $r (x)$ for some $x \in
  \qs (y)$.
\end{lem}
\begin{proof}
  For every $y \in Y$, $\upc y = \Smyth r (\qs (y))$.  Since $y \in
  \Smyth r (\qs (y))$, $r (x) \leq y$ for some $x \in \qs (y)$.  But
  $r (x)$ is then in $\Smyth r (\qs (y)) = \upc y$, so $y \leq r (x)$.
  Therefore $y = r (x)$.
\end{proof}

The following is reminiscent of the fact that every retract of a
stably compact space is again stably compact \cite[Proposition, bottom
of p.153, and subsequent discussion]{Lawson:versatile}: we shall show
that any $T_0$ quasi-retract of a stably compact space is stably
compact.  We start with compactness.
\begin{lem}
  \label{lemma:qretr:comp}
  Every $T_0$ quasi-retract $Y$ of a compact space $Y$ is compact.
\end{lem}
\begin{proof}
  The image of a compact set by a continuous map is compact.  Now
  apply Lemma~\ref{lemma:qretr:surj}.
\end{proof}

\begin{lem}
  \label{lemma:qretr:wf}
  Any quasi-retract $Y$ of a well-filtered space $X$ is well-filtered.
\end{lem}
\begin{proof}
  Let $r : X \to Y$ be the quasi-retraction, with quasi-section $\qs :
  Y \to \SV (X)$.

  Let ${(Q_i)}_{i \in I}$ be a filtered family of compact saturated
  subsets of $Y$, and assume that $\bigcap_{i \in I}^\downarrow Q_i
  \subseteq V$, where $V$ is open in $Y$.  Let $Q'_i = \qs^\dagger
  (Q_i)$.  This is compact saturated, and forms a directed family,
  since $\qs^\dagger$ is monotonic.  We claim that $\bigcap_{i \in I}
  Q'_i \subseteq r^{-1} (V)$.  Indeed, every $x \in \bigcap_{i \in I}
  Q'_i$ is such that, for every $i \in I$, there is a $y_i \in Q_i$
  such that $x \in \qs (y_i)$; then $r (x) \in \Smyth r (\qs (y_i)) =
  \upc {y_i}$, so $r (x) \in Q_i$, for every $i \in I$.  Since
  $\bigcap_{i \in I}^\downarrow Q_i \subseteq V$, $r (x)$ is in $V$,
  whence the claim.

  Since $X$ is well-filtered, $Q'_i \subseteq r^{-1} (V)$ for some $i
  \in I$.  Then, for every $y \in Q_i$, $\qs (y) \subseteq \qs^\dagger
  (Q_i) = Q'_i \subseteq r^{-1} (V)$, so $y \in \Smyth r (\qs (y))
  \subseteq \Smyth r (r^{-1} (V)) \subseteq V$.  So $Q_i \subseteq V$.
\end{proof}

\begin{lem}
  \label{lemma:qretr:coh}
  Any $T_0$ quasi-retract $Y$ of a coherent space $X$ is coherent.
\end{lem}
\begin{proof}
  Let $r : X \to Y$ be the quasi-retraction, with quasi-section $\qs :
  Y \to \SV (X)$.

  We use the fact that $\Smyth r \circ \qs^\dagger$ is the identity on
  $\SV (Y)$.  This is a well-known identity on monads: by the monad
  law $(g^\dagger \circ h)^\dagger = g^\dagger \circ h^\dagger$, and
  since $\Smyth r = (\eta_Y \circ r)^\dagger$, $\Smyth r \circ
  \qs^\dagger = (\Smyth r \circ \qs)^\dagger$, and this is
  $\eta_Y^\dagger = \identity {\SV (Y)}$ by the first monad law.

  Let $Q_1$, $Q_2$ be two compact saturated subsets of $Y$.  Then
  $\qs^\dagger (Q_1) \cap \qs^\dagger (Q_2)$ is compact saturated in
  $X$, using the fact that $X$ is coherent.  So $\Smyth r (\qs^\dagger
  (Q_1) \cap \qs^\dagger (Q_2))$ is compact saturated in $Y$.  We
  claim that $\Smyth r (\qs^\dagger (Q_1) \cap \qs^\dagger (Q_2)) =
  Q_1 \cap Q_2$, which will finish the proof.  In one direction, every
  element $y$ of $Q_1 \cap Q_2$ is in $\Smyth r (\qs^\dagger (Q_1)
  \cap \qs^\dagger (Q_2))$: by Lemma~\ref{lemma:qretr:surj}, pick $x
  \in \qs (y)$ such that $y = r (x)$, and observe that $x \in
  \qs^\dagger (Q_1)$ (indeed $x \in \qs (y)$, where $y \in Q_1$) and
  $x \in \qs^\dagger (Q_2)$.  In the other direction, $\Smyth r
  (\qs^\dagger (Q_1) \cap \qs^\dagger (Q_2)) \subseteq \Smyth r
  (\qs^\dagger (Q_1)) \cap \Smyth r (\qs^\dagger (Q_2)) = Q_1 \cap
  Q_2$, since $\Smyth r \circ \qs^\dagger$ is the identity on $\Smyth
  (Y)$.
\end{proof}

\begin{lem}
  \label{lemma:qretr:lcomp}
  Any quasi-retract $Y$ of a locally compact space $X$ is locally
  compact.
\end{lem}
\begin{proof}
  Let $r : X \to Y$ be the quasi-retraction, with quasi-section $\qs :
  Y \to \SV (X)$.  Let $y$ be any point of $Y$, and $V$ be an open
  neighborhood of $y$.  Since $y \in V$, $\Smyth r (\qs (y)) = \upc y
  \subseteq V$, so $\qs (y) \subseteq r^{-1} (V)$.  Observe that $\qs
  (y)$ is compact saturated and $r^{-1} (V)$ is open in $X$.  Use
  interpolation in the locally compact space $X$: there is a compact
  saturated subset $Q_1$ such that $\qs (y) \subseteq \interior {Q_1}
  \subseteq Q_1 \subseteq r^{-1} (V)$.

  In particular, $\qs (y) \in \Box \interior {Q_1}$, so $y$ is in the
  open subset $\qs^{-1} (\Box \interior {Q_1})$.  The latter is
  included in the compact subset $\Smyth r (Q_1)$, since every element
  $y'$ of it is such that $\qs (y') \subseteq \interior {Q_1}
  \subseteq Q_1$, hence $\upc y' = \Smyth r (\qs (y')) \subseteq
  \Smyth r (Q_1)$.  In particular, $y$ is in the interior of $\Smyth r
  (Q_1)$.  Finally, since $Q_1 \subseteq r^{-1} (V)$, $\Smyth r (Q_1)
  \subseteq V$.
\end{proof}

\begin{prop}
  \label{prop:qretr:scomp}
  Every $T_0$ quasi-retract $Y$ of a stably compact space $X$ is
  stably compact.
\end{prop}
\begin{proof}
  $Y$ is $T_0$ by assumption, and locally compact, well-filtered,
  compact, and coherent by Lemma~\ref{lemma:qretr:comp},
  Lemma~\ref{lemma:qretr:wf}, Lemma~\ref{lemma:qretr:coh}, and
  Lemma~\ref{lemma:qretr:lcomp}.  In the presence of local
  compactness, it is equivalent to require sobriety or to require the
  space to be $T_0$ and well-filtered
  \cite[Theorem~II-1.21]{GHKLMS:contlatt}.
\end{proof}

Call a space $X$ \emph{locally finitary} if and only if for every $x
\in X$ and every open neighborhood $U$ of $x$, there is a finitary
compact $\upc E$ such that $x \in \interior {\upc E}$ and $\upc E
\subseteq U$.  This is the same definition as for local compactness,
replacing compact saturated subsets by finitary compacts.  The
interpolation property of locally compact spaces refines to the
following: In a locally finitary space $X$, if $Q$ is compact
saturated and included in some open subset $U$, then there is a
finitary compact $\upc E$ such that $Q \subseteq \interior {\upc E}$
and $\upc E \subseteq U$.  The proof is as for interpolation in
locally compact spaces: for each $x \in Q$, pick a finitary compact
$\upc E_x$ such that $x \in \interior {\upc E_x}$ and $\upc E_x
\subseteq U$.  ${(\interior {\upc E_x})}_{x \in Q}$ is an open cover
of $Q$.  Since $Q$ is compact, it has a finite subcover $\upc E_1$,
\ldots, $\upc E_n$.  Then take $E = E_1 \cup \ldots \cup E_n$.

We observe right away the following analog of
Lemma~\ref{lemma:qretr:lcomp}.
\begin{lem}
  \label{lemma:qretr:lfin}
  Any quasi-retract $Y$ of a locally finitary space $X$ is locally
  finitary.
\end{lem}
\begin{proof}
  As in the proof of Lemma~\ref{lemma:qretr:lcomp}, let $y \in Y$ and
  $V$ be an open neighborhood of $y$.  By interpolation between $Q =
  \qs (y)$ and $U = r^{-1} (V)$ in the locally finitary space $X$, we
  find a finitary compact subset $Q_1 = \upc E_1$ of $X$ such that
  $\qs (y) \subseteq \interior {Q_1} \subseteq Q_1 \subseteq r^{-1}
  (V)$.  The rest of the proof is as for
  Lemma~\ref{lemma:qretr:lcomp}, only noticing that $\Smyth r (Q_1) =
  \upc r (E_1)$ is finitary compact.
\end{proof}
The importance of locally finitary spaces lies in the following
result: see Banaschewski \cite{Banaschewski:essn:ext}, or the
equivalence between Items (6) and (11) in Lawson
\cite[Theorem~2]{Lawson:T0:pw:conv}.  See also Isbell
\cite{Isbell:meetcont} for the notion of locally finitary space, up to
change of names.
\begin{prop}
  \label{prop:locfin=qcont}
  The locally finitary sober 
  spaces are exactly the quasi-continuous dcpos in their Scott
  topology.
\end{prop}

We use this, in particular, in the following proposition.
\begin{prop}
  \label{prop:qretr:qrb}
  Every $T_0$ quasi-retract of an ($\omega$)$\QRB$-domain is an
  ($\omega$)$\QRB$-domain.
\end{prop}
\begin{proof}
  Let $X$ be a $\QRB$-domain, $Y$ be a $T_0$ space, $r : X \to Y$ be a
  quasi-retraction, and $\qs : Y \to \SV (X)$ be a matching
  quasi-section.  We first note that $Y$ is stably compact, by
  Proposition~\ref{prop:qretr:scomp}, using the fact that $X$ is
  itself stably compact (Theorem~\ref{thm:qrb:scomp}).  So $Y$ is
  sober. 
  By Proposition~\ref{prop:locfin=qcont}, $X$ is locally finitary, so
  $Y$ is, too, by Lemma~\ref{lemma:qretr:lfin}.  By
  Proposition~\ref{prop:locfin=qcont} again, $Y$ is a quasi-continuous
  dcpo, and its topology is the Scott topology.

  Note that $Y$ is pointed.  Letting $\bot$ be the least element of
  $X$, $r (\bot)$ is the least element of $Y$: for every $y \in Y$,
  pick some $x \in X$ such that $r (x) = y$ by
  Lemma~\ref{lemma:qretr:surj}, then $r (\bot) \leq r (x) = y$.

  For each quasi-deflation $\varphi$ on $X$, $\varphi$ is continuous
  from $X$ to $\Fin_\V (X)$: indeed it is continuous from $X$ to
  $\Fin_\sigma (X)$ and $\Fin_\sigma (X) = \Fin_\V (X)$ by
  Corollary~\ref{corl:BoxU}, since $X$ is quasi-continuous
  (Corollary~\ref{corl:qrb:qcont}).  So $\varphi^\dagger$ makes sense.
  Let $\widehat\varphi : Y \to \Fin_\V (Y)$ map $y$ to $\Smyth r
  (\varphi^\dagger (\qs (y)))$; $\widehat\varphi (y)$ is in $\Fin (Y)$
  because $\varphi^\dagger (\qs (y)) \in \Fin (X)$
  (Lemma~\ref{lemma:dagger:new}, second part), and $\Smyth r (\upc E)
  = \upc \{r (z) \mid z \in E\}$ is finitary compact for every finite
  set $E$.

  Explicitly, $\widehat\varphi (y) = \upc \{r (z) \mid \exists x \in
  \qs (y) \cdot z \in \varphi (x)\}$.

  For every open subset $V$ of $Y$, $\widehat\varphi^{-1} (\Box V)$ is
  the set of all $y \in Y$ such that for every $x \in \qs (y)$, for
  every $z \in \varphi (x)$, $r (z) \in V$.  I.e., for every $x \in
  \qs (y)$, $\varphi (x) \subseteq r^{-1} (V)$, that is, $\qs (y)
  \subseteq \varphi^{-1} (\Box r^{-1} (V))$.  So $\widehat\varphi^{-1}
  (\Box V) = \qs^{-1} (\Box \varphi^{-1} (\Box r^{-1} (V)))$.  Since
  the latter is open, and the sets $\Box V$ form a subbase of the
  topology of $\SV (Y)$, $\widehat\varphi$ is continuous from $Y$ to
  $\Fin_\V (Y)$.  Since $Y$ is a quasi-continuous dcpo and its
  topology is Scott, by Corollary~\ref{corl:BoxU} $\Fin_\sigma (Y) =
  \Fin_\V (Y)$, so $\widehat\varphi$ is also Scott-continuous from $Y$
  to $\Fin (Y)$.  (Alternatively, apply Corollary~\ref{corl:fin:cont:Fin}.)

  We claim that $y \in \widehat\varphi (y)$ for every $y \in Y$.
  Since $\Smyth r (\qs (y)) = \upc y$, $y \in \Smyth r (\qs (y))$, so
  there is an $x \in \qs (y)$ such that $r (x) \leq y$.  Now $x \in
  \varphi (x)$, so taking $z = x$ in the definition of
  $\widehat\varphi (y)$, $y$ is in $\widehat\varphi (y)$.

  Let now ${(\varphi_i)}_{i \in I}$ be a generating family of
  quasi-deflations on $X$.  Clearly, if $\varphi_i$ is below
  $\varphi_j$, then $\widehat\varphi_i$ is below $\widehat\varphi_j$,
  so ${(\widehat\varphi_i)}_{i \in I}$ is directed.

  It remains to show that $\bigcap_{i \in I}^\downarrow
  \widehat\varphi_i (y) = \upc y$ for every $y \in Y$.  Since $y \in
  \widehat\varphi_i (y)$, it remains to show $\bigcap_{i \in
    I}^\downarrow \widehat\varphi_i (y) \subseteq \upc y$: we show
  that every open $V$ containing $y$ contains $\bigcap_{i \in
    I}^\downarrow \widehat\varphi_i (y)$.  Since $y \in V$ and $\Smyth
  r (\qs (y)) = \upc y$, $\Smyth r (\qs (y)) \subseteq V$, so $\qs (y)
  \in \Smyth r^{-1} (\Box V) = \Box r^{-1} (V)$, i.e., $\qs (y)
  \subseteq r^{-1} (V)$. 
  By Lemma~\ref{lemma:psi*}, $\bigcup_{i \in I}^\uparrow
  \varphi_i^{-1} (\Box r^{-1} (V)) = r^{-1} (V)$.  Since $\qs (y)$ is
  compact, $\qs (y) \subseteq \varphi_i^{-1} (\Box r^{-1} (V))$ for
  some $i \in I$.  So $y$ is in $\qs^{-1} (\Box \varphi_i^{-1} (\Box
  r^{-1} (V)))$, which is equal to $\widehat\varphi_i^{-1} (\Box V)$
  (see above).  It follows that $V$ contains $\widehat\varphi_i (y)$,
  hence $\bigcap_{i \in I}^\downarrow \widehat\varphi_i (y)$.  So $Y$
  is a $\QRB$-domain.

  The case of $\omega\QRB$-domains is similar, where now
  ${(\varphi_i)}_{i \in \nat}$ is a generating {\em sequence\/} of
  quasi-deflations.
\end{proof}

Later, we shall need a refinement of the notion of quasi-retraction,
which is to the latter as projections are to retractions.  Recall that
a {\em projection\/} is a retraction $r : X \to Y$, with section $s$,
such that additionally $s \circ r \leq \identity X$.  Similarly, it is
tempting to define a \emph{quasi-projection} as a quasi-retraction
(with quasi-section $\qs$) such that $x \in \qs (r (x))$ for every $x
\in X$.  If $r$ is a retraction, with section $s$, and we see $r$ as a
quasi-retraction in the canonical way, defining $\qs (y)$ as $\upc s
(y)$, then the quasi-projection condition $x \in \qs (r (x))$ is
equivalent to the projection condition $(s \circ r) (x) \leq x$.


The point $x$ shown in Figure~\ref{fig:qretr} satisfies the condition
$x \in \qs (r (x))$: $x$ is in the gray area $\qs (y)$, where $y = r
(x)$.  However, Lemma~\ref{lemma:qproj:unique} below shows that $r$ is
not a quasi-projection: for this to be the case, the gray area $\qs
(y)$ should fill the whole of $r^{-1} (\upc y)$.

There is no need to invent a new term, though:
Lemma~\ref{lemma:qproj:unique} shows that quasi-projections are
nothing else than proper surjective maps.
A map $r : X \to Y$ is {\em proper\/} if and only if it is continuous,
$\dc r (F)$ is closed in $Y$ for every closed subset $F$ of $X$, and
$r^{-1} (\upc y)$ is compact in $X$ for every element $y$ of $Y$
\cite[Lemma~VI-6.21~$(i)$]{GHKLMS:contlatt}.
\begin{lem}
  \label{lemma:qproj:unique}
  Let $X$ be a topological space, and $Y$ be a $T_0$ topological
  space.  For a map $r : X \to Y$, the following two conditions are
  equivalent:
  \begin{enumerate}[\em(1)]
  \item\label{q:qproj} $r$ is a quasi-retraction, with matching
    quasi-section $\qs : Y \to \SV (X)$, such that additionally $x \in
    \qs (r (x))$ for every $x \in X$;
  \item $r$ is proper and surjective.
  \end{enumerate}
  Then the quasi-section $\qs$ in (\ref{q:qproj}) is unique, and it is
  defined by $\qs (y) = r^{-1} (\upc y)$.
\end{lem}
\begin{proof}
  We first prove the following fact, which will serve in both
  directions of proof: $(*)$ assume $\qs (y) = r^{-1} (\upc y)$ for
  every $y \in Y$, then for every open subset $U$ of $X$, the
  complement of $\qs^{-1} (\Box U)$ in $Y$ is $\dc r (F)$, where $F$
  is the complement of $U$ in $X$.  Indeed, the complement of
  $\qs^{-1} (\Box U)$ is the set of elements $y \in Y$ such that $\qs
  (y)$ is not included in $U$, i.e., such that there is an $x \in \qs
  (y)$ that is not in $U$, i.e., in $F$.  Since $\qs (y) = r^{-1}
  (\upc y)$, this is the set of elements $y$ such that there is an $x
  \in F$ such that $y \leq r (x)$, namely, $\dc r (F)$.

  Assume $r$ is a quasi-retraction, and $\qs$ is a matching
  quasi-section such that $x \in \qs (r (x))$ for every $x \in X$.  We
  have seen that $r$ is surjective (Lemma~\ref{lemma:qretr:surj}).

  Since $\Smyth r (\qs (y)) = \upc y$, every element $x$ of $\qs (y)$
  is such that $r (x)$ is in $\upc y$, so $\qs (y) \subseteq r^{-1}
  (\upc y)$.  Conversely, for every $x \in r^{-1} (\upc y)$, i.e., if
  $y \leq r (x)$, then $\qs (y) \supseteq \qs (r (x))$ since $\qs$ is
  monotonic.  We have assumed that $x$ was in $\qs (r (x))$, so $x \in
  \qs (y)$.  It follows that $\qs (y) = r^{-1} (\upc y)$, which proves
  the last claim in the Lemma.

  It also follows that $r^{-1} (\upc y)$ is compact in $X$.  And,
  using $(*)$, for every closed subset $F$ of $X$, with complement
  $U$, $\dc r (F)$ is the complement of $\qs^{-1} (\Box (U))$, which
  is open since $\qs$ is continuous, so $\dc r (F)$ is closed.
  Therefore $r$ is proper.

  Conversely, assume that $r$ is proper and surjective.  Define $\qs
  (y)$ as $r^{-1} (\upc y)$.  Since $r$ is surjective, $\qs (y)$ is
  non-empty.  It is saturated, i.e., upward closed, because $r$ is
  monotonic.  Since $r^{-1} (\upc y)$ is compact, $\qs (y)$ is an
  element of $\Smyth (Y)$.  For every open subset $U$ of $X$, with
  complement $F$, $\qs^{-1} (\Box U)$ is the complement of $\dc r (F)$
  by $(*)$, hence is open since $r$ is proper.  So $\qs$ is
  continuous.

  The equation $\Smyth r (\qs (y)) = \upc y$ follows from $\Smyth r
  (\qs (y)) = \upc \{r (x) \mid x \in r^{-1} (\upc y)\}$ and the fact
  that $r$ is surjective.  It is clear that $x$ is in $\qs (r (x)) = r^{-1}
  (\upc r (x))$ for every $x \in X$.
\end{proof}


Let us turn to bifinite domains, or rather to their countably-based
variant.  Countability will be needed in a few crucial places.

A pointed dcpo $X$ is an {\em $\omega\B$-domain\/} (a.k.a.\ an
SFP-domain) iff there is a non-decreasing sequence of idempotent
deflations ${(f_i)}_{i \in \nat}$ such that, for every $x \in X$, $x =
\sup_{i \in \nat} f_i (x)$.  I.e., an $\omega\B$-domain is just like a
$\B$-domain, except that we take a non-decreasing sequence, not a
general directed family of idempotent deflations.



The key lemma to prove Theorem~\ref{thm:qrb:qretr} below is the
following refinement of Rudin's Lemma \cite[III-3.3]{GHKLMS:contlatt}.
Note that Rudin's Lemma would only secure the existence of a directed
family $Z$ whose least upper bound is $y$, and which intersects each
$E^0_i$; but $Z$ may intersect each $E^0_i$ in more than one element
$y_i$.  We pick exactly one element $y_i$ in each $E^0_i$, and for
this countability seems to be needed.
\begin{lem}
  \label{lemma:qs:nonempty}
  Let $Y$ be a dcpo, $y \in Y$, and ${(\upc E^0_i)}_{i \in \nat}$ a
  non-decreasing sequence in $\Fin (Y)$ (w.r.t.\ $\supseteq$) such that $\upc y = \bigcap_{i
    \in \nat}^\downarrow \upc E^0_i$.
  There is a non-decreasing sequence ${(y_i)}_{i \in \nat}$ in $Y$
  such that $y_i \in E^0_i 
  $ for every $i \in \nat$, and $\sup_{i \in \nat} y_i = y$.
\end{lem}
\begin{proof}
  Let $E_i = E^0_i \cap \dc y$ for every $i \in \nat$.  ${(E_i)}_{i
    \in \nat}$ is a non-decreasing sequence in $\Fin (Y)$ such that $y
  \in \bigcap_{i \in \nat}^\downarrow \upc E_i$, and $E_i \subseteq
  \dc y$.

  Build a tree as follows.  Informally, there is a root node, all
  (non-root) nodes at distance $i \geq 1$ from the root node are
  labeled by some element of $E_{i-1}$, and each such node $N$,
  labeled $y_{i-1}$, say, has as many successors as there are elements
  $y_i$ in $E_i$ such that $y_{i-1} \leq y_i$.  Formally, one can
  define the nodes as being the sequences $y_0, y_1, \ldots, y_{i-1}$,
  $i \in \nat$, where $y_0 \in E_0$, $y_1 \in E_1$, \ldots, $y_{i-1}
  \in E_{i-1}$, and $y_0 \leq y_1 \leq \ldots \leq y_{i-1}$.  Such a
  node is labeled $y_{i-1}$ (if $i \geq 1$), and its successors are
  all the sequences $y_0, y_1, \ldots, y_{i-1}, y_i$ with $y_i$ chosen
  in $E_i$, and above $y_{i-1}$ if $i \geq 1$.

  This tree has arbitrarily long branches (paths from the root).
  Indeed, for every $i \in \nat$, pick an element $y_i \in E_i$---this
  is possible since $y \in \upc E_i$, hence $E_i$ is non-empty---,
  then an element $y_{i-1} \in E_{i-1}$ below $y_i$---since $\upc
  E_{i-1} \supseteq \upc E_i$---, then an element $y_{i-2} \in
  E_{i-2}$ below $y_{i-1}$, \ldots, and finally an element $y_0 \in
  E_0$ below $y_1$.  This is a node at distance $i+1$ from the root.

  It follows that the tree is infinite.  It is finitely-branching,
  meaning that every node has only finitely many successors---because
  $E_i$ is finite.  K\H{o}nig's Lemma then states that this tree must
  have an infinite branch.  Reading the labels on non-root nodes in
  this branch, we obtain an infinite sequence $y_0 \leq y_1 \leq
  \ldots \leq y_i \leq \ldots$ of elements $y_i \in E_i$, $i \in
  \nat$.  Clearly, $y_i \in E^0_i$ for each $i \in \nat$.  In
  particular, $\sup_{i \in \nat} y_i \in \bigcap_{i \in
    \nat}^\downarrow \upc E^0_i = \upc y$, so $y \leq \sup_{i \in
    \nat} y_i$.  Since $E_i \subseteq \dc y$ for every $i \in \nat$,
  the converse inequality holds.  So $\sup_{i \in \nat} y_i = y$.
\end{proof}

\begin{thm}
  \label{thm:qrb:qretr}
  The following are equivalent for a dcpo $Y$:
  \begin{description}
  \item[$(i)$] $Y$ is an $\omega\QRB$-domain;
  \item[$(ii)$] $Y$ is a quasi-retract of an $\omega\B$-domain;
  \item[$(iii)$] $Y$ is the image of an $\omega\B$-domain under a proper map.
  \end{description}
\end{thm}
\begin{proof}
  $(iii) \limp (ii)$.  Because any proper surjective map is a
  quasi-retraction (Lemma~\ref{lemma:qproj:unique}).

  $(ii) \limp (i)$.  Write $Y$ as a quasi-retract of an
  $\omega\B$-domain $X$.  $X$ is trivially an $\omega\QRB$-domain.
  Since $Y$, as a dcpo, is $T_0$, Proposition~\ref{prop:qretr:qrb}
  applies, so $Y$ is an $\omega\QRB$-domain.

  $(i) \limp (iii)$.  Let $Y$ be an $\omega\QRB$-domain, with
  generating sequence of quasi-deflations ${(\varphi_i)}_{i \in
    \nat}$.  Let $\img \varphi_i = \{\upc E_{i1}, \ldots, \upc
  E_{in_i}\}$, and define $E_i$ as the finite set $\bigcup_{j=1}^{n_i}
  E_{ij}$, for each $i \in \nat$.  Let $X$ be the set of all
  non-decreasing sequences $\vec y = {(y_i)}_{i \in \nat}$ in $Y$ such
  that $y_i \in \bigcup_{j \leq i} E_j$, and $y_i \in \varphi_i
  (\sup_{k \in \nat} y_k)$.  Order $X$ componentwise.  As in
  \cite[Theorem~4.9, Theorem~4.1]{Jung:CCC}, $X$ is an
  $\omega\B$-domain: for each $i_0 \in \nat$, consider the idempotent
  deflation $f_{i_0}$ defined by $f_{i_0} (\vec y) = {(y_{\min (i,
      i_0)})}_{i \in \nat}$.  To show that this is well-defined, we
  must show that $y_{\min (i,i_0)} \in \varphi_i (\sup_{k \in \nat}
  y_{\min (k,i_0)})$, i.e., that $y_{\min (i, i_0)} \in \varphi_i
  (y_{i_0})$.  If $i \leq i_0$, then $y_{\min (i,i_0)} = y_i \in
  \varphi_i (\sup_{k \in \nat} y_k) \subseteq \varphi_i (y_{i_0})$
  since $\vec y \in X$ and $\varphi_i$ is monotonic, else $y_{\min
    (i,i_0)} = y_{i_0} \in \varphi_i (y_0)$ since $\varphi_i$ is a
  quasi-deflation.  It is easy to see that $f_{i_0}$ is
  Scott-continuous.

  Let now $r : X \to Y$ map $\vec y$ to $\sup_{i \in \nat} y_i$.  This
  is evidently Scott-continuous.  For any fixed $y \in Y$, apply
  Lemma~\ref{lemma:qs:nonempty} with $\upc E^0_i = \varphi_i (y)$ to
  obtain a non-decreasing sequence $\vec y = {(y_i)}_{i \in \nat}$
  such that $y_i \in \varphi_i (y)$ for every $i \in \nat$ and
  $\sup_{i \in \nat} y_i = y$: in particular, $\vec y$ is in $Y$, and
  $r (\vec y) = y$.  So $r$ is surjective.  Let us show that it is
  proper.

  To this end, we first remark that $r^{-1} (\upc y) = \{\vec y \in X
  \mid \forall i \in \nat \cdot y_i \in \varphi_i (y)\}$.  Indeed, if
  $\vec y = {(y_i)}_{i \in \nat}$ is in $r^{-1} (\upc y)$, then $y
  \leq r (\vec y) = \sup_{k \in \nat} y_k$, and since $\vec y \in X$,
  $y_i \in \varphi_i (\sup_{k \in \nat} y_k) \subseteq \varphi_i (y)$,
  using the fact that $\varphi_i$ is monotonic.  Conversely, if $y_i
  \in \varphi_i (y)$ for every $i \in \nat$, then $r (\vec y) =
  \sup_{i \in \nat} y_i \in \bigcap_{i \in \nat} \varphi_i (y) = \upc
  y$.

  This remark makes it easier for us to show that $r^{-1} (\upc y)$ is
  compact for every $y \in Y$.  For each $i_0 \in \nat$, let $Q_{i_0}
  = \{\vec y \in X \mid \forall i \leq i_0 \cdot y_i \in \varphi_i
  (y)\}$.  Let $K_{i_0}$ be the set of all elements $\vec y$ of
  $Q_{i_0}$ such that $y_i = y_{i_0}$ for every $i \geq i_0$.  Note
  that $K_{i_0}$ is finite, (recall that each $y_i$ with $i \leq i_0$
  is taken from the finite set $\bigcup_{j \leq i} E_j$), and that
  $Q_{i_0} = \upc K_{i_0}$.  Indeed, for every $\vec y \in Q_{i_0}$,
  its image $f_{i_0} (\vec y)$ by the idempotent deflation $f_{i_0}$
  is in $K_{i_0}$, and is below $\vec y$.  So $Q_{i_0}$ is (finitary)
  compact.  Every $\omega\B$-domain is stably compact
  \cite[Theorem~4.2.18]{AJ:domains}, and any intersection of saturated
  compacts in a stably compact space is compact, so $r^{-1} (\upc y) =
  \bigcap_{i_0 \in \nat} Q_{i_0}$ is compact.

  Let us now show that $\dc r (F)$ is closed for every closed subset
  $F$ of $X$.  Consider a directed family ${(z_j)}_{j \in J}$ of
  elements of $\dc r (F)$, and let $z = \sup_{j \in J} z_j$.  Since
  $z_j \in \dc r (F)$, $F$ intersects $r^{-1} (\upc z_j)$.  The family
  ${(r^{-1} (\upc z_j))}_{j \in J}$ is a filtered family of compact
  saturated subsets of $X$, each of which intersects the closed set
  $F$.  Since $X$ is an $\omega\B$-domain, it is stably compact, hence
  well-filtered: so $\bigcap^\downarrow_{j \in J} r^{-1} (\upc z_j)$
  intersects $F$.  (Explicitly: if it did not, it would be included in
  the open complement $U$ of $F$, hence some $r^{-1} (\upc z_j)$ would
  be included in $U$, contradicting the fact that it intersects $F$.)
  Let $\vec y$ be any element of $\bigcap^\downarrow_{j \in J} r^{-1}
  (\upc z_j) \cap F$.  Then $z_j \leq r (\vec y)$ for every $j \in J$,
  so $z = \sup_{j \in J} z_j \leq r (\vec y)$, hence $z \in \dc r
  (F)$.
\end{proof}

\section{Products, Bilimits}
\label{sec:bilim}

We first show that finite products of $\QRB$-domains are again
$\QRB$-domains.
\begin{lem}
  \label{lem:qrb:prod}
  If ${(\varphi_i)}_{i \in I}$ (resp.\ ${(\psi_j)}_{j \in J}$) is a
  generating family of quasi-deflations on $X$ (resp.\ $Y$), then
  ${(\chi_{ij})}_{i \in I, j \in J}$ is one on $X \times Y$, where
  $\chi_{ij} (x,y) = \varphi_i (x) \times \psi_j (y)$.
\end{lem}
\begin{proof}
  Clearly, $(x,y) \in \chi_{ij} (x,y)$, $\chi_{ij} (x, y)$ is finitary
  compact, and $\img \chi_{ij}$ is finite.  For all $i, j$,
  $\chi_{ij}$ is easily seen to be Scott-continuous, and $\bigcap_{i
    \in I,\: j \in J}^\downarrow \chi_{ij} (x,y) = \bigcap_{i \in I,\: j
    \in J}^\downarrow (\varphi_i (x) \times \psi_j (y)) = \bigcap_{i
    \in I}^\downarrow \varphi_i (x) \times \bigcap_{j \in
    J}^\downarrow \psi_j (y) = \upc x \times \upc y = \upc (x,y)$.
\end{proof}
So:
\begin{lem}
  \label{lemma:qrb:prod}
  For any two ($\omega$)$\QRB$-domains $X$, $Y$, $X \times Y$, with
  the product ordering, is an ($\omega$)$\QRB$-domain.
\end{lem}

Recall that a retraction $p : X \to Y$, with section $e : Y \to X$, is
a projection iff, additionally, $e (p (x)) \leq x$ for every $x \in
X$; then $e$ is usually called an {\em embedding\/}, and is determined
uniquely from $p$.  An \emph{expanding system} of dcpos is a family
${(X_i)}_{i \in I}$, where $I$ is a directed poset (with ordering
$\leq$), with projection maps ${(p_{ij})}_{i,j\in I, i\leq j}$ where
$p_{ij} : X_j \to X_i$, $p_{ii} = \identity {X_i}$, and $p_{ik} =
p_{ij} \circ p_{jk}$ whenever $i\leq j\leq k$
\cite[Section~3.3.2]{AJ:domains}. This is nothing else than a
projective system of dcpos, where the connecting maps $p_{ij}$ must be
projections.  If $e_{ij} : X_i \to X_j$ is the associated embedding,
then one checks that $e_{ii} = \identity {X_i}$ and $e_{ik} = e_{jk}
\circ e_{ij}$ whenever $i \leq j \leq k$, so that ${(X_i)}_{i \in I}$
together with ${(e_{ij})}_{i, j \in I, i \leq j}$ forms an inductive
system of dcpos as well.  In the category of dcpos, the projective
limit of the former coincides with the inductive limit of the latter
(up to natural isomorphism), and is called the \emph{bilimit} of the
expanding system of dcpos.  We write this bilimit as $\lim_{i \in I}
X_i$, leaving the dependence on $\leq$, $p_{ij}$, $e_{ij}$, implicit.
This can be built as the dcpo of all those elements $\vec x =
{(x_i)}_{i \in I} \in \prod_{i \in I} X_i$ such that $p_{ij} (x_j) =
x_i$ for all $i, j \in I$ with $i \leq j$, with the componentwise
ordering.

General bilimits of countably-based dcpos will fail to be
countably-based in general, so we shall restrict to bilimits of
\emph{expanding sequences} of dcpos
\cite[Definition~3.3.6]{AJ:domains}: these are expanding systems of
dcpos where the index poset $I$ is $\nat$, with its usual ordering.
To make it clear what we are referring to, we shall call
\emph{$\omega$-bilimit} of spaces any bilimit of an expanding sequence
(not system) of spaces.

One can appreciate bilimits by realizing that the $\B$-domains are (up
to isomorphism) the bilimits of expanding systems of finite, pointed
posets \cite[Theorem~4.2.7]{AJ:domains}.  Similarly, the
$\omega\B$-domains are the $\omega$-bilimits of expanding sequences of
finite, pointed posets.

Bilimits are harder to deal with than products.  But the difficulty
was solved by Jung \cite[Section~4.1]{Jung:CCC} in the case of
$\RB$-domains and deflations, and we proceed in a very similar way.
We first recapitulate the notion of bilimit.

Consider any set $G$ of functions $\psi$ from $X$ to $\Fin (X)$ such
that $\psi (x) \supseteq \upc x$, i.e., $x \in \psi (x)$, for every $x
\in X$.  We say that $G$ is {\em qfs\/} (for {\em quasi-finitely
  separating\/}) iff given any finitely many pairs $(\upc E_k, x_k)
\in \Fin (X) \times X$ with $\upc E_k \cll x_k$, $1\leq k\leq n$,
there is a $\psi \in G$ that {\em separates\/} the pairs, i.e., such
that $\upc E_k \supseteq \psi (x_k) \supseteq \upc x_k$ (equivalently,
$x_k \in \psi (x_k) \subseteq \upc E_k$) for every $k$, $1\leq k\leq
n$.
\begin{prop}
  \label{prop:qfs}
  Let $X$ be a poset.  Then $X$ is a $\QRB$-domain iff $X$ is a
  quasi-continuous dcpo and the set $G$ of quasi-deflations on $X$ is
  qfs.
\end{prop}
\begin{proof}
  If $X$ is a $\QRB$-domain, then let $(\upc E_k, x_k) \in \Fin (X)
  \times X$ be such that $\upc E_k \cll x_k$ for every $k$, $1\leq
  k\leq n$, and ${(\varphi_i)}_{i \in I}$ be a generating family of
  quasi-deflations.  For each $k$, $1\leq k\leq n$, $\upc x_k =
  \bigcap_{i \in I}^\downarrow \varphi_i (x_k) \subseteq \uuarrow
  E_k$, so by Proposition~\ref{prop:heckmann} there is an $i \in I$
  such that $\varphi_i (x_k) \subseteq \uuarrow E_k \subseteq \upc
  E_k$.  And we may pick the same $i$ for every $k$, by directedness.
  So $\varphi_i$ is the desired $\psi \in G$.

  Also, $X$ is a quasi-continuous dcpo by
  Corollary~\ref{corl:qrb:qcont}.

  Conversely, assume that $X$ is a quasi-continuous dcpo and $G$ is
  qfs.  We show that $H = \{\varphi^\dagger \circ \varphi \mid \varphi
  \in G\}$ is a generating family of quasi-deflations.  Using
  Corollary~\ref{corl:BoxU}, $\Fin_\V (X) = \Fin_\sigma (X)$.  Write
  it $\Fin (X)$, for short.  For each $\varphi \in G$, $\varphi$ is
  continuous from $X$ to $\Fin (X)$, and $\varphi^\dagger$ is
  continuous from $\Fin (X)$ to $\Fin (X)$ by
  Lemma~\ref{lemma:dagger:new}, so $\varphi^\dagger \circ \varphi$ is
  continuous from $X$ to $\Fin (X)$.  Since $x \in \varphi (x)$, $x$
  is also in $\bigcup_{x' \in \varphi (x)} \varphi (x') =
  (\varphi^\dagger \circ \varphi) (x)$.  Also, $\img (\varphi^\dagger
  \circ \varphi)$ is finite, since all its elements are unions of
  elements of the finite set $\img \varphi$.  So $\varphi^\dagger
  \circ \varphi$ is a quasi-deflation.

  Let us show that $H$ is directed.  Pick $\varphi$ and $\varphi'$
  from $G$.  Let $\img \varphi = \{\upc E_1, \ldots, \upc E_m\}$, and
  $E = \bigcup_{i=1}^m E_i$.  Similarly, let $\img \varphi' = \{\upc
  E'_1, \ldots, \upc E'_n\}$ and $E' = \bigcup_{j=1}^n E'_j$.  For
  each $y \in E$, $\varphi (y) \cll y$ by Lemma~\ref{lemma:qdefl:ll}.
  Since $X$ is quasi-continuous, use interpolation, and pick a
  finitary compact $\upc E_y$ such that $\varphi (y) \cll \upc E_y
  \cll y$.  Similarly, let $\upc E'_{y'}$ be a finitary compact such
  that $\upc E'_{y'} \cll y'$ and $\varphi' (y') \cll \upc E'_{y'}$
  for each $y' \in E'$.

  Consider the finite collection of all pairs $(\upc E_y, y)$,
  $(\varphi (y), z)$, $(\upc E'_{y'}, y')$, and $(\varphi' (y'), z')$,
  where $y \in E$, $z \in E_y$, $y' \in E'$, $z' \in E_{y'}$.  Since
  $G$ is qfs, there is a $\psi \in G$ such that $\upc E'' \supseteq
  \psi (x) \supseteq \upc x$ for all the above pairs $(E'', x)$.  In
  particular, looking at the pair $(\upc E_y, y)$, we get: $(a)$ $\upc
  E_y \supseteq \psi (y)$ for every $y \in E$.  And looking at the
  pair $(\varphi (y), z)$, $\varphi (y) \supseteq \psi (z)$ for all $y
  \in E$, $z \in E_y$.  So $\varphi (y) \supseteq \bigcup_{z \in E_y}
  \psi (z) = \bigcup_{z \in \upc E_y} \psi (z) = \psi^\dagger (\upc
  E_y)$.  We have proved: $(b)$ $\varphi (y) \supseteq \psi^\dagger
  (\upc E_y)$ for every $y \in E$.  Then, for every $x \in X$,
  $(\varphi^\dagger \circ \varphi) (x) = \bigcup_{y \in \varphi (x)}
  \varphi (y) \supseteq 
  \bigcup_{y \in \varphi (x)} \psi^\dagger (\upc E_y)$ (by $(b)$)
  $\supseteq 
  \bigcup_{y \in \varphi (x)} (\psi^\dagger \circ \psi) (y)$ (by
  $(a)$) $= (\psi^\dagger \circ \psi)^\dagger (\varphi (x)) \supseteq
  (\psi^\dagger \circ \psi)^\dagger (\upc x)$ (since $\varphi (x)
  \supseteq \upc x$) $= (\psi^\dagger \circ \psi)^\dagger (\eta_X (x))
  = (\psi^\dagger \circ \psi) (x)$ (by one of the monad laws).  So
  $\varphi^\dagger \circ \varphi$ is below $\psi^\dagger \circ \psi$.
  Similarly, ${\varphi'}^\dagger \circ {\varphi'}$ is below
  $\psi^\dagger \circ \psi$, so $H$ is directed.

  Finally, we claim that $\bigcap_{\varphi \in G} (\varphi^\dagger
  \circ \varphi) (x) = \upc x$.  In the $\supseteq$ direction, this is
  because $\varphi^\dagger \circ \varphi$ is a quasi-retraction.
  Conversely, let $\upc E \in \Fin (X)$ be such that $\upc E \cll x$.
  By interpolation, find $\upc E' \in \Fin (X)$ such that $\upc E \cll
  \upc E' \cll x$.  Since $G$ is qfs, applied to the pairs $(\upc E',
  x)$ and $(\upc E, y)$ for each $y \in E'$, there is an element
  $\varphi \in G$ such that $\upc E' \supseteq \varphi (x)$ and $\upc
  E \supseteq \varphi (y)$ for every $y \in E'$.  So $\upc E \supseteq
  \varphi^\dagger (\upc E') \supseteq (\varphi^\dagger \circ \varphi)
  (x)$.  So $\bigcap_{\varphi \in G} (\varphi^\dagger \circ \varphi)
  (x) \subseteq \bigcap_{\upc E \in \Fin (X),\: \upc E \cll
    x}^\downarrow \upc E = \upc x$, as $X$ is quasi-continuous.
\end{proof}

\begin{thm}
  \label{thm:qrb:bilimit}
  Any ($\omega$-)bilimit of ($\omega$)$\QRB$-domains is an
  ($\omega$)$\QRB$-domain.
\end{thm}
\begin{proof}
  Let ${(X_i)}_{i \in I}$ be an expanding system of $\QRB$-domains,
  with projections $p_{ij} : X_j \to X_i$ and embeddings $e_{ij} : X_i
  \to X_j$, $i\leq j$.  Let $X = \lim_{i \in I} X_i$.  There is a
  projection $p_i : X \to X_i$, given by $p_i (\vec x) = x_i$ (where
  $\vec x = {(x_i)}_{i \in I}$), and an embedding $e_i : X_i \to X$
  for every $i \in I$.

  We observe that: $(a)$ if $\upc E \cll p_i (\vec x)$ in $X_i$, then
  $\Smyth e_{ij} (\upc E) \cll p_j (\vec x)$ for every $j \geq i$.
  Indeed, consider any directed family ${(y_k)}_{k \in K}$ such that
  $p_j (\vec x) \leq \sup_{k \in K} y_k$.  Then $p_i (\vec x) = p_{ij}
  (p_j (\vec x)) \leq \sup_{k \in K} p_{ij} (y_k)$, so for some $k \in
  K$, there is a $z \in E$ with $z \leq p_{ij} (y_k)$.  Then $e_{ij}
  (z) \leq e_{ij} (p_{ij} (y_k)) \leq y_k$.  We conclude since $e_{ij}
  (z) \in \Smyth e_{ij} (\upc E)$.

  We now claim that the family $\mathcal D_{\vec x}$ of all finitary
  compacts of the form $\Smyth e_i (\upc E)$, where $\upc E \in \Fin
  (X_i)$ and $\upc E \cll p_i (\vec x)$, $i \in I$, is directed.
  Given $\Smyth e_i (\upc E)$ and $\Smyth e_j (\upc E')$ in $\mathcal
  D_{\vec x}$, find some $k \in I$ such that $i, j \leq k$, by
  directedness.  Then $\Smyth e_i (\upc E) = \Smyth e_k (\Smyth e_{ik}
  (\upc E))$, and by $(a)$ $\Smyth e_{ik} (\upc E) \cll p_k (\vec x)$,
  and similarly $\Smyth e_j (\upc E') = \Smyth e_k (\Smyth e_{jk}
  (\upc E'))$, with $\Smyth e_{jk} (\upc E') \cll p_k (\vec x)$.
  Replacing $i$ by $k$, $\upc E$ by the finitary compact $\Smyth
  e_{ik} (\upc E)$, $j$ by $k$, and $\upc E'$ by $\Smyth e_{jk} (\upc
  E')$ if necessary, we can therefore simply assume that $i=j$.  Since
  $X_i$ is quasi-continuous, there is an $E'' \in \Fin (X_i)$ such
  that $\upc E, \upc E' \cll \upc E'' \cll p_i (\vec x)$, and then
  $\Smyth e_i (\upc E'')$ is an element of $\mathcal D_x$ above both
  $\Smyth e_i (\upc E)$ and $\Smyth e_i (\upc E')$.

  Moreover, we claim that $\bigcap_{\Smyth e_i (\upc E) \in \mathcal
    D_{\vec x}} \Smyth e_i (\upc E)$ equals $\upc \vec x$.  That it
  contains $\vec x$ is obvious: whenever $\upc E \cll p_i (\vec x)$,
  pick $z \in E$ with $z \leq p_i (\vec x)$, so that $e_i (z) \leq e_i
  (p_i (\vec x)) \leq \vec x$, hence $\vec x \in \Smyth e_i (\upc E)$.
  Conversely, every $\vec z \in \bigcap_{\Smyth e_i (\upc E) \in
    \mathcal D_{\vec x}} \Smyth e_i (\upc E)$ must be such that $z_i =
  p_i (\vec z) \in \Smyth p_i (\bigcap_{\upc E \cll p_i (\vec x)}
  \Smyth e_i (\upc E)) \subseteq \bigcap_{\upc E \cll p_i (\vec x)}
  \Smyth p_i (\Smyth e_i (\upc E)) = \bigcap_{\upc E \cll p_i (\vec
    x)} \upc E = \upc p_i (\vec x) = \upc x_i$, since $X_i$ is
  quasi-continuous.  As this holds for every $i$, $\vec x \leq \vec
  z$.  So $\bigcap_{\Smyth e_i (\upc E) \in \mathcal D_{\vec x}}
  \Smyth e_i (\upc E) \subseteq \upc \vec x$.

  In particular, $X$ is a quasi-continuous dcpo.

  We check that the set of quasi-deflations on $X$ is qfs.  Consider a
  finite collection of pairs $(\upc \vec D_k, \vec x_k) \in \Fin (X)
  \times X$ with $\upc \vec D_k \cll \vec x_k$, $1\leq k\leq n$.
  Recall that $\upc \vec D_k \cll \vec x_k$ can be rephrased
  equivalently as: $\vec x_k$ is in the open subset $\uuarrow \vec
  D_k$.
  Since $\bigcap_{\Smyth e_i (\upc E) \in \mathcal D_{\vec x_k}}
  \Smyth e_i (\upc E) = \upc \vec x_k$, by
  Proposition~\ref{prop:heckmann}, for each $k$, pick $\Smyth e_i
  (\upc E_k) \in \mathcal D_{\vec x_k}$ included in $\uuarrow \vec
  D_k$, in particular above $\upc \vec D_k$.  I.e., pick $i \in I$ and
  $\upc E_k \in \Fin (X_i)$ such that $\upc E_k \cll p_i (\vec x_k)$,
  and such that $\upc \vec D_k \supseteq \Smyth e_i (\upc E_k)$.
  (We can pick the same $i$ for every $k$, by directedness, as above.)
  Since $X_i$ is a $\QRB$-domain, and $\upc E_k \cll p_i (\vec x_k)$,
  using Proposition~\ref{prop:heckmann}, there is a quasi-deflation
  $\varphi$ on $X_i$ such that $\varphi (p_i (\vec x_k)) \subseteq
  \uuarrow E_k$.  So $\varphi (p_i (\vec x_k)) \subseteq \upc E_k$,
  for every $k$, $1\leq k\leq n$.  Consider $\psi : X \to \Fin (X)$
  defined as $\Smyth e_i \circ \varphi \circ p_i$.  $\Smyth e_i$,
  restricted to $\Fin (X_i)$, takes its values in $\Fin (X)$, using
  Lemma~\ref{lemma:dagger:new} and the fact that $\Smyth e_i =
  {(\eta_X \circ e_i)}^\dagger$.  Moreover, $\psi$ is continuous from
  $X$ to $\Fin_\V (X)$, hence to $\Fin_\sigma (X)$ since $X$ is
  quasi-continuous, by Corollary~\ref{corl:BoxU}.  For every $\vec x
  \in X$, $p_i (\vec x) \in \varphi (p_i (\vec x))$, since $\varphi$
  is a quasi-deflation.  Then $e_i (p_i (\vec x))$ is below $\vec x$,
  and is in $\psi (\vec x)$, so $\vec x \in \psi (\vec x)$.  So $\psi$
  is a quasi-deflation.

  Moreover, by construction, for each $k$, $1\leq k\leq n$, $\varphi
  (p_i (\vec x_k)) \subseteq \upc E_k$, so $\psi (\vec x_k) \subseteq
  \Smyth e_i (\upc E_k)$, so $\psi (\vec x_k) \subseteq \upc \vec
  D_k$, since $\upc \vec D_k \supseteq \Smyth e_i (\upc E_k)$.  So the set
  of quasi-deflations on $X$ is qfs.

  By Proposition~\ref{prop:qfs}, $X$ is then a $\QRB$-domain.

  To deal with $\omega$-bilimits of $\omega\QRB$-domains, observe that
  any bilimit of a countable expanding system (in particular, an
  expanding sequence) of countably-based quasi-continuous dcpos is
  countably-based.  Indeed, a countably based quasi-continuous dcpo
  $X_i$ has a countable base of sets of the form $\uuarrow E_{ik}$,
  $\upc E_{ik} \in \Fin (X_i)$, $k \in \nat$.  The $\mathcal D_{\vec
    x}$ construction above, suitably modified, shows that the sets
  $\uuarrow \vec {E'}_{ik}$, where $\upc \vec {E'}_{ik} = \Smyth e_i
  (E_{ik})$, $i, k \in \nat$, form a, necessarily countable, base of
  the topology on $X$.  By Proposition~\ref{prop:omega:qrb:2}, $X$ is
  an $\omega\QRB$-domain.
\end{proof}

\section{The Probabilistic Powerdomain}
\label{sec:qretr:V}

Let $X$ be a fixed topological space, and let $\Open (X)$ be the
lattice of open subsets of $X$.
A {\em continuous valuation\/} $\nu$ on $X$ \cite{JP:proba} is a map
from $\Open (X)$ to $\real^+$ such that $\nu (\emptyset)=0$, which is
{\em monotonic\/} ($\nu (U) \leq \nu (V)$ whenever $U \subseteq V$),
{\em modular\/} ($\nu (U \cup V) + \nu (U \cap V) = \nu (U) + \nu (V)$
for all opens $U, V$), and {\em continuous\/} ($\nu (\bigcup_{i \in
  I}^\uparrow U_i)= \sup_{i \in I} \nu (U_i)$ for every directed
family ${(U_i)}_{i \in I}$ of opens).  A {\em (sub)probability\/}
valuation $\nu$ is additionally such that $\nu$ is {\em
  (sub)normalized\/}, i.e., that $\nu (X)=1$ ($\nu (X) \leq 1$).  Let
$\Val_1 (X)$ ($\Val_{\leq 1} (X)$) be the dcpo of all (sub)probability
valuations on $X$, ordered pointwise, i.e., $\nu \leq \nu'$ iff $\nu
(U) \leq \nu' (U)$ for every open $U$.  $\Val_1$ ($\Val_{\leq 1}$)
defines a endofunctor on the category of dcpos, and its action is
defined on morphisms $f$ by $\Val_1 f (\nu) (U) = \nu (f^{-1} (U))$.

We write $\delta_x$ for the \emph{Dirac valuation} at $x$, a.k.a., the
point mass at $x$.  This is the continuous valuation such that
$\delta_x (U)=1$ if $x \in U$, $\delta_x (U) = 0$ otherwise.

The probabilistic powerdomain construction $\Val_1$ is an elusive one,
and natural intuitions are often wrong.  For example, one might
imagine that if $X$ has all binary least upper bounds, then so has
$\Val_1 (X)$.  This was dispelled by Jones and Plotkin
\cite{JP:proba}.  Consider $X = \{\bot, a, b, \top\}$, with $a$ and
$b$ incomparable, $\bot$ below every element and $\top$ above every
element (see Figure~\ref{fig:v3}, right).  Then the upper bounds of
$\frac 1 2 \delta_\bot + \frac 1 2 \delta_a$ and $\frac 1 2
\delta_\bot + \frac 1 2 \delta_b$ in $\Val_1 (X)$ are the probability
valuations of the form $(1 - \alpha_a - \alpha_b - \alpha_\top)
\delta_\bot + \alpha_a \delta_a + \alpha_b \delta_b + \alpha_\top
\delta_\top$ where $\alpha_a + \alpha_\top \geq \frac 1 2$, $\alpha_b
+ \alpha_\top \geq \frac 1 2$, and $\alpha_a + \alpha_b + \alpha_\top
\leq 1$.  The minimal upper bounds are those of the form $\alpha
\delta_\bot + (\frac 1 2 - \alpha) \delta_a + (\frac 1 2 - \alpha)
\delta_b + \alpha \delta_\top$, $\alpha \in [0,1]$.  So there is no
unique least upper bound; in fact, there are uncountably many of them,
even on this small example.

\begin{figure}
  \centering
  \ifpdf
  \input{V-RB.pdftex_t}
  \else
  \input{V-RB.pstex_t}
  \fi
  \caption{Discretizations of $\Val_1 (X)$, $X = \{\bot, a, b, \top\}$}
  \label{fig:V-RB}
\end{figure}

It is unknown whether $\Val_1 (X)$, with $X = \{\bot, a, b, \top\}$ is
an $\RB$-domain, although it is an $\FS$-domain, as a consequence of
\cite[Theorem~17]{JT:troublesome}.  Again, some of the most natural
ideas one can have about $\Val_1 (X)$ are flawed.  It seems obvious
indeed that $\Val_1 (X)$ should be the bilimit of the sequence of
finite posets $\Val_1^{\frac 1 n} (X)$, defined as those probability
valuations $(1 - \alpha_a - \alpha_b - \alpha_\top) \delta_\bot +
\alpha_a \delta_a + \alpha_b \delta_b + \alpha_\top \delta_\top$ where
$\alpha_a$, $\alpha_b$, $\alpha_\top$ are integer multiples of $\frac
1 n$.  See Figure~\ref{fig:V-RB} for Hasse diagrams of a few of these
posets, for $n$ small.

That $\Val_1 (X)$ is such a bilimit is necessarily wrong, because any
bilimit of finite posets is an $\omega\B$-domain, hence is algebraic,
but $\Val_1 (X)$ is not algebraic, since no element except
$\delta_\bot$ is finite.

However, one may imagine to define (non-idempotent) deflations $f_n$
on $\Val_1 (X)$ directly, which would send $\nu \in \Val_1 (X)$ to
some discretized probability valuation in $\Val_1^{\frac 1 n} (X)$.
However, all known attempts fail.  A careful study of
\cite{JT:troublesome} will make this precise.  Let us only note that
if we decide to define $f_n (\nu)$ through its values on open sets,
typically letting $f_n (\nu) (U)$ be the largest integer multiple of
$\frac 1 n$ that is zero-or-strictly-below $\nu (U)$, we obtain a set
function that is not modular.  If we decide to define $f_n (\sum_{x
  \in X} \alpha_x \delta_x)$ as $\sum_{x \in X} \beta_x \delta_x$
where for each $x \neq \bot$ $\beta_x$ is the largest integer multiple
of $\frac 1 n$ that is zero-or-strictly-below $\alpha_x$, then $f_n$
is not monotonic.  If we decide to define $f_n (\nu)$ as the largest
probability valuation way-below $\nu$ in $\Val_1^{\frac 1 n} (X)$, we
run into the problem that there is no {\em unique\/} such largest
probability valuation.  For example, $\nu = \frac 1 3 \delta_a + \frac
1 3 \delta_b + \frac 1 3 \delta_\top$ admits four largest probability
valuations in $\Val_1^{\frac 1 3} (X)$ way-below it: $\frac 1 3
\delta_\bot + \frac 2 3 \delta_a$, $\frac 1 3 \delta_\bot + \frac 1 3
\delta_a + \frac 1 3 \delta_b$, $\frac 2 3 \delta_\bot + \frac 1 3
\delta_\top$, and $\frac 1 3 \delta_\bot + \frac 2 3 \delta_b$, see
Figure~\ref{fig:V-approx}.

\begin{figure}
  \centering
  \ifpdf
  \input{V-approx.pdftex_t}
  \else
  \input{V-approx.pstex_t}
  \fi
  \caption{Largest discretizations below $\nu$ fail to be unique}
  \label{fig:V-approx}
\end{figure}


Observe that the number of largest discretizations of $\nu$ in
$\Val_1^{\frac 1 n} (X)$ is always finite, provided $X$ is finite.
This was our original intuition that replacing deflations by
quasi-deflations, hence moving from $\RB$-domains to $\QRB$-domains,
might provide a nice enough category of domains that would be stable
under the probabilistic powerdomain functor $\Val_1$.  However,
defining quasi-deflations directly, as hinted above, does not work
either: monotonicity fails again.  This is where the characterization
of $\QRB$-domains as quasi-retracts of bifinite domains (up to details
we have already mentioned) will be decisive.

If $Y$ is a retract of $X$, then $\Val_1 (Y)$ is easily seen to be a
retract of $\Val_1 (X)$, using the $\Val_1$ endofunctor.  We wish to
show a similar result for quasi-retracts.  We have not managed to do
so.  Instead we shall rely on the stronger assumptions that $X$ is
stably compact, that $Y$ is a quasi-projection of $X$, not just a
quasi-retract (i.e., the image of $X$ under a proper map).

Moreover, we shall need to replace the Scott topology on $\Val_1 (X)$
by the {\em weak topology\/}, which is the smallest one containing the
subbasic opens $[U > a]$, defined as $\{\nu \in \Val_1 (X) \mid \nu
(U) > a\}$, for each open subset $U$ of $X$ and $a \in \real$.  When
$X$ is a continuous pointed dcpo, the \emph{Kirch-Tix Theorem} states
that it coincides with the Scott topology (see \cite{AMJK:scs:prob},
who attribute it to Tix \cite[Satz~4.10]{Tix:bewertung}, who in turn
attributes it to Kirch \cite[Satz~8.6]{Kirch:bewertung}).

However, the weak topology is better behaved in the general case.  For
example, writing $\creal$ for $\real^+ \cup \{+\infty\}$ with the
Scott topology, and $[X \to \creal]_{\texttt{i}}$ for the space of all
continuous maps from $X$ to $\creal$ with the Isbell topology, there
is a natural homeomorphism between the space of linear continuous maps
from $[X \to \creal]_{\texttt{i}}$ to $\creal$ and the space of of
(extended, i.e., possibly taking the value $+\infty$) continuous
valuations on $X$, with the weak topology
\cite[Theorem~8.1]{Heckmann:space:val}.  This is an analog of the
Riesz Representation Theorem in measure theory, of which one can find
variants in \cite{Tix:bewertung,Gou-csl07} among others, and which we
shall use silently in the proof of Theorem~\ref{thm:qretr:V}.  Let
$\Val_{1\;wk} (X)$ be $\Val_1 (X)$ with its weak topology.

$\Val_{1\;wk}$ defines an endofunctor on the category of topological
spaces, by $\Val_{1\;wk} (f) (\nu) (V) = \nu (f^{-1} (V))$, where $f :
X \to Y$, $\nu \in \Val_{1\;wk} (X)$, and $V \in \Open (Y)$.  That
$\Val_{1\;wk} (f)$ is continuous for every continuous $f$, in
particular, is obvious, since for every open subset $V$ of $Y$,
$\Val_{1\;wk} (f)^{-1} [V > a] = [f^{-1} (V) > a]$.

As we have said above, we shall also require $X$ to be stably compact.
If this is so, then the \emph{cocompact topology} on $X$ consists of
all complements of compact saturated subsets.  Write $X^\dG$, the
\emph{de Groot dual} of $X$, for $X$ with its cocompact topology.
Then $X^\dG$ is again stably compact, and $X^{\dG\dG} = X$ (see
\cite[Corollary~12]{AMJK:scs:prob} or
\cite[Corollary~VI-6.19]{GHKLMS:contlatt}).  The patch topology on
$X$, mentioned earlier, is nothing else than the join of the two
topologies of $X$ and $X^\dG$.

Write $X^\patch$ for $X$ equipped with its patch topology.  If $X$ is
stably compact, then $X^\patch$ is not only compact Hausdorff, but the
graph of the specialization preorder $\leq$ of $X$ is closed in
$X^\patch$: one says that $(X^\patch, \leq)$ is a \emph{compact
  pospace}.  The study of compact pospaces originates in Nachbin's
classic work \cite{Nachbin:toporder}.  Conversely, given a compact
pospace $(Z, \preceq)$, i.e., a compact space with a closed ordering
$\preceq$ on it, the \emph{upwards topology} on $Z$ consists of those
open subsets of $Z$ that are upward closed in $\preceq$.  The space
$Z^\uparrow$, obtained as $Z$ with the upwards topology, is then
stably compact.  Moreover, the two constructions are inverse of each
other.  (See \cite[Section~VI-6]{GHKLMS:contlatt}.)

If $X$ and $Y$ are stably compact, then $f : X \to Y$ is proper if and
only if $f : X^\patch \to Y^\patch$ is continuous, and monotonic with
respect to the specialization orderings of $X$ and $Y$
\cite[Proposition~VI.6.23]{GHKLMS:contlatt}, i.e., if and only if $f$
is a morphism of compact pospaces.

Now, the structure of the cocompact topology on $\Val_{1\;wk} (X)$,
when $X$ is stably compact, is as follows.  For every continuous
valuation $\nu$ on $X$, following Tix \cite{Tix:bewertung}, define
$\nu^\dagger (Q)$ as $\inf_{U \in \Open (X), U \supseteq Q} \nu (U)$,
for every compact saturated subset $Q$ of $X$.
Define $\langle Q \geq a \rangle$ as the set of probability valuations
$\nu$ such that $\nu^\dagger (Q) \geq a$.  The sets $\langle Q \geq
a\rangle$ are compact saturated in $\Val_{1\;wk} (X)$, and
Proposition~6.8 of \cite{JGL-mscs09} even states that they form a
subbase of compact saturated subsets.  This means that the complements
of the sets of the form $\langle Q \geq a \rangle$, $Q$ compact
saturated in $X$, $a \in \real$, form a base of the topology of
$\Val_{1\;wk} (X)^\dG$.  A similar claim was already stated in
\cite[last lines]{Jung:scs:prob}.

\begin{lem}
  \label{lemma:Vf:dagger}
  Let $X$, $Y$ be stably compact spaces, and $r$ be a proper
  surjective map from $X$ to $Y$.  Then $\Val_{1\;wk} (r)
  (\nu)^\dagger (Q) = \nu^\dagger (r^{-1} (Q))$, for every compact
  saturated subset $Q$ of $Y$.
\end{lem}
\begin{proof}
  We must show that $\inf_{V \supseteq Q} \nu (r^{-1} (V)) = \inf_{U
    \supseteq r^{-1} (Q)} \nu (U)$, where $V$ ranges over opens in $Y$
  and $U$ over opens in $X$.

  For every open $V$ containing $Q$, $U = r^{-1} (V)$ is an open
  subset of $X$ containing the compact saturated subset $r^{-1} (Q)$,
  so $\inf_{V \supseteq Q} \nu (r^{-1} (V)) \geq \inf_{U \supseteq
    r^{-1} (Q)} \nu (U)$.

  Conversely, for every open $U$ containing $r^{-1} (Q)$, we shall
  build an open subset $V$ containing $Q$ such that $r^{-1} (V)
  \subseteq U$.  This will establish $\inf_{V \supseteq Q} \nu (r^{-1}
  (V)) \leq \inf_{U \supseteq r^{-1} (Q)} \nu (U)$, hence the
  equality.

  Recall from Lemma~\ref{lemma:qproj:unique} that $r$ forms a
  quasi-retraction, with a unique matching quasi-section $\qs : Y \to
  \SV (X)$ such that $x \in \qs (r (x))$ for every $x \in X$, and such
  that $\qs (y) = r^{-1} (\upc y)$ for every $y \in Y$.  We let $V =
  \qs^{-1} (\Box U)$.  Since $r^{-1} (Q) \subseteq U$, $r^{-1} (Q)$ is
  in $\Box U$.  For every $y \in Q$, $\qs (y) = r^{-1} (\upc y)
  \subseteq r^{-1} (Q)$ is then also in $\Box U$, so $y$ is in
  $\qs^{-1} (\Box U) = V$.  So $Q \subseteq V$.  On the other hand,
  for every element $x$ of $r^{-1} (V)$, $r (x)$ is in $V = \qs^{-1}
  (\Box U)$, so $\qs (r (x))$ is in $\Box U$.  Then $x \in \qs (r (x))
  \subseteq U$.  So $r^{-1} (V) \subseteq U$, and we are done.
\end{proof}

Similarly to the formula $\Val_{1\;wk} (f)^{-1} [V > a] = [f^{-1} (V)
> a]$, which allowed us to conclude that $\Val_{1\;wk} (f)$ was
continuous for every continuous $f$, we obtain:
\begin{lem}
  \label{lemma:Vf:dagger:inv}
  Let $X$, $Y$ be stably compact spaces, and $r$ be a proper
  surjective map from $X$ to $Y$.  Then $\Val_{1\;wk} (r)^{-1} \langle
  Q \geq a \rangle = \langle r^{-1} (Q) \geq a \rangle$ for every
  compact saturated subset $Q$ of $Y$, and $a \in \real$.
\end{lem}
\begin{proof}
  Using Lemma~\ref{lemma:Vf:dagger}, $\Val_{1\;wk} (r)^{-1} \langle Q
  \geq a \rangle = \{\nu \in \Val_{1\;wk} (X) \mid \Val_{1\;wk} (r)
  (\nu)^\dagger (Q) \geq a\} = \{\nu \in \Val_{1\;wk} (X) \mid
  \nu^\dagger (r^{-1} (Q)) \geq a\} = \langle r^{-1} (Q) \geq a
  \rangle$.
\end{proof}

\begin{prop}
  \label{prop:Vf:proper}
  Let $X$ be a stably compact space, $Y$ be a $T_0$ space, and $r$ be
  a proper surjective map from $X$ to $Y$.  Then $\Val_{1\;wk} (r)$ is
  a proper map from $\Val_{1\;wk} (X)$ to $\Val_{1\;wk} (X)$.
\end{prop}
\begin{proof}
  First, since $r$ is proper and surjective, $r$ is a quasi-retraction
  (Lemma~\ref{lemma:qproj:unique}), so $Y$ is stably compact by
  Proposition~\ref{prop:qretr:scomp}.  $\Val_{1\;wk} (r)$ is
  continuous from $\Val_{1\;wk} (X)$ to $\Val_{1\;wk} (Y)$.
  Lemma~\ref{lemma:Vf:dagger:inv} implies that $\Val_{1\;wk} (r)$ is
  also continuous from $\Val_{1\;wk} (X)^\patch$ to $\Val_{1\;wk}
  (Y)^\patch$: it suffices to check that the inverse images of
  subbasic patch-open subsets, of the form $[U > a]$ or whose
  complements are of the form $\langle Q \geq a\rangle$, are
  patch-open.  Also, $\Val_{1\;wk} (r)$ is monotonic with respect to
  the specialization orderings of $\Val_{1\;wk} (X)$ and $\Val_{1\;wk}
  (Y)$, being continuous.  So $\Val_{1\;wk} (r)$ is proper.
\end{proof}

Let us establish surjectivity.  One possible proof goes as follows.
Let $\mathcal M_1 (Z)$ denote the space of all Radon probability
measures on the space $Z$.  If $X$ is stably compact, then $\mathcal
M_1 (X^\patch)$ is compact in the vague topology, and forms a compact
pospace with the stochastic ordering, where $\mu$ is below $\mu'$ if
and only if $\mu (U) \leq \mu' (U)$ for every open subset $U$ of $X$
\cite[Theorem~31]{AMJK:scs:prob}.  By
\cite[Theorem~36]{AMJK:scs:prob}, there is an isomorphism between
$\Val_{1\;wk} (X)$ and $\mathcal M_1^\uparrow (X^\patch)$.

Now assume a second stably compact space $Y$.  For two measurable
spaces $A$ and $B$, and $f : A \to B$ measurable, let $\mathcal M (f)$
map the Radon measure $\mu$ to its image measure, whose value on the
Borel subset $E$ of $B$ is $\mu (f^{-1} (E))$.  A standard result
\cite[2.4, Lemma~1]{Bourbaki:int:IX} states that for any two compact
Hausdorff spaces $A$ and $B$, if $r$ is continuous surjective from $A$
to $B$, then $\mathcal M (r)$ is surjective.  The desired result
follows, up to a few technical details, by taking $A=X^\patch$,
$B=Y^\patch$, remembering that since $r$ is proper from $X$ to $Y$, it
is continuous from $X^\patch$ to $Y^\patch$.

Instead of working out the---technically subtle but boring---technical
details, let us give a direct proof, similar to the above cited
Lemma~1, 2.4 \cite{Bourbaki:int:IX}.  Instead of using the Hahn-Banach
Theorem, we rest on the following Keimel Sandwich Theorem
\cite[Theorem~8.2]{Keimel:topcones}: let $C$ be a topological cone, $q
: C \to \creal$ be a continuous superlinear map, $p : C \to \creal$ be
a sublinear map, and assume $q \leq p$; then there is a continuous
linear map $\Lambda : C \to \creal$ such that $q \leq \Lambda \leq p$.
Here, a \emph{cone} is an additive commutative monoid, with a scalar
multiplication by elements of $\real^+$ satisfying $a (x+y)=ax+ay$,
$(a+b)x = ax+bx$, $(ab)x=a(bx)$, $1x=x$, $0x=0$ for all $a, b \in
\real^+$, $x, y \in C$.  A cone is \emph{topological} if and only if
addition and multiplication are continuous.  The continuous maps $f :
C \to \creal$ are sometimes called lower semi-continuous in the
literature.  Such a map is \emph{superlinear} (resp.,
\emph{sublinear}, \emph{linear}) if and only if $f (ax)=af (x)$ for
all $a \in \real^+$, $x \in C$ and $f (x+y) \geq f (x)+f (y)$ for all
$x, y \in C$ (resp., $\leq$, $=$).  It is easy to see that the space
$[X \to \creal]$ of all continuous maps from $X$ to $\creal$, equipped
with the obvious addition and scalar multiplication and with the Scott
topology of the pointwise ordering, is a topological cone.

\begin{prop}
  \label{prop:qretr:surj}
  Let $X$, $Y$ be stably compact spaces, and $r$ be a proper
  surjective map from $X$ to $Y$.  Then $\Val_{1\;wk} (r)$ is
  surjective.
\end{prop}
\begin{proof}
  Fix some continuous probability valuation $\nu$ on $Y$.  Let $C$ be
  $[X \to \creal]$.  Since $r$ is proper, it has an associated
  quasi-section $\qs$, with $x \in \qs (r (x))$ for every $x \in X$,
  by Lemma~\ref{lemma:qproj:unique}.  Define $q : C \to \creal$ by $q
  (h) = \int_{y \in Y} h_* (\qs (y)) d\nu$, where $h_* (Q) = \min_{x
    \in Q} h (x)$, and integration of continuous maps from $Y$ to
  $\creal$ is defined by a Choquet formula
  \cite{Tix:bewertung,JGL-icalp07}, or equivalently by Heckmann's
  general construction \cite{Heckmann:space:val}.

  Note that $h_* (Q)$ is well-defined as $\min \Smyth h (Q)$, since
  $\Smyth h (Q)$ is compact saturated hence of the form $[a, +\infty]$
  for some $a \in \creal$---then $h_* (Q)=a$.  Moreover, $h_*$ is
  continuous from $\SV (X)$ to $\creal$, because $h_*^{-1} (a,
  +\infty] = \Box h^{-1} (a, +\infty]$.  So $h_* \circ \qs$ is
  continuous, whence the integral defining $q$ makes sense.  We now
  claim that the map $h \mapsto h_*$ is (Scott-)continuous.  First, $h
  \mapsto h_*$ is clearly monotonic.  Now let ${(h_i)}_{i \in I}$ be a
  directed family in $[X \to \creal]$ with a least upper bound $h$.
  By monotonicity, for every $Q \in \Smyth (X)$, ${h_i}_* (Q) \leq h_*
  (Q)$, so $\sup_{i \in I} {h_i}_* (Q)$ exists and is below $h_* (Q)$.
  Conversely, we must show that for every $a \in \real^+$ such that $a
  < h_* (Q)$, $a < \sup_{i \in I} {h_i}_* (Q)$.  The elements $Q \in
  \Smyth (X)$ such that $a < h_* (Q)$ are those such that for every $x
  \in Q$, there is an $i \in I$ such that $h_i (x) \in (a, +\infty)$,
  i.e., they are the elements of $\Box \bigcup_{i \in I} h_i^{-1} (a,
  +\infty)$.  Since $\Box$ commutes with directed unions, if $a < h_*
  (Q)$ then $Q \in \Box h_i^{-1} (a, +\infty)$ for some $i \in I$,
  i.e., ${h_i}_* (Q) > a$, and we are done.  Since $h \mapsto h^*$ is
  continuous, and since the Choquet integral is Scott-continuous in
  the integrated function (see \cite[Satz~4.4]{Tix:bewertung}, or
  \cite[Theorem~7.1~(3)]{Heckmann:space:val}), we obtain that $q$ is
  (Scott-)continuous.

  For every $a \in \real^+$, $q (ah) = a q(h)$.  Moreover, since
  $(h_1+h_2)_* \geq {h_1}_* + {h_2}_*$, and integration is linear, $q$
  is superlinear.

  Define $p (h)$ as $\inf \left\{\int_{y \in Y} h' (y) d\nu \mathbin{\Big|} h'
    \in [Y \to \creal],\: h \leq h' \circ r\right\}$.  Clearly, $p$ is
  sublinear.  Notably,
  \[
  \begin{array}{rclcl}
    p (h_1) + p (h_2) & = &
    \inf
    \Big\{\int_{y \in Y} \left[h'_1 (y) + h'_2 (y)\right] d\nu
    & \mathbin{\Big|}{} &
    h'_1 \in [Y \to \creal], \: h_1 \leq h'_1 \circ r, \\
    &&&&
    h'_2 \in [Y \to \creal], \: h_2 \leq h'_2 \circ r\Big\} \\
    & \geq &
    \multicolumn{3}{l}{
      \inf
      \left\{\int_{y \in Y} h' (y)\: d\nu \mathbin{\Big|}
        h' \in [Y \to \creal],\: h_1 + h_2 \leq h' \circ r
      \right\}
      = p (h_1 + h_2).
    }
  \end{array}
  \]
  Whenever $h \leq h' \circ r$, we claim that $h_* (\qs (y)) \leq h'
  (y)$.  Indeed, since $y = r (x)$ for some $x \in X$
  (Lemma~\ref{lemma:qretr:surj}), and since $x \in \qs (r (x)) = \qs
  (y)$, $h_* (\qs (y)) \leq h (x) \leq h' (r (x)) = h' (y)$.

  It follows that $q (h) = \int_{y \in Y} h_* (\qs (y)) d\nu \leq
  \int_{y \in Y} h' (y) d\nu$.  By taking infs over $h'$, $q \leq p$.

  So Keimel's Sandwich Theorem applies.  There is a continuous linear
  map $\Lambda : C \to \creal$ such that $q \leq \Lambda \leq p$.
  Define $\nu_0 : \Open (X) \to \creal$ by $\nu_0 (U) = \Lambda
  (\chi_U)$, where $\chi_U$ is the characteristic function of $U$.
  Then $\nu_0$ is a continuous valuation on $X$; in particular, $\nu_0
  (U \cup V) + \nu_0 (U \cap V) = \nu_0 (U) + \nu_0 (V)$ because
  $\chi_{U \cup V} + \chi_{U \cap V} = \chi_U + \chi_V$.

  Now, given an open subset $V$ of $Y$, take $h = \chi_{r^{-1} (V)}$.
  Then $h_* (Q) = 1$ iff $Q \subseteq r^{-1} (V)$, so $h_* =
  \chi_{\Box r^{-1} (V)}$, and therefore $h_* (\qs (y)) = \chi_{\Box
    r^{-1} (V)} (r^{-1} (\upc y)) = \chi_V (y)$, using the fact that
  $r$ is surjective.  It follows that $q (h) = \int_{y \in Y} \chi_V
  d\nu = \nu (V)$.  On the other hand, take $h' = \chi_V$ in the
  definition of $p$, and check that $h \leq h' \circ r$.  It follows
  that $p (h) \leq \int_{y \in Y} \chi_V (y) d\nu = \nu (V)$.  Since
  $q (h) \leq \nu_0 (r^{-1} (V)) \leq p (h)$, $\nu_0 (r^{-1} (V)) =
  \nu (V)$.  This holds for every open subset $V$ of $Y$.  In
  particular, taking $V = Y$, we obtain that $\nu_0$ is a probability
  valuation: $\nu_0 (X) = \nu_0 (r^{-1} (Y)) = \nu (Y) = 1$.  And
  finally, that $\nu_0 (r^{-1} (V)) = \nu (V)$ holds for every open
  $V$ of $Y$ means that $\nu = \Val_{1\;wk} (\nu_0)$.
\end{proof}

Putting together Proposition~\ref{prop:Vf:proper} and
Proposition~\ref{prop:qretr:surj}, we obtain:
\begin{thm}[Key Claim]
  \label{thm:qretr:V}
  Let $X$ be a stably compact space, and $Y$ be a $T_0$ space.  If $r$
  is a proper surjective map from $X$ to $Y$, then $\Val_{1\;wk} (r)$
  is a proper surjective map from $\Val_{1\;wk} (X)$ to $\Val_{1\;wk}
  (X)$.
\end{thm}
In particular, if $Y$ is a quasi-projection of $X$, then $\Val_{1\;wk}
(Y)$ is a quasi-projection of $\Val_{1\;wk} (X)$.

\begin{figure}
  \centering
  \ifpdf
  \input{ex1-path.pdftex_t}
  \else
  \input{ex1-path.pstex_t}
  \fi
  \caption{The path space of Figure~\ref{fig:ex1}~$(i)$}
  \label{fig:ex1-path}
\end{figure}

We shall apply this theorem twice, and first, to finite pointed
posets.  Let $<$ be the strict part of $\leq$.
\begin{defi}[Path Space]
  \label{defn:path}
  Let $Y$ be any finite pointed poset.  Write $y \to y'$ iff %
  $y$ is {\em immediately below\/} $y'$, i.e., $y < y'$ and there is
  no $z \in Y$ such that $y < z < y'$.  A {\em path\/} $\pi$ in $Y$ is
  any set $\{y_0, y_1, \ldots, y_n\} \subseteq Y$ with $y_0 = \bot \to
  y_1 \to \ldots \to y_n$.  The {\em path space\/} $\Pi (Y)$ is the
  set of paths in $Y$, ordered by $\subseteq$.
\end{defi}
Alternatively, the ordering on paths $y_0 \to y_1 \to \ldots \to y_n$
is the prefix ordering on sequences $y_0 y_1 \ldots y_n$.

Note that $\Pi (Y)$ is always a finite tree, i.e., a finite pointed
poset such that the downward closure of a point is always totally
ordered.  Up to questions of finiteness, this is exactly how we built
a tree from an ordering in the proof of Lemma~\ref{lemma:qs:nonempty},
by the way.

We observe that every finite pointed poset $Y$ is a quasi-projection
of its path space $\Pi (Y)$.
\begin{lem}
  \label{lemma:path:qretr}
  For every finite pointed poset $Y$, the map $r : \Pi (Y) \to Y$
  defined by $r (\pi) = \max \pi$ is proper and surjective.
\end{lem}
\begin{proof}
  See Figure~\ref{fig:ex1-path}, which displays the path space of the
  space $Y$ of Figure~\ref{fig:ex1}~$(i)$.  Each gray region is
  labeled with an element from $Y$, which is the image by $r$ of every
  point in the region; e.g., the top right, $5$-element region is
  mapped to $j$ in $Y$.

  Formally, let $X = \Pi (Y)$, and define $r : X \to Y$ by $r (\pi) =
  \max \pi$, i.e., $r (y_0\to y_1 \to \ldots \to y_n) = y_n$.  The map
  $r$ is surjective, and monotonic.  Since $X$ and $Y$ are finite, $r$ is
  then trivially proper.
\end{proof}
$Y$ is certainly not a retract of $\Pi (Y)$ in general: it is, if and
only if $Y$ is a tree.  Indeed, if $Y$ is a tree, then $Y$ is
isomorphic to $\Pi (Y)$, and conversely, every retract of a tree is a
tree.

Finite trees are very special.  Jung and Tix proved that $\Val_{\leq
  1} (T)$ is an $\RB$-domain \cite[Theorem~13]{JT:troublesome} for
every finite tree $T$.  They noted (comment after op.cit.) that
$\Val_{\leq 1} (T)$ is even a bc-domain in this case, i.e., a pointed
continuous dcpo in which every pair of elements with an upper bound
has a least upper bound.  It is well-known that every bc-domain is an
$\RB$-domain: given any finite subset $A$ of a basis $B$ of a
bc-domain $X$, the map $f_A (x) = \sup (A \cap \ddarrow x)$ is a
deflation, the family of these deflations is directed, and their least
upper bound is the identity map.
\begin{lem}
  \label{lemma:V1:bc}
  For every finite tree $T$, $\Val_1 (T)$ is a countably-based
  bc-domain.
\end{lem}
\begin{proof}
  Since $T$ is a finite tree, it is trivially a continuous pointed
  dcpo, so $\Val_1 (T)$ is again continuous
  \cite[Section~3]{Edalat:int}.  A basis is given by the valuations of
  the form $\sum_{t \in T} a_t \delta_t$ with $a_t \in [0, 1]$,
  $\sum_{t \in T} a_t=1$ and each $a_t$ rational.  Since $\Val_1 (T)$
  has a countable basis $B$, its topology has a countable base
  consisting of the subsets $\uuarrow b$, $b \in B$.  So $\Val_1 (T)$
  is countably-based.

  To show that $\Val_1 (T)$ is a bc-domain, we observe that every
  probability valuation $\nu$ on $T$ is entirely characterized by the
  values $\nu (\upc t)$, $t \in T$.  Indeed, for every open subset $U$
  of $T$, let $\Min U$ be the (finite) set of minimal elements of $U$;
  the sets $\upc t$, $t \in \Min U$, are pairwise disjoint, so $\nu
  (U) = \sum_{t \in \Min U} \nu (\upc t)$.  The map $f : T \to [0, 1]$
  defined by $f (t) = \nu (\upc t)$ satisfies $f (\bot)=1$ and $f (t)
  \geq \sum_{t' \in T, t \to t'} f (t')$ for every $t \in T$.  Let us
  call such maps {\em admissible\/}.  Given any admissible map $f$,
  there is a unique probability valuation $\nu$ such that $f (t) = \nu
  (\upc t)$ for every $t \in T$, namely $\sum_{t \in T} a_t \delta_t$
  with $a_t = f (t) - \sum_{t' \in T, t \to t'} f (t')$.  So $\Val_1
  (T)$ is order-isomorphic to the poset of admissible maps, with the
  pointwise ordering.  Therefore we only have to show that any two
  admissible maps $f_1$, $f_2$ below a third one $f_0$ have a least
  upper bound $f$.  As a least upper bound, $f (t)$ must be above $f_1
  (t)$, $f_2 (t)$, and $\sum_{t' \in T, t \to t'} f (t')$, so define
  $f (t)$ by descending induction on $t$ by $f (t) = \max (f_1 (t),
  f_2 (t), \sum_{t' \in T, t \to t'} f (t'))$.  (By descending
  induction, we mean induction on the largest length $n$ of a sequence
  $t_0 \to t_1 \to \ldots \to t_n$ in $T$ such that $t_0=t$.)  This is
  admissible if and only if $f (\bot)=1$, and in this case will be the
  least upper bound of $f_1$, $f_2$.  By definition $f (\bot) \geq 1$.
  It is easy to see that $f (t) \leq f_0 (t)$ for every $t$, by
  descending induction on $t$: so $f (\bot) \leq f_0 (\bot)=1$, hence
  $f$ is admissible.
\end{proof}
We retrieve the Jung-Tix result that $\Val_{\leq 1} (T)$ is a
bc-domain for every tree $T$: let $T_\bot$ be $T$ with an extra bottom
element added below all elements of $T$, and apply
Lemma~\ref{lemma:V1:bc} to $\Val_1 (T_\bot) \cong \Val_{\leq 1} (T)$.

\begin{prop}
  \label{prop:path:V}
  For every finite pointed poset $Y$, $\Val_1 (Y)$ is a continuous
  $\omega\QRB$-domain.
\end{prop}
\begin{proof}
  $Y$ is trivially a continuous pointed dcpo.  Then we know that
  $\Val_1 (Y)$ is again continuous \cite[Section~3]{Edalat:int}, and
  that $\Val_1 (Y) = \Val_{1\;wk} (Y)$ by the Kirch-Tix Theorem.
  Similarly for $\Val_1 (\Pi (Y))$.  $\Pi (Y)$ is clearly stably
  compact, since finite.  By Theorem~\ref{thm:qretr:V}, using
  Lemma~\ref{lemma:path:qretr}, $\Val_1 (Y)$ is the image of $\Val_1
  (\Pi (Y))$ under some proper surjective map.  But $\Pi (Y)$ is a tree,
  so $\Val_1 (\Pi (Y))$ is a countably-based bc-domain by
  Lemma~\ref{lemma:V1:bc}, hence a countably-based $\RB$-domain, hence
  an $\omega\QRB$-domain, by Proposition~\ref{prop:FS:QRB} and
  Proposition~\ref{prop:omega:qrb:2}.  By
  Proposition~\ref{prop:qretr:qrb}, $\Val_1 (Y)$ must also be an
  $\omega\QRB$-domain.
\end{proof}
We can finally prove the main theorem of this paper.
\begin{thm}
  \label{thm:qrb:V}
  The probabilistic powerdomain of any $\omega\QRB$-domain is an
  $\omega\QRB$-domain.
\end{thm}
\begin{proof}
  Let $Y$ be an $\omega\QRB$-domain.  By Theorem~\ref{thm:qrb:qretr},
  $Y$ is the image of some $\omega\B$-domain $X = \lim_{i \in \nat}
  X_i$ under some proper surjective map.  Since $\Val_1$ is a locally
  continuous functor on the category of dcpos, (as mentioned in proof
  of \cite[Lemma~11]{JT:troublesome}), $\Val_1 (X)$ is also a bilimit
  of the spaces $\Val_1 (X_i)$, $i \in I$.  Each $\Val_1 (X_i)$ is a
  continuous $\omega\QRB$-domain by Proposition~\ref{prop:path:V},
  hence so is $\Val_1 (X)$, by Theorem~\ref{thm:qrb:bilimit} and since
  bilimits of continuous dcpos are continuous
  \cite[Theorem~3.3.11]{AJ:domains}.

  Since $X$ is bifinite, it is stably compact, (use, e.g.,
  Theorem~\ref{thm:qrb:scomp}), and $\Val_1 (X) = \Val_{1\;wk} (X)$
  because $X$ is continuous and pointed, using the Kirch-Tix Theorem.
  So $\Val_{1\;wk} (Y)$ is the image of $\Val_1 (X)$ under a proper
  surjective map, by Theorem~\ref{thm:qretr:V}.  It is clear that
  $\Val_{1\;wk} (Y)$ is $T_0$, so by Proposition~\ref{prop:qretr:qrb}
  $\Val_{1\;wk} (Y))$ is an $\omega\QRB$-domain in its specialization
  preorder $\preceq$, and its topology must be the Scott topology of
  $\preceq$.

  But it is easy to see that $\preceq$ is the usual ordering on
  $\Val_1 (Y)$, i.e., $\nu \preceq \nu'$ iff $\nu (U) \leq \nu' (U)$
  for every open $U$ of $Y$: note that if $\nu \preceq \nu'$, then
  $\nu' \in [U > r]$ for every $r < \nu (U)$.  So $\Val_{1\;wk} (Y)) =
  \Val_1 (Y)$, and we conclude.
\end{proof}

Using the fact that $\Val_1 (X)$ is continuous whenever $X$ is
continuous and pointed \cite[Section~3]{Edalat:int}, it also follows:
\begin{cor}
  \label{corl:qrb:V}
  The probabilistic powerdomain of any continuous $\omega\QRB$-domain
  (in particular, every $\RB$-domain) is again a continuous
  $\omega\QRB$-domain.  
\end{cor}

\section{Conclusion, Failures and Perspectives}
\label{sec:conc}

We have shown that the category $\omega\QRB$ of $\omega\QRB$-domains
and continuous maps is a category of quasi-continuous, stably compact
dcpos that is closed, not only under finite products, bilimits of
expanding sequences, retracts (and even quasi-retracts), but also
under the probabilistic powerdomain functor $\Val_1$.  It is thus
reasonably well-behaved.

\begin{figure}
  \centering
  \ifpdf
  \input{T.pdftex_t}
  \else
  \input{T.pstex_t}
  \fi
  \caption{Plotkin's Domain $T$}
  \label{fig:T}
\end{figure}

But $\omega\QRB$ is {\em not\/} Cartesian-closed.  Consider the
space~$T$ of \cite[Figure~12]{AJ:domains}, see Figure~\ref{fig:T}.
This is an $\omega\QRB$-domain: define the quasi-deflations
$\varphi_i$, $i \in \nat$, as mapping $\bot$ to $\upc \{\bot\}$, any
element $(j, n)$ to $\upc \{(j,n)\}$ if $n < i$, and any other
element to $\upc \{(0,i), (1,i)\}$.

However, $[T \to T]$ is not an $\omega\QRB$-domain.

Assume ${(\varphi_i)}_{i \in \nat}$ were a generating sequence of
quasi-deflations on $[T \to T]$.  For each function $f : \nat \to
\{0,1\}$, there is a continuous map $\hat f : T \to T$ that sends
$\bot$ to $\bot$, $\top$ to $\top$, $(0,n)$ to $(f (n), n)$ and
$(1,n)$ to $(1-f(n), n)$ ($\hat f$ exchanges $(0,n)$ and $(1,n)$ if
$f(n)=1$, leaves them unswapped if $f (n)=0$).  Write $\varphi_i (\hat
f)$ as $\upc E_{i,f}$, where $E_{i,f}$ is finite.

We claim that: $(*)$ for each $f : \nat \to \{0,1\}$, there is an
index $i \in \nat$ such that $\hat f \in E_{i, f}$.  If there were an
element $g$ of $\varphi_i (\hat f)$ such that $g (0,0) = \bot$, for
infinitely many values of $i \in \nat$, then this would hold for every
$i$; but the map sending $\bot$ and $(0,0)$ to $\bot$, and all other
elements to $\top$ would be in $\bigcap_{i \in \nat} {\varphi_i (\hat
  f)} = \upc \hat f$, which is impossible.  So, for $i$ large enough,
no element $g$ of $\varphi_i (\hat f)$ maps $(0,0)$ to $\bot$.
Similarly, for $i$ large enough, no element $g$ of $\varphi_i (\hat
f)$ maps $(1,0)$ to $\bot$.  Since $\hat f \in \varphi_i (\hat f)$,
for $i$ large enough we find $g \in E_{i,f}$ with $g (0,0) \neq \bot$,
$g (1,0) \neq \bot$, and $g \leq \hat f$.

We check that $g=\hat f$.  First, $g (\bot) \leq \hat f (\bot) = \bot$
so $g (\bot)=\hat f (\bot)$.  Next, $g (0, 0)$ is an element below
$\hat f (0, 0) = (f (0), 0)$ and different from $\bot$, and the only
element satisfying this is $\hat f (0, 0)$.  Similarly, $g (1, 0) =
\hat f (1, 0)$.  By induction on $n \in \nat$, we show that $g (j, n)
= \hat f (j, n)$.  At rank $n+1$, $g (0, n+1)$ is an element below
$\hat f (0, n+1) = (f (n+1), n+1)$ and above both $g (0, n) = \hat f
(0, n) = (f (n), n)$ and $g (1, n) = \hat f (1, n) = (1-f(n), n)$.
The only such element is $(f (n+1), n+1) = \hat f (0, n+1)$.
Similarly, $g (1, n+1) = \hat f(1, n+1)$.  Finally, $g (\top)$ is an
element above all $g (j, n)$, hence must equal $\top = \hat f (\top)$.

Since $g \in E_{i, f}$, and $g = \hat f$, Claim $(*)$ is proved.

However, there are uncountably many functions of the form $\hat f$,
and only countably many elements of $\bigcup_{\substack{i \in \nat\\f
    : \nat \to \{0, 1\}}} E_{i,f}$, since each set $E_{i,f}$ is
finite, and for each $i \in \nat$, there are only finitely many
distinct sets $E_{i, f}$ with $f : \nat \to \{0, 1\}$.  We have
reached a contradiction.

Since exponentials in any full subcategory of the category of dcpos
must be isomorphic to the ordinary continuous function space
\cite{Smyth:CCC}, it follows:
\begin{prop}
  $\omega\QRB$ is not Cartesian-closed.
\end{prop}
The above argument also shows that, although $T$ is both continuous
(even algebraic) and an $\omega\QRB$-domain, $T$ is not an
$\RB$-domain: so Corollary~\ref{corl:qrb:V} is not enough to settle
the Jung-Tix problem in the positive either.

One might hope that countability would be the problem.  However, we
required countability in at least two places.  The first one is
Lemma~\ref{lemma:qs:nonempty}, which would fail in case we allowed for
directed families ${(E^0_i)}_{i \in I}$ instead of non-decreasing
sequences.  The second one is in the $(i) \limp (iii)$ direction of
Theorem~\ref{thm:qrb:qretr}, where we need countability to obtain $Y$
as a quasi-projection, and not just a quasi-retract.  (This is similar
to an open problem in the theory of $\RB$-domains, see \cite[Remark
after Theorem~4.9]{Jung:CCC}.)  In turn, we need quasi-projections,
not just quasi-retracts, in the Key Claim, Theorem~\ref{thm:qretr:V}.

To get around the Jung-Tix problem using our results, one might shift
the focus towards the Kleisli category $\omega\QRB_{\Smyth}$, for
example.  This is a full subcategory of Jung, Kegelmann and Moshier's
pleasing category $\pmb{SCS}^*$ of stably compact spaces and closed
relations \cite{DBLP:journals/entcs/JungKM01}.

On the other hand, we would like to point out the following deep
connection between $\QRB$-domains and $\FS$-domains.
\begin{defi}
  \label{defi:ctrlQRB}
  A {\em controlled quasi-deflation\/} on a poset $X$ is a pair of a
  Scott-continuous map $f : X \to X$ and of a quasi-deflation $\varphi
  : X \to \Fin (X)$ such that $\varphi (x) \subseteq \upc f (x)$ for
  every $x \in X$.  The map $f$ is the {\em control\/}.

  A {\em controlled $\QRB$-domain\/} is a pointed dcpo $X$ with a {\em
    generating family of controlled quasi-deflations\/}, i.e., a
  directed family of controlled quasi-deflations (in the pointwise
  ordering) ${(f_i, \varphi_i)}_{i \in I}$ such that $x = \sup_{i \in
    I} f_i (x)$ for every $x \in X$.
\end{defi}
So a controlled quasi-deflation is a pair $(f, \varphi)$ with the
property that $\upc f (x) \supseteq \varphi (x) \supseteq \upc x$ for
every $x \in X$.  Every controlled $\QRB$-domain is a $\QRB$-domain:
given a generating family of controlled quasi-deflations ${(f_i,
  \varphi_i)}_{i \in I}$, $\bigcap_{i \in I}^\downarrow {\varphi_i (x)}
\subseteq \bigcap_{i \in I}^\downarrow \upc f_i (x) = \upc \sup_{i \in
  I} f_i (x) = \upc x$, and the converse inclusion $\upc x \subseteq
\bigcap_{i \in I}^\downarrow {\varphi_i (x)}$ is obvious; so
${(\varphi_i)}_{i \in I}$ is a generating family of quasi-deflations.

\begin{thm}
  \label{thm:ctrlQRB}
  The controlled $\QRB$-domains are exactly the $\FS$-domains, and
  hence form a Cartesian-closed category.
\end{thm}
\begin{proof}
  If ${(f_i, \varphi_i)}_{i \in I}$ is a generating family of
  controlled quasi-deflations on a pointed dcpo $X$, then $f_i$ is
  finitely separated from $\identity X$: indeed, let $\img \varphi_i =
  \{E_{i1}, \ldots, E_{in_i}\}$ and $M_i = \bigcup_{j=1}^{n_i}
  E_{ij}$, then for every $x \in X$, since $x \in \varphi_i (x)$,
  there is a point $m \in E_{ij}$ where $\upc E_{ij} = \varphi_i (x)$
  such that $m \leq x$, and since $\upc E_{ij} \subseteq \upc f_i
  (x)$, we have $f_i (x) \leq m$; so $M_i$ is a finitely separating
  set for $f_i$ on $X$.

  Conversely, assume $X$ is an $\FS$-domain, and let ${(g_i)}_{i \in
    I}$ be a directed family of continuous maps, finitely separated
  from $\identity X$, and whose pointwise least upper bound is
  $\identity X$.  Let $f_i = g_i \circ g_i$.  By
  \cite[Lemma~2]{Jung:CCC:LICS}, $f_i$ is {\em strongly finitely
    separated\/} from $\identity X$, i.e., there is a finite set $E_i$
  of pairs of elements $m \ll m'$ such that for every $x \in X$, one
  can find a pair $(m, m') \in E_i$ such that $f_i (x) \leq m \ll m'
  \leq x$.  Moreover, the pointwise least upper bound of ${(f_i)}_{i
    \in I}$ is again $\identity X$.

  For each pair $(m, m') \in E_i$ with $m \ll m'$, $\upc m' \subseteq
  \uuarrow m \subseteq \bigcup_{j \in I} f_j^{-1} (\uuarrow m)$.
  Indeed, for every $x \in \uuarrow m$, since $x = \sup_{j \in I} f_j
  (x)$ and $\uuarrow m$ is Scott-open, $f_j (x) \in \uuarrow m$ for
  some $j \in I$.  Since $\upc m'$ is compact, and the family
  ${(f_j^{-1} (\uuarrow m))}_{j \in I}$ is directed, $\upc m'
  \subseteq f_j^{-1} (\uuarrow m)$ for some $j \in I$.  By
  directedness again, we can take the same $j$ for all pairs $m \ll
  m'$ in $E_i$.  But now $\upc m' \subseteq f_j^{-1} (\uuarrow m)$
  implies that whenever $m' \leq x$, then $m \ll f_j (x)$.  Using the
  separation property of $E_i$, for every $x \in X$, one can find a
  pair $(m, m') \in E_i$ such that $f_i (x) \leq m \ll f_j (x)$.  In
  particular, letting $M_i$ be the set of elements $m$ such that $(m,
  m') \in E_i$ for some $m' \in X$, $f_i$ is {\em finitely separated
    from $f_j$\/}, with separating set $M_i$, with the obvious
  meaning: for every $x \in X$, there is an $m \in M_i$ such that $f_i
  (x) \leq m \leq f_j (x)$.  In this case, we write $f_i \prec_{M_i}
  f_j$.

  We now define $\varphi_i (x)$ as $\upc (M_i \cap \upc f_i (x))$.
  Since $M_i$ is finite, $\varphi_i$ is a map from $X$ to $\Fin (X)$.
  It is monotonic, and we claim it is Scott-continuous.  Let
  ${(x_k)}_{k \in K}$ be a directed family in $X$.  Then ${\varphi_i
    (\sup_{k \in K} x_k)} = \upc (M_i \cap \upc \sup_{k \in K} f_i
  (x_k))$ (since $f_i$ is continuous) $= \upc (M_i \cap \bigcap_{k \in
    K} \upc f_i (x_k)) = \upc \bigcap_{k \in K} (M_i \cap \upc f_i
  (x_k))$.  The latter intersection is an intersection of finite sets,
  hence there is a $k \in K$ such that ${\varphi_i (\sup_{k \in K}
    x_k)} = \varphi_i (x_k)$, from which Scott-continuity is
  immediate.

  We must now check that ${(\varphi_i)}_{i \in I}$ is directed.  Given
  $i, i' \in I$, one can find $j, j' \in I$ so that $f_i \prec_{M_i}
  f_j$ and $f_{i'} \prec_{M_{i'}} f_{j'}$.  By directedness, there is
  an $\ell \in I$ such that $f_j, f_{j'} \leq f_\ell$.  We claim that
  for every $x \in X$, $\varphi_i (x), \varphi_{i'} (x) \supseteq
  \varphi_\ell (x)$.  For every $y \in \varphi_\ell (x)$, $f_\ell (x)
  \leq y$.  So $f_j (x) \leq y$.  Since $f_i \prec_{M_i} f_j$, there
  is an element $m \in M_i$ such that $f_i (x) \leq m \leq f_j (x)
  \leq y$.  So $m$ is in $M_i \cap \upc f_i (x)$, and below $y$,
  whence $y \in \varphi_i (x)$.  Similarly, $y$ is in $\varphi_{i'}
  (x)$.

  Finally, $x \in {\varphi_i (x)}$ for every $x \in X$: by definition,
  there is a pair $(m, m') \in E_i$ such that $f_i (x) \leq m \ll m'
  \leq x$; so $m \in M_i$, $m \in \upc f_i (x)$, and $m$ is below $x$.
  And $\bigcap_{i \in I}^\downarrow {\varphi_i (x)} \subseteq
  \bigcap_{i \in I}^\downarrow \upc f_i (x) = \upc x$, while the
  converse inclusion is obvious.  So ${(\varphi_i)}_{i \in I}$ is a
  generating family of quasi-deflations.
\end{proof}

Defining the {\em controlled $\omega\QRB$-domains\/} as the
$\omega\QRB$-domains, except with sequences of controlled
quasi-deflations instead of directed families, and similarly for the
$\omega\FS$-domains (a.k.a., the countably-based $\FS$-domains, again
a Cartesian-closed category \cite[Theorem~11]{Jung:CCC:LICS}), we
prove similarly:
\begin{thm}
  \label{thm:omega:ctrlQRB}
  The controlled $\omega\QRB$-domains are exactly the
  $\omega\FS$-domains, and hence form a Cartesian-closed category.
\end{thm}

Using this last observation, Corollary~\ref{corl:qrb:V} settles half
of the conjecture that the probabilistic powerdomain of an
$\omega\FS$-domain would be an $\omega\FS$-domain again.  We are only
lacking {\em control\/}.

\section*{Open Problems}

\begin{enumerate}[(1)]
\item Is countability necessary in Theorem~\ref{thm:qrb:qretr}?
  Precisely, can one show that the $\QRB$-domains are exactly the
  quasi-retracts of $\B$-domains?  The main difficulty seems to lie in
  the fact that a non-countable analog of
  Lemma~\ref{lemma:qs:nonempty} is missing---and Rudin's Lemma does
  not quite give us what we need, as discussed before the statement
  of the lemma.
\item If $Y$ is a quasi-retract of $X$, $X$ is stably compact, and $Y$
  is $T_0$, then is $\Val_{1\;wk} (Y)$ a quasi-retract of
  $\Val_{1\;wk} (X)$?  This would be the analog of
  Theorem~\ref{thm:qretr:V}, only with quasi-retractions instead of
  quasi-projections.
\item Is stable compactness necessary to derive Theorem~\ref{thm:qretr:V}?
\item One way of trying to prove that the probabilistic powerdomain of
  an $\omega\FS$-domain is again an $\omega\FS$-domain would be by
  inventing a new notion, say of \emph{good maps}, and show that the
  $\omega\FS$-domains, or alternatively the controlled
  $\omega\QRB$-domains, are exactly the images under good maps of
  $\omega\B$-domains.  Good maps should intuitively be intermediate
  between projections and proper surjective maps, in the sense that
  every projection should be good, and that every good map should be
  proper and surjective.  Indeed surjective proper maps preserve the
  $\QRB$ part, but not the control, while projections preserve too
  much, in the sense that not all $\omega\QRB$-domains, only the
  $\omega\RB$-domains, are retracts of $\omega\B$-domains.  Such a
  characterization of $\omega\FS$-domains would be of independent
  interest, too.
\end{enumerate}

\section*{Acknowledgment}

I would like to thank G. Plotkin for discussions, the anonymous
referees at LICS'2010, and the anonymous referees of this paper, who
offered counterexamples to conjectures, and numerous simplifications
and improvements in almost every part of the paper.  I am greatly
indebted to them.

\bibliographystyle{alpha}
\bibliography{qrb}

\end{document}

%% file: v3.pdf_t
\begin{picture}(0,0)%
\includegraphics{v3.pdf}%
\end{picture}%
%
%
\setlength{\unitlength}{1450sp}%
\begingroup\makeatletter\ifx\SetFigFont\undefined%
\gdef\SetFigFont#1#2#3#4#5{%
  \reset@font\fontsize{#1}{#2pt}%
  \fontfamily{#3}\fontseries{#4}\fontshape{#5}%
  \selectfont}%
\fi\endgroup%
\begin{picture}(13106,8580)(259,-11503)
\put(9181,-3976){\makebox(0,0)[lb]{\smash{{\SetFigFont{11}{13.2}{\rmdefault}{\mddefault}{\updefault}{\color[rgb]{0,0,0}Here,}%
}}}}
\put(9181,-7531){\makebox(0,0)[lb]{\smash{{\SetFigFont{11}{13.2}{\rmdefault}{\mddefault}{\updefault}{\color[rgb]{0,0,0}$x \in X$, indicative}%
}}}}
\put(9181,-8071){\makebox(0,0)[lb]{\smash{{\SetFigFont{11}{13.2}{\rmdefault}{\mddefault}{\updefault}{\color[rgb]{0,0,0}of the weight of $x$}%
}}}}
\put(9766,-8611){\makebox(0,0)[lb]{\smash{{\SetFigFont{11}{13.2}{\rmdefault}{\mddefault}{\updefault}{\color[rgb]{0,0,0}$=0$,}%
}}}}
\put(11071,-10456){\makebox(0,0)[lb]{\smash{{\SetFigFont{11}{13.2}{\rmdefault}{\mddefault}{\updefault}{\color[rgb]{0,0,0}$=\frac 1 3 \delta_a + \frac 2 3 \delta_\top$}%
}}}}
\put(9181,-10456){\makebox(0,0)[lb]{\smash{{\SetFigFont{11}{13.2}{\rmdefault}{\mddefault}{\updefault}{\color[rgb]{0,0,0}E.g., }%
}}}}
\put(9181,-8611){\makebox(0,0)[lb]{\smash{{\SetFigFont{11}{13.2}{\rmdefault}{\mddefault}{\updefault}{\color[rgb]{0,0,0}(}%
}}}}
\put(9766,-9286){\makebox(0,0)[lb]{\smash{{\SetFigFont{11}{13.2}{\rmdefault}{\mddefault}{\updefault}{\color[rgb]{0,0,0}$=\frac 2 3$,}%
}}}}
\put(10621,-3976){\makebox(0,0)[lb]{\smash{{\SetFigFont{11}{13.2}{\rmdefault}{\mddefault}{\updefault}{\color[rgb]{0,0,0}$X =$}%
}}}}
\put(11791,-3931){\makebox(0,0)[lb]{\smash{{\SetFigFont{11}{13.2}{\rmdefault}{\mddefault}{\updefault}{\color[rgb]{0,0,0}$a$}%
}}}}
\put(12601,-3931){\makebox(0,0)[lb]{\smash{{\SetFigFont{11}{13.2}{\rmdefault}{\mddefault}{\updefault}{\color[rgb]{0,0,0}$b$}%
}}}}
\put(12151,-3346){\makebox(0,0)[lb]{\smash{{\SetFigFont{11}{13.2}{\rmdefault}{\mddefault}{\updefault}{\color[rgb]{0,0,0}$\top$}%
}}}}
\put(12151,-4381){\makebox(0,0)[lb]{\smash{{\SetFigFont{11}{13.2}{\rmdefault}{\mddefault}{\updefault}{\color[rgb]{0,0,0}$\bot$}%
}}}}
\put(9181,-6451){\makebox(0,0)[lb]{\smash{{\SetFigFont{11}{13.2}{\rmdefault}{\mddefault}{\updefault}{\color[rgb]{0,0,0}is drawn just as $X$ itself,}%
}}}}
\put(11386,-9286){\makebox(0,0)[lb]{\smash{{\SetFigFont{11}{13.2}{\rmdefault}{\mddefault}{\updefault}{\color[rgb]{0,0,0}$=1$).}%
}}}}
\put(11431,-8611){\makebox(0,0)[lb]{\smash{{\SetFigFont{11}{13.2}{\rmdefault}{\mddefault}{\updefault}{\color[rgb]{0,0,0}$=\frac 1 3$,}%
}}}}
\put(9181,-5191){\makebox(0,0)[lb]{\smash{{\SetFigFont{11}{13.2}{\rmdefault}{\mddefault}{\updefault}{\color[rgb]{0,0,0}Legend:}%
}}}}
\put(9181,-5911){\makebox(0,0)[lb]{\smash{{\SetFigFont{11}{13.2}{\rmdefault}{\mddefault}{\updefault}{\color[rgb]{0,0,0}Each valuation}%
}}}}
\put(9181,-6991){\makebox(0,0)[lb]{\smash{{\SetFigFont{11}{13.2}{\rmdefault}{\mddefault}{\updefault}{\color[rgb]{0,0,0}with blobs on each}%
}}}}
\end{picture}%

%% file: v3.pstex_t
\begin{picture}(0,0)%
\includegraphics{v3.pstex}%
\end{picture}%
%
%
\setlength{\unitlength}{1450sp}%
\begingroup\makeatletter\ifx\SetFigFont\undefined%
\gdef\SetFigFont#1#2#3#4#5{%
  \reset@font\fontsize{#1}{#2pt}%
  \fontfamily{#3}\fontseries{#4}\fontshape{#5}%
  \selectfont}%
\fi\endgroup%
\begin{picture}(13106,8580)(259,-11503)
\put(9181,-6451){\makebox(0,0)[lb]{\smash{{\SetFigFont{11}{13.2}{\rmdefault}{\mddefault}{\updefault}{\color[rgb]{0,0,0}is drawn just as $X$ itself,}%
}}}}
\put(9181,-7531){\makebox(0,0)[lb]{\smash{{\SetFigFont{11}{13.2}{\rmdefault}{\mddefault}{\updefault}{\color[rgb]{0,0,0}$x \in X$, indicative}%
}}}}
\put(9181,-8071){\makebox(0,0)[lb]{\smash{{\SetFigFont{11}{13.2}{\rmdefault}{\mddefault}{\updefault}{\color[rgb]{0,0,0}of the weight of $x$}%
}}}}
\put(9766,-8611){\makebox(0,0)[lb]{\smash{{\SetFigFont{11}{13.2}{\rmdefault}{\mddefault}{\updefault}{\color[rgb]{0,0,0}$=0$,}%
}}}}
\put(11071,-10456){\makebox(0,0)[lb]{\smash{{\SetFigFont{11}{13.2}{\rmdefault}{\mddefault}{\updefault}{\color[rgb]{0,0,0}$=\frac 1 3 \delta_a + \frac 2 3 \delta_\top$}%
}}}}
\put(9181,-10456){\makebox(0,0)[lb]{\smash{{\SetFigFont{11}{13.2}{\rmdefault}{\mddefault}{\updefault}{\color[rgb]{0,0,0}E.g., }%
}}}}
\put(9181,-8611){\makebox(0,0)[lb]{\smash{{\SetFigFont{11}{13.2}{\rmdefault}{\mddefault}{\updefault}{\color[rgb]{0,0,0}(}%
}}}}
\put(11881,-3346){\makebox(0,0)[lb]{\smash{{\SetFigFont{11}{13.2}{\rmdefault}{\mddefault}{\updefault}{\color[rgb]{0,0,0}$\top$}%
}}}}
\put(11386,-9286){\makebox(0,0)[lb]{\smash{{\SetFigFont{11}{13.2}{\rmdefault}{\mddefault}{\updefault}{\color[rgb]{0,0,0}$=1$).}%
}}}}
\put(12331,-3931){\makebox(0,0)[lb]{\smash{{\SetFigFont{11}{13.2}{\rmdefault}{\mddefault}{\updefault}{\color[rgb]{0,0,0}$b$}%
}}}}
\put(11431,-8611){\makebox(0,0)[lb]{\smash{{\SetFigFont{11}{13.2}{\rmdefault}{\mddefault}{\updefault}{\color[rgb]{0,0,0}$=\frac 1 3$,}%
}}}}
\put(9181,-5191){\makebox(0,0)[lb]{\smash{{\SetFigFont{11}{13.2}{\rmdefault}{\mddefault}{\updefault}{\color[rgb]{0,0,0}Legend:}%
}}}}
\put(11881,-4381){\makebox(0,0)[lb]{\smash{{\SetFigFont{11}{13.2}{\rmdefault}{\mddefault}{\updefault}{\color[rgb]{0,0,0}$\bot$}%
}}}}
\put(11521,-3931){\makebox(0,0)[lb]{\smash{{\SetFigFont{11}{13.2}{\rmdefault}{\mddefault}{\updefault}{\color[rgb]{0,0,0}$a$}%
}}}}
\put(9766,-9286){\makebox(0,0)[lb]{\smash{{\SetFigFont{11}{13.2}{\rmdefault}{\mddefault}{\updefault}{\color[rgb]{0,0,0}$=\frac 2 3$,}%
}}}}
\put(10351,-3976){\makebox(0,0)[lb]{\smash{{\SetFigFont{11}{13.2}{\rmdefault}{\mddefault}{\updefault}{\color[rgb]{0,0,0}$X =$}%
}}}}
\put(9181,-3976){\makebox(0,0)[lb]{\smash{{\SetFigFont{11}{13.2}{\rmdefault}{\mddefault}{\updefault}{\color[rgb]{0,0,0}Here,}%
}}}}
\put(9181,-5911){\makebox(0,0)[lb]{\smash{{\SetFigFont{11}{13.2}{\rmdefault}{\mddefault}{\updefault}{\color[rgb]{0,0,0}Each valuation}%
}}}}
\put(9181,-6991){\makebox(0,0)[lb]{\smash{{\SetFigFont{11}{13.2}{\rmdefault}{\mddefault}{\updefault}{\color[rgb]{0,0,0}with blobs on each}%
}}}}
\end{picture}%

%% file: ex1.pdftex_t
\begin{picture}(0,0)%
\includegraphics{ex1.pdftex}%
\end{picture}%
\setlength{\unitlength}{2368sp}%
\begingroup\makeatletter\ifx\SetFigFont\undefined%
\gdef\SetFigFont#1#2#3#4#5{%
  \reset@font\fontsize{#1}{#2pt}%
  \fontfamily{#3}\fontseries{#4}\fontshape{#5}%
  \selectfont}%
\fi\endgroup%
\begin{picture}(11076,2847)(480,-4642)
\put(2176,-3361){\makebox(0,0)[lb]{\smash{{\SetFigFont{11}{13.2}{\rmdefault}{\mddefault}{\updefault}{\color[rgb]{0,0,0}$b$}%
}}}}
\put(2701,-3361){\makebox(0,0)[lb]{\smash{{\SetFigFont{11}{13.2}{\rmdefault}{\mddefault}{\updefault}{\color[rgb]{0,0,0}$c$}%
}}}}
\put(2176,-2761){\makebox(0,0)[lb]{\smash{{\SetFigFont{11}{13.2}{\rmdefault}{\mddefault}{\updefault}{\color[rgb]{0,0,0}$e$}%
}}}}
\put(2701,-2761){\makebox(0,0)[lb]{\smash{{\SetFigFont{11}{13.2}{\rmdefault}{\mddefault}{\updefault}{\color[rgb]{0,0,0}$f$}%
}}}}
\put(2401,-2491){\makebox(0,0)[lb]{\smash{{\SetFigFont{11}{13.2}{\rmdefault}{\mddefault}{\updefault}{\color[rgb]{0,0,0}$i$}%
}}}}
\put(2026,-2116){\makebox(0,0)[lb]{\smash{{\SetFigFont{11}{13.2}{\rmdefault}{\mddefault}{\updefault}{\color[rgb]{0,0,0}$j$}%
}}}}
\put(2101,-3886){\makebox(0,0)[lb]{\smash{{\SetFigFont{11}{13.2}{\rmdefault}{\mddefault}{\updefault}{\color[rgb]{0,0,0}$\bot$}%
}}}}
\put(6451,-3307){\makebox(0,0)[lb]{\smash{{\SetFigFont{11}{13.2}{\familydefault}{\mddefault}{\updefault}{\color[rgb]{0,0,0}$(1,2)$}%
}}}}
\put(6451,-3672){\makebox(0,0)[lb]{\smash{{\SetFigFont{11}{13.2}{\familydefault}{\mddefault}{\updefault}{\color[rgb]{0,0,0}$(1,1)$}%
}}}}
\put(6451,-4036){\makebox(0,0)[lb]{\smash{{\SetFigFont{11}{13.2}{\familydefault}{\mddefault}{\updefault}{\color[rgb]{0,0,0}$(1,0)$}%
}}}}
\put(4726,-3307){\makebox(0,0)[lb]{\smash{{\SetFigFont{11}{13.2}{\familydefault}{\mddefault}{\updefault}{\color[rgb]{0,0,0}$(0,2)$}%
}}}}
\put(4726,-3672){\makebox(0,0)[lb]{\smash{{\SetFigFont{11}{13.2}{\familydefault}{\mddefault}{\updefault}{\color[rgb]{0,0,0}$(0,1)$}%
}}}}
\put(4726,-4036){\makebox(0,0)[lb]{\smash{{\SetFigFont{11}{13.2}{\familydefault}{\mddefault}{\updefault}{\color[rgb]{0,0,0}$(0,0)$}%
}}}}
\put(10726,-3307){\makebox(0,0)[lb]{\smash{{\SetFigFont{11}{13.2}{\familydefault}{\mddefault}{\updefault}{\color[rgb]{0,0,0}$(1,2)$}%
}}}}
\put(10726,-3672){\makebox(0,0)[lb]{\smash{{\SetFigFont{11}{13.2}{\familydefault}{\mddefault}{\updefault}{\color[rgb]{0,0,0}$(1,1)$}%
}}}}
\put(10726,-4036){\makebox(0,0)[lb]{\smash{{\SetFigFont{11}{13.2}{\familydefault}{\mddefault}{\updefault}{\color[rgb]{0,0,0}$(1,0)$}%
}}}}
\put(9001,-3307){\makebox(0,0)[lb]{\smash{{\SetFigFont{11}{13.2}{\familydefault}{\mddefault}{\updefault}{\color[rgb]{0,0,0}$(0,2)$}%
}}}}
\put(9001,-3672){\makebox(0,0)[lb]{\smash{{\SetFigFont{11}{13.2}{\familydefault}{\mddefault}{\updefault}{\color[rgb]{0,0,0}$(0,1)$}%
}}}}
\put(9001,-4036){\makebox(0,0)[lb]{\smash{{\SetFigFont{11}{13.2}{\familydefault}{\mddefault}{\updefault}{\color[rgb]{0,0,0}$(0,0)$}%
}}}}
\put(901,-2686){\makebox(0,0)[lb]{\smash{{\SetFigFont{11}{13.2}{\rmdefault}{\mddefault}{\updefault}{\color[rgb]{0,0,0}$d$}%
}}}}
\put(676,-2086){\makebox(0,0)[lb]{\smash{{\SetFigFont{11}{13.2}{\rmdefault}{\mddefault}{\updefault}{\color[rgb]{0,0,0}$g$}%
}}}}
\put(1276,-2086){\makebox(0,0)[lb]{\smash{{\SetFigFont{11}{13.2}{\rmdefault}{\mddefault}{\updefault}{\color[rgb]{0,0,0}$h$}%
}}}}
\put(826,-3361){\makebox(0,0)[lb]{\smash{{\SetFigFont{11}{13.2}{\rmdefault}{\mddefault}{\updefault}{\color[rgb]{0,0,0}$a$}%
}}}}
\put(5701,-2011){\makebox(0,0)[lb]{\smash{{\SetFigFont{11}{13.2}{\familydefault}{\mddefault}{\updefault}{\color[rgb]{0,0,0}$\omega$}%
}}}}
\put(9001,-2011){\makebox(0,0)[lb]{\smash{{\SetFigFont{11}{13.2}{\familydefault}{\mddefault}{\updefault}{\color[rgb]{0,0,0}$(0,\omega)$}%
}}}}
\put(10726,-2011){\makebox(0,0)[lb]{\smash{{\SetFigFont{11}{13.2}{\familydefault}{\mddefault}{\updefault}{\color[rgb]{0,0,0}$(1,\omega)$}%
}}}}
\put(751,-4561){\makebox(0,0)[lb]{\smash{{\SetFigFont{11}{13.2}{\familydefault}{\mddefault}{\updefault}{\color[rgb]{0,0,0}$(i)$ A finite poset}%
}}}}
\put(4051,-4561){\makebox(0,0)[lb]{\smash{{\SetFigFont{11}{13.2}{\familydefault}{\mddefault}{\updefault}{\color[rgb]{0,0,0}$(ii)$ The non-continuous dcpo $\mathcal N_2$}%
}}}}
\put(9376,-4561){\makebox(0,0)[lb]{\smash{{\SetFigFont{11}{13.2}{\familydefault}{\mddefault}{\updefault}{\color[rgb]{0,0,0}$(iii)$ $\nat_\omega + \nat_\omega$}%
}}}}
\end{picture}%

%% file: ex1.pstex_t
\begin{picture}(0,0)%
\includegraphics{ex1.pstex}%
\end{picture}%
\setlength{\unitlength}{2368sp}%
\begingroup\makeatletter\ifx\SetFigFont\undefined%
\gdef\SetFigFont#1#2#3#4#5{%
  \reset@font\fontsize{#1}{#2pt}%
  \fontfamily{#3}\fontseries{#4}\fontshape{#5}%
  \selectfont}%
\fi\endgroup%
\begin{picture}(11076,2847)(480,-4642)
\put(2176,-3361){\makebox(0,0)[lb]{\smash{{\SetFigFont{11}{13.2}{\rmdefault}{\mddefault}{\updefault}{\color[rgb]{0,0,0}$b$}%
}}}}
\put(2701,-3361){\makebox(0,0)[lb]{\smash{{\SetFigFont{11}{13.2}{\rmdefault}{\mddefault}{\updefault}{\color[rgb]{0,0,0}$c$}%
}}}}
\put(2176,-2761){\makebox(0,0)[lb]{\smash{{\SetFigFont{11}{13.2}{\rmdefault}{\mddefault}{\updefault}{\color[rgb]{0,0,0}$e$}%
}}}}
\put(2701,-2761){\makebox(0,0)[lb]{\smash{{\SetFigFont{11}{13.2}{\rmdefault}{\mddefault}{\updefault}{\color[rgb]{0,0,0}$f$}%
}}}}
\put(2401,-2491){\makebox(0,0)[lb]{\smash{{\SetFigFont{11}{13.2}{\rmdefault}{\mddefault}{\updefault}{\color[rgb]{0,0,0}$i$}%
}}}}
\put(2026,-2116){\makebox(0,0)[lb]{\smash{{\SetFigFont{11}{13.2}{\rmdefault}{\mddefault}{\updefault}{\color[rgb]{0,0,0}$j$}%
}}}}
\put(2101,-3886){\makebox(0,0)[lb]{\smash{{\SetFigFont{11}{13.2}{\rmdefault}{\mddefault}{\updefault}{\color[rgb]{0,0,0}$\bot$}%
}}}}
\put(6451,-3307){\makebox(0,0)[lb]{\smash{{\SetFigFont{11}{13.2}{\familydefault}{\mddefault}{\updefault}{\color[rgb]{0,0,0}$(1,2)$}%
}}}}
\put(6451,-3672){\makebox(0,0)[lb]{\smash{{\SetFigFont{11}{13.2}{\familydefault}{\mddefault}{\updefault}{\color[rgb]{0,0,0}$(1,1)$}%
}}}}
\put(6451,-4036){\makebox(0,0)[lb]{\smash{{\SetFigFont{11}{13.2}{\familydefault}{\mddefault}{\updefault}{\color[rgb]{0,0,0}$(1,0)$}%
}}}}
\put(4726,-3307){\makebox(0,0)[lb]{\smash{{\SetFigFont{11}{13.2}{\familydefault}{\mddefault}{\updefault}{\color[rgb]{0,0,0}$(0,2)$}%
}}}}
\put(4726,-3672){\makebox(0,0)[lb]{\smash{{\SetFigFont{11}{13.2}{\familydefault}{\mddefault}{\updefault}{\color[rgb]{0,0,0}$(0,1)$}%
}}}}
\put(4726,-4036){\makebox(0,0)[lb]{\smash{{\SetFigFont{11}{13.2}{\familydefault}{\mddefault}{\updefault}{\color[rgb]{0,0,0}$(0,0)$}%
}}}}
\put(10726,-3307){\makebox(0,0)[lb]{\smash{{\SetFigFont{11}{13.2}{\familydefault}{\mddefault}{\updefault}{\color[rgb]{0,0,0}$(1,2)$}%
}}}}
\put(10726,-3672){\makebox(0,0)[lb]{\smash{{\SetFigFont{11}{13.2}{\familydefault}{\mddefault}{\updefault}{\color[rgb]{0,0,0}$(1,1)$}%
}}}}
\put(10726,-4036){\makebox(0,0)[lb]{\smash{{\SetFigFont{11}{13.2}{\familydefault}{\mddefault}{\updefault}{\color[rgb]{0,0,0}$(1,0)$}%
}}}}
\put(9001,-3307){\makebox(0,0)[lb]{\smash{{\SetFigFont{11}{13.2}{\familydefault}{\mddefault}{\updefault}{\color[rgb]{0,0,0}$(0,2)$}%
}}}}
\put(9001,-3672){\makebox(0,0)[lb]{\smash{{\SetFigFont{11}{13.2}{\familydefault}{\mddefault}{\updefault}{\color[rgb]{0,0,0}$(0,1)$}%
}}}}
\put(9001,-4036){\makebox(0,0)[lb]{\smash{{\SetFigFont{11}{13.2}{\familydefault}{\mddefault}{\updefault}{\color[rgb]{0,0,0}$(0,0)$}%
}}}}
\put(901,-2686){\makebox(0,0)[lb]{\smash{{\SetFigFont{11}{13.2}{\rmdefault}{\mddefault}{\updefault}{\color[rgb]{0,0,0}$d$}%
}}}}
\put(676,-2086){\makebox(0,0)[lb]{\smash{{\SetFigFont{11}{13.2}{\rmdefault}{\mddefault}{\updefault}{\color[rgb]{0,0,0}$g$}%
}}}}
\put(1276,-2086){\makebox(0,0)[lb]{\smash{{\SetFigFont{11}{13.2}{\rmdefault}{\mddefault}{\updefault}{\color[rgb]{0,0,0}$h$}%
}}}}
\put(826,-3361){\makebox(0,0)[lb]{\smash{{\SetFigFont{11}{13.2}{\rmdefault}{\mddefault}{\updefault}{\color[rgb]{0,0,0}$a$}%
}}}}
\put(5701,-2011){\makebox(0,0)[lb]{\smash{{\SetFigFont{11}{13.2}{\familydefault}{\mddefault}{\updefault}{\color[rgb]{0,0,0}$\omega$}%
}}}}
\put(9001,-2011){\makebox(0,0)[lb]{\smash{{\SetFigFont{11}{13.2}{\familydefault}{\mddefault}{\updefault}{\color[rgb]{0,0,0}$(0,\omega)$}%
}}}}
\put(10726,-2011){\makebox(0,0)[lb]{\smash{{\SetFigFont{11}{13.2}{\familydefault}{\mddefault}{\updefault}{\color[rgb]{0,0,0}$(1,\omega)$}%
}}}}
\put(751,-4561){\makebox(0,0)[lb]{\smash{{\SetFigFont{11}{13.2}{\familydefault}{\mddefault}{\updefault}{\color[rgb]{0,0,0}$(i)$ A finite poset}%
}}}}
\put(4051,-4561){\makebox(0,0)[lb]{\smash{{\SetFigFont{11}{13.2}{\familydefault}{\mddefault}{\updefault}{\color[rgb]{0,0,0}$(ii)$ The non-continuous dcpo $\mathcal N_2$}%
}}}}
\put(9376,-4561){\makebox(0,0)[lb]{\smash{{\SetFigFont{11}{13.2}{\familydefault}{\mddefault}{\updefault}{\color[rgb]{0,0,0}$(iii)$ $\nat_\omega + \nat_\omega$}%
}}}}
\end{picture}%

%% file: qretr.pdftex_t
\begin{picture}(0,0)%
\includegraphics{qretr.pdftex}%
\end{picture}%
\setlength{\unitlength}{2368sp}%
\begingroup\makeatletter\ifx\SetFigFont\undefined%
\gdef\SetFigFont#1#2#3#4#5{%
  \reset@font\fontsize{#1}{#2pt}%
  \fontfamily{#3}\fontseries{#4}\fontshape{#5}%
  \selectfont}%
\fi\endgroup%
\begin{picture}(3675,3906)(2476,-5425)
\put(2626,-1711){\makebox(0,0)[lb]{\smash{{\SetFigFont{11}{13.2}{\rmdefault}{\mddefault}{\updefault}{\color[rgb]{0,0,0}$X$}%
}}}}
\put(6151,-2161){\makebox(0,0)[lb]{\smash{{\SetFigFont{11}{13.2}{\rmdefault}{\mddefault}{\updefault}{\color[rgb]{0,0,0}$\qs (y)$}%
}}}}
\put(5776,-4111){\makebox(0,0)[lb]{\smash{{\SetFigFont{11}{13.2}{\rmdefault}{\mddefault}{\updefault}{\color[rgb]{0,0,0}$r^{-1} (\upc y)$}%
}}}}
\put(4426,-3061){\makebox(0,0)[lb]{\smash{{\SetFigFont{11}{13.2}{\rmdefault}{\mddefault}{\updefault}{\color[rgb]{0,0,0}$x$}%
}}}}
\put(4426,-4486){\makebox(0,0)[lb]{\smash{{\SetFigFont{11}{13.2}{\rmdefault}{\mddefault}{\updefault}{\color[rgb]{0,0,0}$r$}%
}}}}
\put(4276,-5161){\makebox(0,0)[lb]{\smash{{\SetFigFont{11}{13.2}{\rmdefault}{\mddefault}{\updefault}{\color[rgb]{0,0,0}$y=r(x)$}%
}}}}
\put(2476,-5161){\makebox(0,0)[lb]{\smash{{\SetFigFont{11}{13.2}{\rmdefault}{\mddefault}{\updefault}{\color[rgb]{0,0,0}$Y$}%
}}}}
\put(2851,-4336){\makebox(0,0)[lb]{\smash{{\SetFigFont{11}{13.2}{\rmdefault}{\mddefault}{\updefault}{\color[rgb]{0,0,0}$\qs$}%
}}}}
\end{picture}%

%% file: qretr.pstex_t
\begin{picture}(0,0)%
\includegraphics{qretr.pstex}%
\end{picture}%
\setlength{\unitlength}{2368sp}%
\begingroup\makeatletter\ifx\SetFigFont\undefined%
\gdef\SetFigFont#1#2#3#4#5{%
  \reset@font\fontsize{#1}{#2pt}%
  \fontfamily{#3}\fontseries{#4}\fontshape{#5}%
  \selectfont}%
\fi\endgroup%
\begin{picture}(3675,3906)(2476,-5425)
\put(2626,-1711){\makebox(0,0)[lb]{\smash{{\SetFigFont{11}{13.2}{\rmdefault}{\mddefault}{\updefault}{\color[rgb]{0,0,0}$X$}%
}}}}
\put(6151,-2161){\makebox(0,0)[lb]{\smash{{\SetFigFont{11}{13.2}{\rmdefault}{\mddefault}{\updefault}{\color[rgb]{0,0,0}$\qs (y)$}%
}}}}
\put(5776,-4111){\makebox(0,0)[lb]{\smash{{\SetFigFont{11}{13.2}{\rmdefault}{\mddefault}{\updefault}{\color[rgb]{0,0,0}$r^{-1} (\upc y)$}%
}}}}
\put(4426,-3061){\makebox(0,0)[lb]{\smash{{\SetFigFont{11}{13.2}{\rmdefault}{\mddefault}{\updefault}{\color[rgb]{0,0,0}$x$}%
}}}}
\put(4426,-4486){\makebox(0,0)[lb]{\smash{{\SetFigFont{11}{13.2}{\rmdefault}{\mddefault}{\updefault}{\color[rgb]{0,0,0}$r$}%
}}}}
\put(4276,-5161){\makebox(0,0)[lb]{\smash{{\SetFigFont{11}{13.2}{\rmdefault}{\mddefault}{\updefault}{\color[rgb]{0,0,0}$y=r(x)$}%
}}}}
\put(2476,-5161){\makebox(0,0)[lb]{\smash{{\SetFigFont{11}{13.2}{\rmdefault}{\mddefault}{\updefault}{\color[rgb]{0,0,0}$Y$}%
}}}}
\put(2851,-4336){\makebox(0,0)[lb]{\smash{{\SetFigFont{11}{13.2}{\rmdefault}{\mddefault}{\updefault}{\color[rgb]{0,0,0}$\qs$}%
}}}}
\end{picture}%

%% file: V-RB.pdftex_t
\begin{picture}(0,0)%
\includegraphics{V-RB.pdftex}%
\end{picture}%
\setlength{\unitlength}{395sp}%
\begingroup\makeatletter\ifx\SetFigFont\undefined%
\gdef\SetFigFont#1#2#3#4#5{%
  \reset@font\fontsize{#1}{#2pt}%
  \fontfamily{#3}\fontseries{#4}\fontshape{#5}%
  \selectfont}%
\fi\endgroup%
\begin{picture}(66082,12572)(-13518,-24507)
\put(-9824,-13111){\makebox(0,0)[lb]{\smash{{\SetFigFont{11}{13.2}{\rmdefault}{\mddefault}{\updefault}{\color[rgb]{0,0,0}$\Val_1^{\frac 1 2} (X)$}%
}}}}
\put(601,-13111){\makebox(0,0)[lb]{\smash{{\SetFigFont{11}{13.2}{\rmdefault}{\mddefault}{\updefault}{\color[rgb]{0,0,0}$\Val_1^{\frac 1 3} (X)$}%
}}}}
\put(12676,-13111){\makebox(0,0)[lb]{\smash{{\SetFigFont{11}{13.2}{\rmdefault}{\mddefault}{\updefault}{\color[rgb]{0,0,0}$\Val_1^{\frac 1 4} (X)$}%
}}}}
\put(27151,-13111){\makebox(0,0)[lb]{\smash{{\SetFigFont{11}{13.2}{\rmdefault}{\mddefault}{\updefault}{\color[rgb]{0,0,0}$\Val_1^{\frac 1 5} (X)$}%
}}}}
\put(42751,-13111){\makebox(0,0)[lb]{\smash{{\SetFigFont{11}{13.2}{\rmdefault}{\mddefault}{\updefault}{\color[rgb]{0,0,0}$\Val_1^{\frac 1 6} (X)$}%
}}}}
\end{picture}%

%% file: V-RB.pstex_t
\begin{picture}(0,0)%
\includegraphics{V-RB.pstex}%
\end{picture}%
\setlength{\unitlength}{395sp}%
\begingroup\makeatletter\ifx\SetFigFont\undefined%
\gdef\SetFigFont#1#2#3#4#5{%
  \reset@font\fontsize{#1}{#2pt}%
  \fontfamily{#3}\fontseries{#4}\fontshape{#5}%
  \selectfont}%
\fi\endgroup%
\begin{picture}(66082,12572)(-13518,-24507)
\put(-9824,-13111){\makebox(0,0)[lb]{\smash{{\SetFigFont{11}{13.2}{\rmdefault}{\mddefault}{\updefault}{\color[rgb]{0,0,0}$\Val_1^{\frac 1 2} (X)$}%
}}}}
\put(601,-13111){\makebox(0,0)[lb]{\smash{{\SetFigFont{11}{13.2}{\rmdefault}{\mddefault}{\updefault}{\color[rgb]{0,0,0}$\Val_1^{\frac 1 3} (X)$}%
}}}}
\put(12676,-13111){\makebox(0,0)[lb]{\smash{{\SetFigFont{11}{13.2}{\rmdefault}{\mddefault}{\updefault}{\color[rgb]{0,0,0}$\Val_1^{\frac 1 4} (X)$}%
}}}}
\put(27151,-13111){\makebox(0,0)[lb]{\smash{{\SetFigFont{11}{13.2}{\rmdefault}{\mddefault}{\updefault}{\color[rgb]{0,0,0}$\Val_1^{\frac 1 5} (X)$}%
}}}}
\put(42751,-13111){\makebox(0,0)[lb]{\smash{{\SetFigFont{11}{13.2}{\rmdefault}{\mddefault}{\updefault}{\color[rgb]{0,0,0}$\Val_1^{\frac 1 6} (X)$}%
}}}}
\end{picture}%

%% file: V-approx.pdftex_t
\begin{picture}(0,0)%
\includegraphics{V-approx.pdftex}%
\end{picture}%
\setlength{\unitlength}{592sp}%
\begingroup\makeatletter\ifx\SetFigFont\undefined%
\gdef\SetFigFont#1#2#3#4#5{%
  \reset@font\fontsize{#1}{#2pt}%
  \fontfamily{#3}\fontseries{#4}\fontshape{#5}%
  \selectfont}%
\fi\endgroup%
\begin{picture}(25550,8201)(-17249,-23068)
\put(-12524,-15811){\makebox(0,0)[lb]{\smash{{\SetFigFont{11}{13.2}{\rmdefault}{\mddefault}{\updefault}{\color[rgb]{0,0,0}$\frac 1 3 \delta_a + \frac 1 3 \delta_b + \frac 1 3 \delta_\top$}%
}}}}
\put(-17249,-18586){\makebox(0,0)[lb]{\smash{{\SetFigFont{11}{13.2}{\rmdefault}{\mddefault}{\updefault}{\color[rgb]{0,0,0}Best discretizations:}%
}}}}
\end{picture}%

%% file: V-approx.pstex_t
\begin{picture}(0,0)%
\includegraphics{V-approx.pstex}%
\end{picture}%
\setlength{\unitlength}{592sp}%
\begingroup\makeatletter\ifx\SetFigFont\undefined%
\gdef\SetFigFont#1#2#3#4#5{%
  \reset@font\fontsize{#1}{#2pt}%
  \fontfamily{#3}\fontseries{#4}\fontshape{#5}%
  \selectfont}%
\fi\endgroup%
\begin{picture}(25550,8201)(-17249,-23068)
\put(-12524,-15811){\makebox(0,0)[lb]{\smash{{\SetFigFont{11}{13.2}{\rmdefault}{\mddefault}{\updefault}{\color[rgb]{0,0,0}$\frac 1 3 \delta_a + \frac 1 3 \delta_b + \frac 1 3 \delta_\top$}%
}}}}
\put(-17249,-18586){\makebox(0,0)[lb]{\smash{{\SetFigFont{11}{13.2}{\rmdefault}{\mddefault}{\updefault}{\color[rgb]{0,0,0}Best discretizations:}%
}}}}
\end{picture}%

%% file: ex1-path.pdftex_t
\begin{picture}(0,0)%
\includegraphics{ex1-path.pdftex}%
\end{picture}%
\setlength{\unitlength}{2368sp}%
\begingroup\makeatletter\ifx\SetFigFont\undefined%
\gdef\SetFigFont#1#2#3#4#5{%
  \reset@font\fontsize{#1}{#2pt}%
  \fontfamily{#3}\fontseries{#4}\fontshape{#5}%
  \selectfont}%
\fi\endgroup%
\begin{picture}(4569,2733)(5551,-4408)
\put(9751,-1891){\makebox(0,0)[lb]{\smash{{\SetFigFont{11}{13.2}{\rmdefault}{\mddefault}{\updefault}{\color[rgb]{0,0,0}$j$}%
}}}}
\put(9751,-2416){\makebox(0,0)[lb]{\smash{{\SetFigFont{11}{13.2}{\rmdefault}{\mddefault}{\updefault}{\color[rgb]{0,0,0}$i$}%
}}}}
\put(9751,-3136){\makebox(0,0)[lb]{\smash{{\SetFigFont{11}{13.2}{\rmdefault}{\mddefault}{\updefault}{\color[rgb]{0,0,0}$f$}%
}}}}
\put(6676,-2461){\makebox(0,0)[lb]{\smash{{\SetFigFont{11}{13.2}{\rmdefault}{\mddefault}{\updefault}{\color[rgb]{0,0,0}$h$}%
}}}}
\put(6526,-3061){\makebox(0,0)[lb]{\smash{{\SetFigFont{11}{13.2}{\rmdefault}{\mddefault}{\updefault}{\color[rgb]{0,0,0}$d$}%
}}}}
\put(6826,-3661){\makebox(0,0)[lb]{\smash{{\SetFigFont{11}{13.2}{\rmdefault}{\mddefault}{\updefault}{\color[rgb]{0,0,0}$a$}%
}}}}
\put(7501,-3136){\makebox(0,0)[lb]{\smash{{\SetFigFont{11}{13.2}{\rmdefault}{\mddefault}{\updefault}{\color[rgb]{0,0,0}$e$}%
}}}}
\put(8326,-3736){\makebox(0,0)[lb]{\smash{{\SetFigFont{11}{13.2}{\rmdefault}{\mddefault}{\updefault}{\color[rgb]{0,0,0}$b$}%
}}}}
\put(9526,-3736){\makebox(0,0)[lb]{\smash{{\SetFigFont{11}{13.2}{\rmdefault}{\mddefault}{\updefault}{\color[rgb]{0,0,0}$c$}%
}}}}
\put(5551,-2461){\makebox(0,0)[lb]{\smash{{\SetFigFont{11}{13.2}{\rmdefault}{\mddefault}{\updefault}{\color[rgb]{0,0,0}$g$}%
}}}}
\put(6451,-4336){\makebox(0,0)[lb]{\smash{{\SetFigFont{11}{13.2}{\rmdefault}{\mddefault}{\updefault}{\color[rgb]{0,0,0}$\bot$}%
}}}}
\end{picture}%

%% file: ex1-path.pstex_t
\begin{picture}(0,0)%
\includegraphics{ex1-path.pstex}%
\end{picture}%
\setlength{\unitlength}{2368sp}%
\begingroup\makeatletter\ifx\SetFigFont\undefined%
\gdef\SetFigFont#1#2#3#4#5{%
  \reset@font\fontsize{#1}{#2pt}%
  \fontfamily{#3}\fontseries{#4}\fontshape{#5}%
  \selectfont}%
\fi\endgroup%
\begin{picture}(4569,2733)(5551,-4408)
\put(9751,-1891){\makebox(0,0)[lb]{\smash{{\SetFigFont{11}{13.2}{\rmdefault}{\mddefault}{\updefault}{\color[rgb]{0,0,0}$j$}%
}}}}
\put(9751,-2416){\makebox(0,0)[lb]{\smash{{\SetFigFont{11}{13.2}{\rmdefault}{\mddefault}{\updefault}{\color[rgb]{0,0,0}$i$}%
}}}}
\put(9751,-3136){\makebox(0,0)[lb]{\smash{{\SetFigFont{11}{13.2}{\rmdefault}{\mddefault}{\updefault}{\color[rgb]{0,0,0}$f$}%
}}}}
\put(6676,-2461){\makebox(0,0)[lb]{\smash{{\SetFigFont{11}{13.2}{\rmdefault}{\mddefault}{\updefault}{\color[rgb]{0,0,0}$h$}%
}}}}
\put(6526,-3061){\makebox(0,0)[lb]{\smash{{\SetFigFont{11}{13.2}{\rmdefault}{\mddefault}{\updefault}{\color[rgb]{0,0,0}$d$}%
}}}}
\put(6826,-3661){\makebox(0,0)[lb]{\smash{{\SetFigFont{11}{13.2}{\rmdefault}{\mddefault}{\updefault}{\color[rgb]{0,0,0}$a$}%
}}}}
\put(7501,-3136){\makebox(0,0)[lb]{\smash{{\SetFigFont{11}{13.2}{\rmdefault}{\mddefault}{\updefault}{\color[rgb]{0,0,0}$e$}%
}}}}
\put(8326,-3736){\makebox(0,0)[lb]{\smash{{\SetFigFont{11}{13.2}{\rmdefault}{\mddefault}{\updefault}{\color[rgb]{0,0,0}$b$}%
}}}}
\put(9526,-3736){\makebox(0,0)[lb]{\smash{{\SetFigFont{11}{13.2}{\rmdefault}{\mddefault}{\updefault}{\color[rgb]{0,0,0}$c$}%
}}}}
\put(5551,-2461){\makebox(0,0)[lb]{\smash{{\SetFigFont{11}{13.2}{\rmdefault}{\mddefault}{\updefault}{\color[rgb]{0,0,0}$g$}%
}}}}
\put(6451,-4336){\makebox(0,0)[lb]{\smash{{\SetFigFont{11}{13.2}{\rmdefault}{\mddefault}{\updefault}{\color[rgb]{0,0,0}$\bot$}%
}}}}
\end{picture}%

%% file: T.pdftex_t
\begin{picture}(0,0)%
\includegraphics{T.pdftex}%
\end{picture}%
\setlength{\unitlength}{2763sp}%
\begingroup\makeatletter\ifx\SetFigFont\undefined%
\gdef\SetFigFont#1#2#3#4#5{%
  \reset@font\fontsize{#1}{#2pt}%
  \fontfamily{#3}\fontseries{#4}\fontshape{#5}%
  \selectfont}%
\fi\endgroup%
\begin{picture}(2180,4338)(2101,-4483)
\put(2101,-3811){\makebox(0,0)[lb]{\smash{{\SetFigFont{12}{14.4}{\rmdefault}{\mddefault}{\updefault}{\color[rgb]{0,0,0}$(0,0)$}%
}}}}
\put(3451,-3811){\makebox(0,0)[lb]{\smash{{\SetFigFont{12}{14.4}{\rmdefault}{\mddefault}{\updefault}{\color[rgb]{0,0,0}$(1,0)$}%
}}}}
\put(2101,-3061){\makebox(0,0)[lb]{\smash{{\SetFigFont{12}{14.4}{\rmdefault}{\mddefault}{\updefault}{\color[rgb]{0,0,0}$(0,1)$}%
}}}}
\put(3451,-3061){\makebox(0,0)[lb]{\smash{{\SetFigFont{12}{14.4}{\rmdefault}{\mddefault}{\updefault}{\color[rgb]{0,0,0}$(1,1)$}%
}}}}
\put(2101,-2311){\makebox(0,0)[lb]{\smash{{\SetFigFont{12}{14.4}{\rmdefault}{\mddefault}{\updefault}{\color[rgb]{0,0,0}$(0,2)$}%
}}}}
\put(3451,-2311){\makebox(0,0)[lb]{\smash{{\SetFigFont{12}{14.4}{\rmdefault}{\mddefault}{\updefault}{\color[rgb]{0,0,0}$(1,2)$}%
}}}}
\put(2101,-1561){\makebox(0,0)[lb]{\smash{{\SetFigFont{12}{14.4}{\rmdefault}{\mddefault}{\updefault}{\color[rgb]{0,0,0}$(0,3)$}%
}}}}
\put(3451,-1561){\makebox(0,0)[lb]{\smash{{\SetFigFont{12}{14.4}{\rmdefault}{\mddefault}{\updefault}{\color[rgb]{0,0,0}$(1,3)$}%
}}}}
\put(3151,-361){\makebox(0,0)[lb]{\smash{{\SetFigFont{12}{14.4}{\rmdefault}{\mddefault}{\updefault}{\color[rgb]{0,0,0}$\top$}%
}}}}
\put(3226,-4411){\makebox(0,0)[lb]{\smash{{\SetFigFont{12}{14.4}{\rmdefault}{\mddefault}{\updefault}{\color[rgb]{0,0,0}$\bot$}%
}}}}
\end{picture}%

%% file: T.pstex_t
\begin{picture}(0,0)%
\includegraphics{T.pstex}%
\end{picture}%
\setlength{\unitlength}{2763sp}%
\begingroup\makeatletter\ifx\SetFigFont\undefined%
\gdef\SetFigFont#1#2#3#4#5{%
  \reset@font\fontsize{#1}{#2pt}%
  \fontfamily{#3}\fontseries{#4}\fontshape{#5}%
  \selectfont}%
\fi\endgroup%
\begin{picture}(2180,4338)(2101,-4483)
\put(2101,-3811){\makebox(0,0)[lb]{\smash{{\SetFigFont{12}{14.4}{\rmdefault}{\mddefault}{\updefault}{\color[rgb]{0,0,0}$(0,0)$}%
}}}}
\put(3451,-3811){\makebox(0,0)[lb]{\smash{{\SetFigFont{12}{14.4}{\rmdefault}{\mddefault}{\updefault}{\color[rgb]{0,0,0}$(1,0)$}%
}}}}
\put(2101,-3061){\makebox(0,0)[lb]{\smash{{\SetFigFont{12}{14.4}{\rmdefault}{\mddefault}{\updefault}{\color[rgb]{0,0,0}$(0,1)$}%
}}}}
\put(3451,-3061){\makebox(0,0)[lb]{\smash{{\SetFigFont{12}{14.4}{\rmdefault}{\mddefault}{\updefault}{\color[rgb]{0,0,0}$(1,1)$}%
}}}}
\put(2101,-2311){\makebox(0,0)[lb]{\smash{{\SetFigFont{12}{14.4}{\rmdefault}{\mddefault}{\updefault}{\color[rgb]{0,0,0}$(0,2)$}%
}}}}
\put(3451,-2311){\makebox(0,0)[lb]{\smash{{\SetFigFont{12}{14.4}{\rmdefault}{\mddefault}{\updefault}{\color[rgb]{0,0,0}$(1,2)$}%
}}}}
\put(2101,-1561){\makebox(0,0)[lb]{\smash{{\SetFigFont{12}{14.4}{\rmdefault}{\mddefault}{\updefault}{\color[rgb]{0,0,0}$(0,3)$}%
}}}}
\put(3451,-1561){\makebox(0,0)[lb]{\smash{{\SetFigFont{12}{14.4}{\rmdefault}{\mddefault}{\updefault}{\color[rgb]{0,0,0}$(1,3)$}%
}}}}
\put(3151,-361){\makebox(0,0)[lb]{\smash{{\SetFigFont{12}{14.4}{\rmdefault}{\mddefault}{\updefault}{\color[rgb]{0,0,0}$\top$}%
}}}}
\put(3226,-4411){\makebox(0,0)[lb]{\smash{{\SetFigFont{12}{14.4}{\rmdefault}{\mddefault}{\updefault}{\color[rgb]{0,0,0}$\bot$}%
}}}}
\end{picture}%

%% file: qrb-lmcs-5.bbl
\newcommand{\etalchar}[1]{$^{#1}$}
\begin{thebibliography}{GHK{\etalchar{+}}03}

\bibitem[AJ94]{AJ:domains}
Samson Abramsky and Achim Jung.
\newblock Domain theory.
\newblock In S.~Abramsky, D.~M. Gabbay, and T.~S.~E. Maibaum, editors, {\em
  Handbook of Logic in Computer Science}, volume~3, pages 1--168. Oxford
  University Press, 1994.

\bibitem[AMJK04]{AMJK:scs:prob}
Mauricio Alvarez-Manilla, Achim Jung, and Klaus Keimel.
\newblock The probabilistic powerdomain for stably compact spaces.
\newblock {\em Theoretical Computer Science}, 328(3):221--244, 2004.

\bibitem[Ban77]{Banaschewski:essn:ext}
Bernhard Banaschewski.
\newblock Essential extensions of {$T_0$}-spaces.
\newblock {\em General Topology and Applications}, 7:233--246, 1977.

\bibitem[Bou69]{Bourbaki:int:IX}
Nicolas Bourbaki.
\newblock {\em Int{\'e}gration}, volume~IX.
\newblock Diffusion CCLS, 1969.

\bibitem[BS09]{BS:twoval}
Ingo Battenfeld and Alex Simpson.
\newblock Two probabilistic powerdomains in topological domain theory.
\newblock Domains IX Workshop, Sussex, UK, September 2009.

\bibitem[BSS06]{BSS:cgdom}
Ingo Battenfeld, Matthias Schr{\"o}der, and Alex Simpson.
\newblock Compactly generated domain theory.
\newblock {\em Math. Struct. Computer Science}, 16:141--161, 2006.
\newblock {K}laus Keimel Festschrift.

\bibitem[BSS07]{BSS:qcb}
Ingo Battenfeld, Matthias Schr{\"o}der, and Alex Simpson.
\newblock A convenient category of domains.
\newblock {\em Electronic Notes in Theoretical Computer Science}, 172:69--99,
  2007.
\newblock Computation, Meaning and Logic, Articles dedicated to G. Plotkin.

\bibitem[Eda95]{Edalat:int}
Abbas Edalat.
\newblock Domain theory and integration.
\newblock {\em Theoretical Computer Science}, 151:163--193, 1995.

\bibitem[GHK{\etalchar{+}}03]{GHKLMS:contlatt}
Gerhard Gierz, Karl~Heinrich Hofmann, Klaus Keimel, Jimmie~D. Lawson, Michael
  Mislove, and Dana~S. Scott.
\newblock Continuous lattices and domains.
\newblock In {\em Encyclopedia of Mathematics and its Applications}, volume~93.
  Cambridge University Press, 2003.

\bibitem[GLS83]{GLS:quasicont}
Gerhard Gierz, Jimmie~D. Lawson, and A.~Stralka.
\newblock Quasicontinuous posets.
\newblock {\em Houston Journal of Mathematics}, 9:191--208, 1983.

\bibitem[Gou07a]{JGL-icalp07}
Jean Goubault{-}Larrecq.
\newblock Continuous capacities on continuous state spaces.
\newblock In Lars Arge, {\relax Ch}ristian Cachin, Tomasz Jurdzi{\'{n}}ski, and
  Andrzej Tarlecki, editors, {\em {P}roceedings of the 34th {I}nternational
  {C}olloquium on {A}utomata, {L}anguages and {P}rogramming ({ICALP}'07)},
  pages 764--776. Springer-Verlag LNCS 4596, July 2007.

\bibitem[Gou07b]{Gou-csl07}
Jean Goubault{-}Larrecq.
\newblock Continuous previsions.
\newblock In Jacques Duparc and Thomas~A. Henzinger, editors, {\em
  {P}roceedings of the 16th {A}nnual {EACSL} {C}onference on {C}omputer
  {S}cience {L}ogic ({CSL}'07)}, pages 542--557. Springer-Verlag LNCS 4646,
  September 2007.

\bibitem[Gou10]{JGL-mscs09}
Jean Goubault{-}Larrecq.
\newblock {D}e~{G}root duality and models of choice: Angels, demons, and
  nature.
\newblock {\em Mathematical Structures in Computer Science}, 20:169--237, 2010.

\bibitem[Gra88]{Graham:rb:V}
S.~Graham.
\newblock Closure properties of a probabilistic powerdomain construction.
\newblock In M.~Main, A.~Melton, M.~Mislove, and D.~Schmidt, editors, {\em 3rd
  MFPS Workshop}, pages 213--233. Springer Verlag LNCS 298, 1988.

\bibitem[Hec96]{Heckmann:space:val}
Reinhold Heckmann.
\newblock Spaces of valuations.
\newblock In S.~Andima, R.~C. Flagg, G.~Itzkowitz, Y.~Kong, R.~Kopperman, and
  P.~Misra, editors, {\em Papers on General Topology and Applications: 11th
  Summer Conference at the University of Southern Maine}, volume 806 of {\em
  Annals of the New York Academy of Sciences}, pages 174--200, New York, USA,
  1996.

\bibitem[Isb75]{Isbell:meetcont}
John~R. Isbell.
\newblock Meet-continuous lattices.
\newblock {\em Symposia Mathematica}, 16:41--54, 1975.
\newblock Convegno sulla Topologica Insiemsistica e Generale, INDAM, Roma,
  Marzo 1973.

\bibitem[JKM01]{DBLP:journals/entcs/JungKM01}
Achim Jung, Mathias Kegelmann, and M.~Andrew Moshier.
\newblock Stably compact spaces and closed relations.
\newblock In {\em Proc. 17th Conference on Mathematical Foundations of
  Programming Semantics, Aarhus, Denmark}. Electronic Notes in Theoretical
  Computer Science, 2001.

\bibitem[JP89]{JP:proba}
Claire Jones and Gordon~D. Plotkin.
\newblock A probabilistic powerdomain of evaluations.
\newblock In {\em Proc. 4th IEEE Symposium on Logics in Computer Science
  (LICS'89)}, pages 186--195. IEEE, 1989.

\bibitem[JT98]{JT:troublesome}
Achim Jung and Regina Tix.
\newblock The troublesome probabilistic powerdomain.
\newblock In A.~Edalat, A.~Jung, K.~Keimel, and M.~Kwiatkowska, editors, {\em
  Proc.\ 3rd Workshop on Computation and Approximation}. Electronic Notes in
  Theoretical Computer Science, 13:70--91, 1998.

\bibitem[Jun88]{Jung:CCC}
Achim Jung.
\newblock {\em Cartesian Closed Categories of Domains}.
\newblock PhD thesis, Technische Hochschule Darmstadt, July 1988.

\bibitem[Jun90]{Jung:CCC:LICS}
Achim Jung.
\newblock The classification of continuous domains.
\newblock In {\em Proceedings of the Fifth Annual IEEE Symposium on Logic in
  Computer Science}, pages 35--40. IEEE Computer Society Press, 1990.

\bibitem[Jun04]{Jung:scs:prob}
Achim Jung.
\newblock Stably compact spaces and the probabilistic powerspace construction.
\newblock In J.~Desharnais and P.~Panangaden, editors, {\em Domain-theoretic
  Methods in Probabilistic Processes}, volume~87 of {\em Electronic Lecture
  Notes in Computer Science}. Elsevier, 2004.
\newblock 15pp.

\bibitem[Kei06]{Keimel:topcones}
Klaus Keimel.
\newblock Topological cones: Foundations for a domain-theoretical semantics
  combining probability and nondeterminism.
\newblock {\em Electronic Notes in Theoretical Computer Science}, 155:423--443,
  2006.

\bibitem[Kir93]{Kirch:bewertung}
Olaf Kirch.
\newblock {\em {B}ereiche und {B}ewertungen}.
\newblock Master's thesis, Technische Hochschule Darmstadt, June 1993.

\bibitem[Law85]{Lawson:T0:pw:conv}
Jimmie~D. Lawson.
\newblock {$T_0$}-spaces and pointwise convergence.
\newblock {\em Topology and its Applications}, 21:73--76, 1985.

\bibitem[Law87]{Lawson:versatile}
Jimmie~D. Lawson.
\newblock The versatile continuous order.
\newblock In Michael~G. Main, Austin Melton, Michael~W. Mislove, and David~A.
  Schmidt, editors, {\em Proc. 3rd MFPS Workshop}, pages 134--160. Springer
  Verlag LNCS 298, 1987.

\bibitem[Mis98]{Mislove:topo:CS}
Michael Mislove.
\newblock Topology, domain theory and theoretical computer science.
\newblock {\em Topology and Its Applications}, 89:3--59, 1998.

\bibitem[Mog91]{Mog91}
Eugenio Moggi.
\newblock Notions of computation and monads.
\newblock {\em Information and Computation}, 93:55--92, 1991.

\bibitem[Nac65]{Nachbin:toporder}
Leopoldo Nachbin.
\newblock {\em Topology and Order}.
\newblock Van Nostrand, Princeton, NJ, 1965.
\newblock Translated from the 1950 monograph ``Topologia e Ordem'' (in
  Portuguese). Reprinted by Robert E. Kreiger Publishing Co. Huntington, NY,
  1967, 1976.

\bibitem[Sch93]{Schalk:PhD}
Andrea Schalk.
\newblock {\em Algebras for Generalized Power Constructions}.
\newblock PhD thesis, Technische Hochschule Darmstadt, 1993.

\bibitem[Smy83]{Smyth:CCC}
Michael~B. Smyth.
\newblock The largest cartesian closed category of domains.
\newblock {\em Theoretical Computer Science}, 27:109--119, 1983.

\bibitem[Tix95]{Tix:bewertung}
Regina Tix.
\newblock {\em {S}tetige {B}ewertungen auf topologischen {R\"a}umen}.
\newblock Diplomarbeit, TH Darmstadt, June 1995.

\end{thebibliography}
